\documentclass[english]{article}
\usepackage{lmodern}
\usepackage[T1]{fontenc}
\usepackage[latin9]{inputenc}
\usepackage{mathrsfs}
\usepackage{amsmath}
\usepackage{amsthm}
\usepackage{amssymb}
\usepackage{stackrel}
\usepackage{graphicx}

\makeatletter
\newenvironment{lyxlist}[1]
	{\begin{list}{}
		{\settowidth{\labelwidth}{#1}
		 \setlength{\leftmargin}{\labelwidth}
		 \addtolength{\leftmargin}{\labelsep}
		 }}
	{\end{list}}
\theoremstyle{plain}
\newtheorem{thm}{\protect\theoremname}
\theoremstyle{plain}
\newtheorem{lem}[thm]{\protect\lemmaname}

\usepackage{amsmath}
\allowdisplaybreaks

\makeatother

\usepackage{babel}
\providecommand{\lemmaname}{Lemma}
\providecommand{\theoremname}{Theorem}

\begin{document}
\title{\textbf{Questionable and Unquestionable in Quantum Mechanics}}
\author{László E. Szabó$^{(1)}$, Márton Gömöri$^{(1)(2)}$ and Zalán Gyenis$^{(3)}$\\
~\\
{\normalsize\emph{$^{(1)}$Department of Logic, Institute of Philosophy,
Eötvös University, Budapest}}\\
{\normalsize\emph{$^{(2)}$Institute of Philosophy, Research Centre
for the Humanities, Budapest}}\\
{\normalsize\emph{$^{(3)}$Department of Logic, Jagiellonian University,
Krakow}}~}
\date{~}
\maketitle
\begin{abstract}
According to the Kolmogorovian Censorship Hypothesis, everything that
quantum theory says about the world in the language of the quantum
mechanical Hilbert space formalism is actually about relationships
between ordinary relative frequencies expressible in operational terms
using classical Kolmogorovian probability theory. In other words,
a quantum theoretical description of a system should in principle
be translatable into a purely operational--probabilistic description.
However, our goal in this paper is different; we do not want to deal
with the problem how to translate the known theory of quantum mechanics
into operational terms, or to reconstruct the theory from postulates
which can be interpreted in operational terms. Our aim is somewhat
broader and points in the opposite direction. We start with a general
scheme for the operational description of an arbitrary physical system.
The description is based solely on the notion of observable events
(measurement operations and measurement results) and on general, empirically
established simple laws concerning their relative frequency. These
laws are so simple and fundamental that they apply equally to any
physical system---no plausibly conceivable physical system is known
that would violate our basic assumptions. In the first part of the
paper, we outline the basic elements of such an operational--probabilistic
theory; such as the notion of state, the mathematical description
of state space, and the basic notions of dynamics. All these notions
are expressed in classical terms, within the framework of Kolmogorovian
probability theory, and, since our goal is not necessarily to reproduce
standard quantum mechanics, we try to avoid making assumptions that
are restrictive and would not hold in the most general case. In the
second part of the paper, we discuss how this operational--probabilistic
description compares to the quantum mechanical description and to
what extent the standard Hilbert space quantum mechanics can be regarded
as a reformulation of the general operational--probabilistic theory.
\end{abstract}

\section{Introduction}

The main motivation of this paper is the so called Kolmogorovian Censorship
Hypothesis, which was first formulated in (Szabó 1995). For long decades
it had been the widely accepted view that quantum mechanics not only
teaches us that the world is essentially of probabilistic nature,
but also that the notion of probability itself, as it is used in quantum
mechanics, is a new one, essentially different from the traditional
Kolmogorovian notion of probability. The main difference, and the
source of many other differences, is that the event algebra in quantum
probability theory is not a Boolean algebra, but a more general algebraic
structure, an orthomodular lattice, isomorphic with the subspace lattice
of a Hilbert space. On the one hand, this insight has helped us to
understand more deeply the curiosities of quantum theory, but on the
other hand, it has raised a number of additional foundational problems.
Nevertheless, it was agreed that classical Kolmogorovian probability
theory was unable to accommodate quantum phenomena and that some form
of quantum probability theory was inevitable.

The Kolmogorovian Censorship Hypothesis challenged this picture. The
hypothesis is a generalization of the observation that whenever quantum
probabilities appear in the description of real physical phenomena,
they appear in combination with classical Kolmogorovian probabilities,
such that, on the surface, the probabilities of all observable, or
at least ontologically occurring, physical events are classical Kolmogorovian
probabilities, which can be interpreted as ordinary relative frequencies.

Take the simplest example. Let $A$ be a quantum observable, a self-adjoint
operator with spectral decomposition $A=\sum_{i}\alpha_{i}P_{i}$
where $\alpha_{i}$ denotes an eigenvalue and $P_{i}$ denotes the
corresponding spectral projector. Let the system be in quantum state
$W$. The most fundamental probabilistic claim of quantum theory is
that the quantum probability of getting value $\alpha_{i}$ in an
$A$-measurement is
\[
q\left(\alpha_{i}\right)=\text{tr\ensuremath{\left(WP_{i}\right)}}
\]
However, comparing this prediction of the theory with the laboratory
observations, what we actually verify is that 
\begin{equation}
p\left(\left[\alpha_{i}\right]\right)=\text{tr}\left(WP_{i}\right)p(a)\label{eq:trace}
\end{equation}
where $p\left(\left[\alpha_{i}\right]\right)$ is the relative frequency
of the outcome event $\left[\alpha_{i}\right]$, say, that the pointer's
position is ``$\alpha_{i}$,'' and $p(a)$ is the relative frequency
of the event $a$ consisting in that an $A$-measurement is performed.
That is, quantum probability $\text{tr}\left(WP_{i}\right)$ never
stands naked, ``facing the tribunal of sense experience,'' but is
surrounded by classical Kolmogorovian probabilities.

While this problem does not arise if we restrict ourselves to the
spectral projectors of a single observable, it is easy to see that
an arbitrary set of quantum probabilities, collectively, cannot constitute
the relative frequencies of events. (Due to violation of Bell--type
inequalities, Pitowsky~1989, Ch.~2. See Szabó~1998; 2001; 2008.)
Don't be misled by the fact that in expressions like (\ref{eq:trace}),
the quantum probability can be interpreted as a classical conditional
probability,
\[
\text{tr}\left(WP_{i}\right)=\frac{p\left(\left[\alpha_{i}\right]\right)}{p(a)}=p\left(\left[\alpha_{i}\right]|a\right)
\]
assuming that $p(a)\neq0$ and $p\left(\left[\alpha_{i}\right]\wedge a\right)=p\left(\left[\alpha_{i}\right]\right)$,
and that the conditional probability can be interpreted as another
probability function on the same event algebra. The different quantum
probabilities can be interpreted as classical conditional probabilities,
but they may belong to different conditioning events, and therefore
do not together form a classical probability function over the event
algebra in question; or over an arbitrary event algebra in general.

This means that there can be no events or any states of affairs in
the ontology of the physical world whose relative frequencies, counted
on the Humean mosaic, are equal to quantum probabilities. According
to the Kolmogorovian Censorship Hypothesis, however, quantum probabilities
can always be expressed in terms of the relative frequencies of such
real events. Typically, as suggested by (\ref{eq:trace}), in terms
of the relative frequencies of the outcomes of measurements and the
relative frequencies of executions of measurements.

It should be noted that the possibility of interpreting quantum probability
as formulated in the Hypothesis can be formally proved under different
special conditions (Bana and Durt 1997; Szabó 2001; Rédei 2010; Hofer-Szabó
\emph{et al}.~2013, Ch.~9).

Thus, according to the Kolmogorovian Censorship Hypothesis, everything
that quantum theory says about the world in the language of the quantum
mechanical Hilbert space formalism is actually about relationships
between ordinary relative frequencies expressible in operational terms
using classical Kolmogorovian probability theory. In other words,
a quantum theoretical description of a system should in principle
be translatable into a purely operational--probabilistic description.

To translate quantum mechanics into operational terms is not a new
idea, of course (e.g. Jauch and Piron~1963; Ludwig~1970; Foulis
and Randall~1974; Davies~1976; Busch, Grabowski, and Lahti~1995;
Spekkens~2005; Barum \emph{et al.}~2007; 2008; Hardy~2008; Aerts~2009;
Abramsky and Heunen~2016; Schmid, Spekkens, and Wolfe~2018). Nevertheless,
what such a translation of standard quantum mechanics looks like is
not a self-evident question. However, we do not wish to discuss this
question here. Our goal is different; we do not want to translate
the known theory of quantum mechanics into operational terms, or to
reconstruct the theory from postulates which can be interpreted in
operational terms. The aim of this paper is somewhat broader and points
in the opposite direction.

We start with a general scheme for the operational description of
an arbitrary physical system. The description is based solely on the
notion of observable events (measurement operations and measurement
results) and on general, empirically established simple laws concerning
their relative frequency. These laws are so simple and fundamental
that they apply equally to any physical system, whether it is traditionally
considered as classical or quantum, or even ``more general than quantum''
(Müller~2021, Sec.~2). In other words, our goal is not necessarily
to reproduce standard quantum mechanics, and therefore we try to avoid
making assumptions that are restrictive and would not hold in the
most general case---no plausibly conceivable physical system is known
that would violate our basic assumptions. In this sense, our operational--probabilistic
model significantly differs from the other similar approaches in the
above mentioned literature; including the most recent GPT approaches
(Hardy~2008; Holevo~2011; Müller~2021). The main differences are
briefly highlighted in several places throughout the paper; on two
key aspects we reflect in more detail in Appendices 2 and 3.

Although our main---technically non-trivial---result concerns how
this operational--probabilistic description compares to the quantum
mechanical description and to what extent the standard Hilbert space
quantum mechanics can be regarded as a reformulation of the general
operational--probabilistic theory, we would like to draw the reader's
attention to the first part of the paper in which we outline the basic
elements of such an operational--probabilistic theory; such as the
notion of state, the mathematical description of state space, and
the basic notions of dynamics. All these notions are expressed in
classical terms, within the framework of Kolmogorovian probability
theory. It remains an intriguing question how to continue this project,
how to apply it, for example, for the description of well-known quantum
mechanical systems with the known symmetries and dynamics, etc.; to
which the authors plan to return somewhere else.

The paper is structured as follows. We describe a typical empirical
scenario in the following way: One can perform different measurement
operations on a physical system, each of which may have different
possible outcomes. The performance of a measuring operation is regarded
as a physical event on par with the measurement outcomes. Empirical
data are, exclusively, the observed relative frequencies of how many
times different measurement operations are performed and how many
times different outcome events occur, including the joint performances
of two or more measurements and the conjunctions of their outcomes.
In terms of the observed relative frequencies we stipulate two empirical
conditions, \textbf{(E1)} and \textbf{(E2)}, which are simple, plausible,
and empirically testable.

Of course, the observed relative frequencies essentially depend on
the frequencies with which the measurement operations are performed;
that is, on circumstances external to the physical system under consideration;
for example, on the free choice of a human. Under a further empirically
testable assumption about the observed frequencies, \textbf{(E3)},
we can isolate a notion which is independent of the relative frequencies
of the measurement operations and can be identified with the system's
own state; in the sense that it characterizes the system's probabilistic
behavior against all possible measurement operations. The largest
part of our further investigation is at the level of generality defined
by assumptions \textbf{(E1)}--\textbf{(E3)}.

In Section~\ref{sec:The-state-of}, we derive important theorems,
solely from conditions \textbf{(E1)}--\textbf{(E3)}, concerning the
possible states of the system. In Section~\ref{sec:Dynamics}, we
characterize the time evolution of these states, first in its most
general form under conditions \textbf{(E1)}--\textbf{(E3)} alone,
then on the basis of a further, empirically testable assumption \textbf{(E4)}.
Section~\ref{sec:Hidden} considers various possible ontological
pictures consistent with our probabilistic notion of state.

All these investigations are expressed in terms of relative frequencies,
which by definition satisfy the Kolmogorovian axioms of classical
probability theory. This means that any physical system---traditionally
categorized as classical or quantum, or ``more general than quantum''---that
can be described in operational terms can be described within classical
Kolmogorovian probability theory; including the system's state, time
evolution or ontology. In the second part of the paper, at the same
time, we will show that anything that can be described in these operational
terms can, if we wish, be represented in the Hilbert space quantum
mechanical formalism. It will be proved that there always exists:
\begin{itemize}
\item a suitable Hilbert space, such that
\item the outcomes of each measurement can be represented by a system of
pairwise orthogonal closed subspaces, spanning the whole Hilbert space,
\item the states of the system can be represented by pure state operators
with suitable state vectors, and
\item the probabilities of the measurement outcomes can be reproduced by
the usual trace formula of quantum mechanics.
\end{itemize}
Moreover, if appropriate, one can label the possible outcomes of a
measurement with numbers, and talk about them as the measured values
of a physical quantity. Each such quantity
\begin{itemize}
\item can be associated with a suitable self-adjoint operator, such that
\item the expectation value of the quantity, in all states of the system,
can be reproduced by the usual trace formula applied to the associated
self-adjoint operator,
\item the possible measurement results are exactly the eigenvalues of the
operator, and
\item the corresponding outcome events are represented by the eigenspaces
pertaining to the eigenvalues respectively, according to the spectral
decomposition of the operator in question.
\end{itemize}
This suggests that the basic postulates of quantum theory are in fact
analytic statements: they do not tell us anything about a physical
system beyond the fact that the system can be described in operational
terms. This is almost true. Nevertheless, it must be mentioned that
the quantum-mechanics-like representation we will obtain is not completely
identical with standard quantum mechanics. The interesting fact is
that most of the deviations from the quantum mechanical folklore,
discussed in Section~\ref{sec:Questionable-and-Unquestionable},
are related to exactly those issues in the foundations of quantum
mechanics that have been hotly debated for long decades.

\section{The General Operational Schema}

Consider a general experimental scenario: we can perform different
measurement operations denoted by $a_{1},a_{2},\ldots a_{r},\ldots a_{m}$
on a physical system. We shall use the same notation $a_{r}$ for
the physical event that the measurement operation $a_{r}$ happened.
Each measurement $a_{r}$ may have different outcomes denoted by $X_{1}^{r},X_{2}^{r},\ldots X_{n_{r}}^{r}$.
Let $M=\sum_{r=1}^{m}n_{r}$, and let $I^{M}$ denote the following
set of indices:
\begin{eqnarray}
I^{M} & = & \left\{ _{i}^{r}\,\left|\,1\leq r\leq m,1\leq i\leq n_{r}\right.\right\} \label{eq:jeloles-I^M}
\end{eqnarray}
Sometimes we perform two or more measurement operations simultaneously---that
is, in the same run of the experiment. So we also consider the double,
triple, and higher conjunctions of measurement operations and the
possible outcome events. In general, we consider the free Boolean
algebra $\mathcal{A}$ generated by the set of all measurement operation
and measurement outcome events 
\begin{eqnarray}
G & = & \left\{ a_{r}\right\} _{r=1,2,\ldots m}\cup\left\{ X_{i}^{r}\right\} _{_{i}^{r}\in I^{M}}\label{eq:generator}
\end{eqnarray}
with the usual Boolean operations, denoted by $\wedge$, $\vee$ and
$\neg$. Introduce the following concise notation: let $S_{max}^{M}$
denote the set of the indices of all double, triple, and higher conjunctions
of the outcome events in $G$. That is, for example, $_{i_{1}i_{2}\ldots i_{L}}^{r_{1}r_{2}\ldots r_{L}}\in S_{max}^{M}$
will stand for the conjunction $X_{i_{1}}^{r_{1}}\wedge X_{i_{2}}^{r_{2}}\ldots\wedge X_{i_{L}}^{r_{L}}$,
etc.

The event algebra $\mathcal{A}$ has $2^{M+m}$ atoms, each having
the form of
\begin{eqnarray}
\varDelta_{\vec{\varepsilon},\vec{\eta}} & = & \left(\underset{_{i}^{r}\in I^{M}}{\wedge}\left[X_{i}^{r}\right]^{\varepsilon_{i}^{r}}\right)\wedge\left(\stackrel[s=1]{m}{\wedge}\left[a_{s}\right]^{\eta_{s}}\right)\label{eq:atoms-0}
\end{eqnarray}
where $\vec{\varepsilon}=\left(\varepsilon_{i}^{r}\right)\in\left\{ 0,1\right\} ^{M}$,
$\vec{\eta}=\left(\eta_{s}\right)\in\left\{ 0,1\right\} ^{m}$, and
\begin{eqnarray*}
\left[X_{i}^{r}\right]^{\varepsilon_{i}^{r}} & = & \begin{cases}
X_{i}^{r} & \text{if }\varepsilon_{i}^{r}=1\\
\neg X_{i}^{r} & \text{if }\varepsilon_{i}^{r}=0
\end{cases}\\
\left[a_{s}\right]^{\eta_{s}} & = & \begin{cases}
a_{s} & \text{if }\eta_{s}=1\\
\neg a_{s} & \text{if }\eta_{s}=0
\end{cases}
\end{eqnarray*}
And, of course, all events in algebra $\mathcal{A}$ can be uniquely
expressed as a disjunction of atoms.

Assume that we can repeat the same experimental situation as many
times as needed; that is, we can prepare the same (or identical) physical
system in the same way and we can repeat the same measuring operations
with the same (or identical) measuring devices, etc. In every run
of the experiment we observe which measurement operations are performed
and which outcome events occur, including the joint performances of
two or more measurements and the conjunctions of their outcomes. In
this way, we observe the \emph{relative frequencies} of all elements
of the event algebra $\mathcal{A}$. Let $p$ denote this relative
frequency function on $\mathcal{A}$. Obviously, $\left(\mathcal{A},p\right)$
constitutes a classical probability model satisfying the Kolmogorovian
axioms. Since the relative frequencies on the whole event algebra
are uniquely determined by the relative frequencies of the atoms,
$p$ can be uniquely given by
\begin{eqnarray}
p(\varDelta_{\vec{\varepsilon},\vec{\eta}}) & \,\,\,\,\,\, & \vec{\varepsilon}\in\left\{ 0,1\right\} ^{M};\,\vec{\eta}\in\left\{ 0,1\right\} ^{m}\label{eq:0a}
\end{eqnarray}
The observed relative frequencies on $\mathcal{A}$ are considered
\emph{the} empirical data, exclusively.

We do not make \emph{a priori} assumptions about these relative frequencies.
Any truth about them will be regarded as empirical fact observed in
the experiment. For example, we do not assume that the stipulated
set of measurements $a_{1},a_{2},\ldots a_{r},\ldots a_{m}$, or a
subset of them, is sufficient to ``fully characterize the system's
state'' (in the GPT terminology: \textquotedblleft fiducial\textquotedblright{}
measurements, Müller~2021, p.~23). The reason is we have no operationally
meaningful notion of ``state'' prior to setting up a collection
of measurements (such a notion will be defined in the next section).
If there is any redundancy in $a_{1},a_{2},\ldots a_{r},\ldots a_{m}$,
this will be reflected in the observed relative frequencies on $\mathcal{A}$.

Another example is the fact that two or more measurements $a_{r_{1}},a_{r_{2}},\ldots a_{r_{L}}$
cannot be performed simultaneously; which reveals in the observed
fact that $p\left(a_{r_{1}}\wedge a_{r_{2}}\ldots\wedge a_{r_{L}}\right)$
\emph{always} equals $0$. Though, this ``always'' needs some further
explanation. For, it is obviously true that the frequencies $p(a_{r})$
sensitively depend on the will of the experimenter. Therefore, it
can be the case that $p\left(a_{r_{1}}\wedge a_{r_{2}}\ldots\wedge a_{r_{L}}\right)=0$
simply because the experimenter never chooses to perform the measurements
$a_{r_{1}},a_{r_{2}},\ldots a_{r_{L}}$ simultaneously. At least at
first sight this seems to significantly differ from the situation
when a certain combination of experiments are never performed due
to objective reasons; because the simultaneous performance of the
measurement operations is---as we usually express---impossible.
Without entering into the metaphysical disputes about possibility--impossibility,
we only say that the impossibility of a combination of measurements
is a contingent fact of the world; the measuring devices and the measuring
operations are such that the joint measurement $a_{r_{1}}\wedge a_{r_{2}}\ldots\wedge a_{r_{L}}$
never occurs. Let us denote by $\mathfrak{I}\subset\mathcal{P}\left(\left\{ 1,2,\ldots m\right\} \right)$
(where $\mathcal{P}\left(A\right)$ is the power set of set $A$)
the set of indices of such ``impossible'' conjunctions. That is,
for all $2\leq L\leq m$, 
\begin{eqnarray}
p\left(a_{r_{1}}\wedge a_{r_{2}}\ldots\wedge a_{r_{L}}\right) & = & 0\,\,\,\,\,\,\,\,\,\,\,\,\,\text{if }\left\{ r_{1},r_{2},\ldots r_{L}\right\} \in\mathfrak{\mathrm{\mathfrak{I}}}\label{eq:nem-a}
\end{eqnarray}
In contrast, let $\mathfrak{P}\subset\mathcal{P}\left(\left\{ 1,2,\ldots m\right\} \right)$
denote the set of indices of the ``possible'' conjunctions:
\begin{eqnarray*}
\mathfrak{P} & = & \Bigl\{\left\{ r_{1},r_{2},\ldots r_{L}\right\} \in\mathcal{P}\left(\left\{ 1,2,\ldots m\right\} \right)\Bigl|\,2\leq L\leq m;\left\{ r_{1},r_{2},\ldots r_{L}\right\} \notin\mathfrak{I}\Bigr\}
\end{eqnarray*}

\begin{lyxlist}{0.0.0}
\item [{\textbf{(E1)}}] We assume, as empirically observed fact, that every
conjunction of measurements that is possible does occur with some
non-zero frequency:
\begin{eqnarray}
p\left(a_{r_{1}}\wedge a_{r_{2}}\ldots\wedge a_{r_{L}}\right) & > & 0\,\,\,\,\,\,\,\,\,\,\text{if }\left\{ r_{1},r_{2},\ldots r_{L}\right\} \in\mathfrak{P}\label{eq:nem-b}
\end{eqnarray}
We also assume that for all $1\leq r\leq m$,
\begin{eqnarray}
p\left(a_{r}\right) & > & 0\label{eq:nem-c}
\end{eqnarray}
\end{lyxlist}
Similarly to (\ref{eq:jeloles-I^M}), we introduce the following sets
of indices:
\begin{eqnarray*}
S & = & \left\{ _{i_{1}i_{2}\ldots i_{L}}^{r_{1}r_{2}\ldots r_{L}}\in S_{max}^{M}\,|\,\left\{ r_{1},r_{2},\ldots r_{L}\right\} \in\mathfrak{P}\right\} \\
S_{\mathfrak{I}} & = & \left\{ _{i_{1}i_{2}\ldots i_{L}}^{r_{1}r_{2}\ldots r_{L}}\in S_{max}^{M}\,|\,\left\{ r_{1},r_{2},\ldots r_{L}\right\} \in\mathfrak{I}\right\} 
\end{eqnarray*}

\begin{lyxlist}{0.0.0}
\item [{\textbf{(E2)}}] The following assumptions are also regarded as
empirically observed regularities: for all $_{i}^{r},{}_{i'}^{r'}\in I^{M}$
and $\left\{ r_{1},r_{2},\ldots r_{L}\right\} \in\mathfrak{P}$,

\begin{eqnarray}
p\left(a_{r}\wedge X_{i}^{r}\right) & = & p\left(X_{i}^{r}\right)\label{eq:1a}\\
\text{if }r=r'\text{ and }i\not=i'\text{ then }p\left(X_{i}^{r}\wedge X_{i'}^{r'}\right) & = & 0\label{eq:1b}\\
\sum_{\begin{array}{c}
k\\
\left(_{k}^{r}\in I^{M}\right)
\end{array}}p\left(X_{k}^{r}|a_{r}\right) & = & 1\label{eq:osszeg=00003D1}\\
\sum_{\begin{array}{c}
k_{1}\ldots k_{L}\\
\left(_{k_{1}\ldots k_{L}}^{r_{1}\ldots r_{L}}\in S\right)
\end{array}}p\left(X_{k_{1}}^{r_{1}}\wedge\ldots\wedge X_{k_{L}}^{r_{L}}|a_{r_{1}}\wedge\ldots\wedge a_{r_{L}}\right) & = & 1\label{eq:1d}
\end{eqnarray}
where $p\left(\,|\,\right)$ denotes the usual conditional relative
frequency defined by the Bayes rule---$p\left(a_{r}\right)\neq0$
and $p\left(a_{r_{1}}\wedge a_{r_{2}}\ldots\wedge a_{r_{L}}\right)\neq0$,
due to (\ref{eq:nem-b})--(\ref{eq:nem-c}). That is to say, an outcome
event does not occur without the performance of the corresponding
measurement operation; it is never the case that two different outcomes
of the same measurement occur simultaneously; whenever a measurement
operation is performed, one of the possible outcomes occurs; whenever
a conjunction of measurement operations is performed, one of the possible
outcome combinations occurs.
\end{lyxlist}
In the picture we suggest, an outcome of a measurement is, primarily,
a physical event, an occurrence of a certain state of affairs at the
end of the measuring process; rather than obtaining a numeric value
of a quantity. To give an example, the state of affairs when the rotated
coil of a voltmeter takes a new position of equilibrium with the distorted
spring is ontologically prior to the number on the scale to which
its pointer points at that moment. Nevertheless, in some cases the
measurement outcomes are labeled by real numbers that are interpreted
as the ``measured value'' of a real-valued physical quantity: 
\begin{equation}
\alpha_{r}:X_{i}^{r}\mapsto\alpha_{i}^{r}\in\mathbb{R}\label{eq:labeling-1}
\end{equation}
In this case, at least formally, it may make sense to talk about conditional
expectation value, that is the average of the measured values, given
that the measurement is performed:
\begin{eqnarray*}
\left\langle \alpha_{r}\right\rangle  & = & \sum_{i=1}^{n_{r}}\alpha_{i}^{r}p\left(X_{i}^{r}|a_{r}\right)
\end{eqnarray*}
About all labelings $\alpha_{r}$ we will assume that $\alpha_{i}^{r}\neq\alpha_{j}^{r}$
for $i\neq j$.

\section{The State of the System \protect\label{sec:The-state-of}}

Of course, the relative frequency $p$ in $\left(\mathcal{A},p\right)$
depends not only on the behavior of the physical system after a certain
physical preparation but also on the autonomous decisions of the experimenter
to perform this or that measurement operation. One can hope a scientific
description of the system only if the two things can be separated.
Whether this is possible is a contingent fact of the empirically observed
reality, reflected in the observed relative frequencies.

Let $\left|S\right|$ denote the number of elements of $S$. Consider
the following vector:
\begin{eqnarray}
\vec{Z} & = & \left(Z_{i}^{r},Z_{i_{1}\ldots i_{L}}^{r_{1}\ldots r_{L}}\right)\in\mathbb{R}^{M+\left|S\right|}\label{eq:state}
\end{eqnarray}
where 
\begin{equation}
Z_{i}^{r}=p\left(X_{i}^{r}|a_{r}\right)\,\,\,\,\,\,\,\,\,\,\,\,\,\,{}_{i}^{r}\in I^{M}\label{eq:conditional00}
\end{equation}
and
\begin{equation}
Z_{i_{1}\ldots i_{L}}^{r_{1}\ldots r_{L}}=p\left(X_{i_{1}}^{r_{1}}\wedge\ldots\wedge X_{i_{L}}^{r_{L}}|a_{r_{1}}\wedge\ldots\wedge a_{r_{L}}\right)\,\,\,\,\,\,{}_{i_{1}\ldots i_{L}}^{r_{1}\ldots r_{L}}\in S\label{eq:conditional11}
\end{equation}
In general, even if the physical preparation of the system is identical
in every run of the experiment, the conditional relative frequencies
on the right hand sides of (\ref{eq:conditional00})--(\ref{eq:conditional11}),
hence the values of $Z_{i}^{r}$ and $Z_{i_{1}\ldots i_{L}}^{r_{1}\ldots r_{L}}$,
may vary if the actual frequencies $\left\{ p\left(a_{r}\right)\right\} _{1\leq r\leq m}$
and $\left\{ p\left(a_{r_{1}}\wedge\ldots\wedge a_{r_{L}}\right)\right\} _{\left\{ r_{1},\ldots r_{L}\right\} \in\mathfrak{P}}$
vary, for example, upon the experimenter's decisions.

However, we make the following stipulation as observed empirical fact:
\begin{lyxlist}{0.0.0}
\item [{\textbf{(E3)}}] For all physical preparations, keeping the preparation
fixed, $\vec{Z}$ is independent of the actual nonzero values of $\left\{ p\left(a_{r}\right)\right\} _{1\leq r\leq m}$
and $\left\{ p\left(a_{r_{1}}\wedge\ldots\wedge a_{r_{L}}\right)\right\} _{\left\{ r_{1},\ldots r_{L}\right\} \in\mathfrak{P}}$.
\end{lyxlist}
In other words, what \textbf{(E3)} says is that for all fixed physical
preparations, the observed relative frequencies are such that
\begin{align}
p\left(X_{i}^{r}\right) & =Z_{i}^{r}p\left(a_{r}\right)\,\,\,\,\,\,\,\,\,\,\,{}_{i}^{r}\in I^{M}\label{eq:pi_is_such-1}\\
p\left(X_{i_{1}}^{r_{1}}\wedge\ldots\wedge X_{i_{L}}^{r_{L}}\right) & =\begin{cases}
Z_{i_{1}\ldots i_{L}}^{r_{1}\ldots r_{L}}p\left(a_{r_{1}}\wedge\ldots\wedge a_{r_{L}}\right) & _{i_{1}\ldots i_{L}}^{r_{1}\ldots r_{L}}\in S\\
0 & _{i_{1}\ldots i_{L}}^{r_{1}\ldots r_{L}}\in S_{\mathfrak{I}}
\end{cases}\label{eq:pi_is_such-2}
\end{align}
 with \emph{one and the same} $\vec{Z}=\left(Z_{i}^{r},Z_{i_{1}\ldots i_{L}}^{r_{1}\ldots r_{L}}\right)$.
$\vec{Z}$ is therefore determined only by the physical preparation.
Notice that \textbf{(E3)} does not exclude that different physical
preparations lead to the same $\vec{Z}$.

$\vec{Z}$ can be regarded as a characterization of the system's \emph{state}
right after the given physical preparation, in the sense that it characterizes
the system's probabilistic behavior against all possible measurement
operations. This characterization is complete in the following sense:
\begin{thm}
\label{State--together}State $\vec{Z}$ together with arbitrary relative
frequencies of measurements, $\left\{ p\left(a_{r}\right)\right\} _{1\leq r\leq m}$
and $\left\{ p\left(a_{r_{1}}\wedge\ldots\wedge a_{r_{L}}\right)\right\} _{\left\{ r_{1},\ldots r_{L}\right\} \in\mathfrak{P}}$,
uniquely determine the relative frequency function $p$ on the whole
event algebra $\mathcal{A}$.
\end{thm}

\begin{proof}
Using the same notations we introduced in (\ref{eq:atoms-0}), each
atom has the form of
\begin{equation}
\varDelta_{\vec{\varepsilon},\vec{\eta}}=\underbrace{\left(\underset{_{i}^{r}\in I^{M}}{\wedge}\left[X_{i}^{r}\right]^{\varepsilon_{i}^{r}}\right)}_{\varGamma_{\vec{\varepsilon}}}\wedge\left(\stackrel[s=1]{m}{\wedge}\left[a_{s}\right]^{\eta_{s}}\right)\label{eq:atoms}
\end{equation}
Notice that the part $\varGamma_{\vec{\varepsilon}}=\underset{_{i}^{r}\in I^{M}}{\wedge}\left[X_{i}^{r}\right]^{\varepsilon_{i}^{r}}$
in (\ref{eq:atoms}) uniquely determines the whole $\varDelta_{\vec{\varepsilon},\vec{\eta}}$,
whenever $p\left(\varDelta_{\vec{\varepsilon},\vec{\eta}}\right)\neq0$.
Namely, due to (\ref{eq:1a}) and (\ref{eq:osszeg=00003D1}),
\begin{equation}
p\left(\varDelta_{\vec{\varepsilon},\vec{\eta}}\right)\neq0\text{ implies that for all \ensuremath{1\leq r\leq m}}\text{, }\sum_{i=1}^{n_{r}}\varepsilon_{i}^{r}=0\text{ iff }\eta_{r}=0\label{eq:X->egesz}
\end{equation}
In other words, for each $\vec{\varepsilon}\in\left\{ 0,1\right\} ^{M}$
there is exactly one $\vec{\eta}\in\left\{ 0,1\right\} ^{m}$ for
which (\ref{eq:1a}) and (\ref{eq:osszeg=00003D1}) do not imply that
$p\left(\varDelta_{\vec{\varepsilon},\vec{\eta}}\right)=0$.  Let
us denote it by $\vec{\eta}\left(\vec{\varepsilon}\right)$; and,
for the sake of brevity, introduce the following notation: $\delta_{\vec{\varepsilon}}=p\left(\varDelta_{\vec{\varepsilon},\vec{\eta}\left(\vec{\varepsilon}\right)}\right)$.
(It is not necessarily the case that $\delta_{\vec{\varepsilon}}\neq0$.
For example, the empirical fact (\ref{eq:1b}) will be accounted for
in terms of the values on the right hand side of (\ref{eq:rendszer-3})
below.)

It must be also noticed that $\left\{ \varGamma_{\vec{\varepsilon}}\right\} _{\vec{\varepsilon}\in\left\{ 0,1\right\} ^{M}}$
constitute the atoms of the free Boolean algebra $\mathcal{A}^{M}$
generated by the set $\left\{ X_{i}^{r}\right\} _{_{i}^{r}\in I^{M}}$.
Events $X_{i}^{r}$ and $X_{i_{1}}^{r_{1}}\wedge\ldots\wedge X_{i_{L}}^{r_{L}}$
on the right hand sides of (\ref{eq:conditional00})--(\ref{eq:conditional11})
are elements of $\mathcal{A}^{M}$, and have therefore a unique decomposition
into disjunction of atoms of $\mathcal{A}^{M}$. Accordingly, taking
into account (\ref{eq:X->egesz}), we have
\begin{alignat}{3}
\sum_{\vec{\varepsilon}\in\left\{ 0,1\right\} ^{M}}\delta_{\vec{\varepsilon}} & \,\,=\,\, &  & 1\label{eq:rendszer-1}\\
\sum_{\vec{\varepsilon}\in\left\{ 0,1\right\} ^{M}}{}_{i}^{r}R_{\vec{\varepsilon}}\,\,\delta_{\vec{\varepsilon}} & \,\,=\,\, &  & p\left(X_{i}^{r}\right)=Z_{i}^{r}p\left(a_{r}\right) & \,\,\,\,\,\,\,\, & _{i}^{r}\in I^{M}\\
\sum_{\vec{\varepsilon}\in\left\{ 0,1\right\} ^{M}}{}_{i_{1}\ldots i_{L}}^{r_{1}\ldots r_{L}}R_{\vec{\varepsilon}}\,\,\delta_{\vec{\varepsilon}}\, & \,=\,\, &  & p\left(X_{i_{1}}^{r_{1}}\wedge\ldots\wedge X_{i_{L}}^{r_{L}}\right)\nonumber \\
 &  &  & =Z_{i_{1}\ldots i_{L}}^{r_{1}\ldots r_{L}}p\left(a_{r_{1}}\wedge\ldots\wedge a_{r_{L}}\right) &  & _{i_{1}\ldots i_{L}}^{r_{1}\ldots r_{L}}\in S\\
\sum_{\vec{\varepsilon}\in\left\{ 0,1\right\} ^{M}}{}_{i_{1}\ldots i_{L}}^{r_{1}\ldots r_{L}}R_{\vec{\varepsilon}}\,\,\delta_{\vec{\varepsilon}} & \,\,=\,\, &  & p\left(X_{i_{1}}^{r_{1}}\wedge\ldots\wedge X_{i_{L}}^{r_{L}}\right)=0 &  & _{i_{1}\ldots i_{L}}^{r_{1}\ldots r_{L}}\in S_{\mathfrak{I}}\label{eq:rendszer-3}
\end{alignat}
with
\begin{eqnarray*}
_{i}^{r}R_{\vec{\varepsilon}} & = & \begin{cases}
1 & \text{if }\varGamma_{\vec{\varepsilon}}\subseteq X_{i}^{r}\\
0 & \text{if }\varGamma_{\vec{\varepsilon}}\nsubseteq X_{i}^{r}
\end{cases}\\
_{i_{1}\ldots i_{L}}^{r_{1}\ldots r_{L}}R_{\vec{\varepsilon}} & = & \begin{cases}
1 & \text{if }\varGamma_{\vec{\varepsilon}}\subseteq X_{i_{1}}^{r_{1}}\wedge\ldots\wedge X_{i_{L}}^{r_{L}}\\
0 & \text{if }\varGamma_{\vec{\varepsilon}}\nsubseteq X_{i_{1}}^{r_{1}}\wedge\ldots\wedge X_{i_{L}}^{r_{L}}
\end{cases}
\end{eqnarray*}
where $\subseteq$ is meant in the sense of the partial ordering in
$\mathcal{A}^{M}$.

Now, (\ref{eq:rendszer-1})--(\ref{eq:rendszer-3}) constitute a
system of $1+M+\left|S_{max}^{M}\right|=2^{M}$ linear equations with
$2^{M}$ unknowns $\delta_{\vec{\varepsilon}},\vec{\varepsilon}\in\left\{ 0,1\right\} ^{M}$.
The equations are linearly independent due to the uniqueness of decomposition
into disjunction of atoms of $\mathcal{A}^{M}$, and due to the fact
that there are only conjunctions on the right hand side. (A similar
equation for, say, $X_{i_{1}}^{r_{1}}\vee X_{i_{2}}^{r_{2}}$ could
be expressed as the sum of equations for $X_{i_{1}}^{r_{1}}$ and
$X_{i_{2}}^{r_{2}}$ minus the one for $X_{i_{1}}^{r_{1}}\wedge X_{i_{2}}^{r_{2}}$.)
Therefore, the system has a unique solution for all $\delta_{\vec{\varepsilon}}$,
that is, for the relative frequencies of $\left\{ \varDelta_{\vec{\varepsilon},\vec{\eta}\left(\vec{\varepsilon}\right)}\right\} _{\vec{\varepsilon}\in\left\{ 0,1\right\} ^{M}}$.
The rest of the atoms of $\mathcal{A}$ have zero relative frequency.
\end{proof}
Thus the notion of state we introduced aligns with the traditional
notion of state in a probabilistic and operational context. In fact,
it corresponds precisely to Lucien Hardy's formulation, which has
been widely used in recent GPT-like approaches:
\begin{quote}
The state associated with a particular preparation is defined to be
(that thing represented by) any mathematical object that can be used
to determine the probability associated with the outcomes of any measurement
that may be performed on a system prepared by the given preparation.
(2008, p.~2)
\end{quote}
To avoid misunderstandings, however, it is worthwhile pointing out
two important differences. On the one hand, in our definition the
state determines not only the probability associated with the outcomes
of any measurement but also the probability associated with the outcomes
of any \emph{conjunction} of measurements that can be jointly performed
on the system. By contrast, in the GPT approach measurement conjunctions
are only introduced in connection with ``composite systems,'' in
a way that one measurement is performed on one ``subsystem'' and
simultaneously another measurement on a different ``subsystem''
(Müller~2021, pp.~24-28). In this paper we do not want to talk about
``composite systems'' and ``subsystems,'' and just note that any
operationally meaningful conception of those will be based on the
notion of an overall state which comprises the statistics of all measurements
performable on the total system, including all possible measurement
\emph{conjunctions}.

On the other hand, it must be clear that the state, in itself, does
not determine the probabilities of the measurement outcome events;
only the state of the system $\vec{Z}$ and the relative frequencies
of the measurements $\left\{ p\left(a_{r}\right)\right\} _{1\leq r\leq m}$
and $\left\{ p\left(a_{r_{1}}\wedge\ldots\wedge a_{r_{L}}\right)\right\} _{\left\{ r_{1},\ldots r_{L}\right\} \in\mathfrak{P}}$
together. And the fact that the frequencies of the measurements in
(\ref{eq:pi_is_such-1})--(\ref{eq:pi_is_such-2}) can be arbitrary
does not imply that the components of $\vec{Z}$ 
\[
\left\{ Z_{i}^{r},Z_{i_{1}\ldots i_{L}}^{r_{1}\ldots r_{L}}\right\} _{\begin{array}{l}
_{i}^{r}\in I^{M};\,{}_{i_{1}\ldots i_{L}}^{r_{1}\ldots r_{L}}\in S\end{array}}
\]
constitute relative frequencies of the corresponding outcome events
\[
\left\{ X_{i}^{r},X_{i_{1}}^{r_{1}}\wedge\ldots\wedge X_{i_{L}}^{r_{L}}\right\} _{\begin{array}{l}
_{i}^{r}\in I^{M};\,{}_{i_{1}\ldots i_{L}}^{r_{1}\ldots r_{L}}\in S\end{array}}
\]
 (or events whatsoever), as will be shown in Section~\ref{sec:Hidden}.

It is essential in our present analysis that the measurement operations
are treated on par with the outcome events; they belong to the ontology.
However, as it is clearly seen from (\ref{eq:rendszer-1})--(\ref{eq:rendszer-3}),
the notion of $\vec{Z}$ detaches the ``system's contribution''
to the totality of statistical facts observed in the measurements
from the ``experimenter's contribution''.

Still, the state of the system depends not only on the features intrinsic
to the system in itself, but also on the content of $\mathfrak{I}$,
i.e., which combinations of measuring operations cannot be performed
simultaneously. This means that the measuring devices and measuring
operations, by means of which we establish the empirically meaningful
semantics of our physical description of the system, play a \emph{constitutive}
role in the notion of state\emph{ attributed to the system}. This
kind of constitutive role of the semantic conventions is however completely
natural in all empirically meaningful physical theories (Szabó~2020).

The following lemma will be important for our further investigations:
\begin{lem}
\label{Lemma:ineq}For all states,
\begin{eqnarray}
Z_{i_{1}\ldots i_{L}}^{r_{1}\ldots r_{L}} & \leq & \min\left\{ Z_{i_{\gamma_{1}}\ldots i_{\gamma_{L-1}}}^{r_{\gamma_{1}}\ldots r_{\gamma_{L-1}}}\right\} _{\left\{ \gamma_{1},\ldots\gamma_{L-1}\right\} \subset\left\{ 1,\ldots L\right\} }\label{eq:L-be_tartozas-1}
\end{eqnarray}
where \textup{$_{i_{1}\ldots i_{L}}^{r_{1}\ldots r_{L}}\in S$}.
\end{lem}

\begin{proof}
It is known that similar inequality holds for arbitrary relative frequencies.
Therefore, 
\begin{eqnarray}
p\left(X_{i_{1}}^{r_{1}}\wedge\ldots\wedge X_{i_{L}}^{r_{L}}\right)\leq\min\left\{ p\left(X_{i_{\gamma_{1}}}^{r_{\gamma_{1}}}\wedge\ldots\wedge X_{i_{\gamma_{L-1}}}^{r_{\gamma_{L-1}}}\right)\right\} _{\left\{ \gamma_{1},\ldots\gamma_{L-1}\right\} \subset\left\{ 1,\ldots L\right\} }\label{eq:L-be_tartozas-1-1}
\end{eqnarray}
for all $_{i_{1}\ldots i_{L}}^{r_{1}\ldots r_{L}}\in S_{max}^{M}$,
and
\begin{eqnarray}
p\left(a_{r_{1}}\wedge\ldots\wedge a_{r_{L}}\right)\leq\min\left\{ p\left(a_{r_{\gamma_{1}}}\wedge\ldots\wedge a_{r_{\gamma_{L-1}}}\right)\right\} _{\left\{ \gamma_{1},\ldots\gamma_{L-1}\right\} \subset\left\{ 1,\ldots L\right\} }\label{eq:L-be_tartozas-1-2}
\end{eqnarray}
for all $2\leq L\leq m$, $1\leq r_{1},\ldots r_{L}\leq m$. It follows
from the definition of state that 
\begin{eqnarray}
p\left(X_{i_{1}}^{r_{1}}\wedge\ldots\wedge X_{i_{L}}^{r_{L}}\right) & = & Z_{i_{1}\ldots i_{L}}^{r_{1}\ldots r_{L}}p\left(a_{r_{1}}\wedge\ldots\wedge a_{r_{L}}\right)\label{eq:freq1}\\
p\left(X_{i_{\gamma_{1}}}^{r_{\gamma_{1}}}\wedge\ldots\wedge X_{i_{\gamma_{L-1}}}^{r_{\gamma_{L-1}}}\right) & = & Z_{i_{\gamma_{1}}\ldots i_{\gamma_{L-1}}}^{r_{\gamma_{1}}\ldots r_{\gamma_{L-1}}}p\left(a_{r_{\gamma_{1}}}\wedge\ldots\wedge a_{r_{\gamma_{L-1}}}\right)\label{eq:freq2}
\end{eqnarray}
for all $_{i_{1}\ldots i_{L}}^{r_{1}\ldots r_{L}}\in S$ and $\left\{ \gamma_{1},\ldots\gamma_{L-1}\right\} \subset\left\{ 1,\ldots L\right\} $.
Consequently, from (\ref{eq:L-be_tartozas-1-1}) we have
\begin{eqnarray}
\frac{Z_{i_{1}\ldots i_{L}}^{r_{1}\ldots r_{L}}}{Z_{i_{\gamma_{1}}\ldots i_{\gamma_{L-1}}}^{r_{\gamma_{1}}\ldots r_{\gamma_{L-1}}}} & \leq & \frac{p\left(a_{r_{\gamma_{1}}}\wedge\ldots\wedge a_{r_{\gamma_{L-1}}}\right)}{p\left(a_{r_{1}}\wedge\ldots\wedge a_{r_{L}}\right)}\label{eq:keresztbe}
\end{eqnarray}
 Since, according to the definition of state, (\ref{eq:freq1})--(\ref{eq:freq2})
hold for all possible relative frequencies $\left\{ p\left(a_{r}\right)\right\} _{1\leq r\leq m}$
and $\left\{ p\left(a_{r_{1}}\wedge\ldots\wedge a_{r_{L}}\right)\right\} _{\left\{ r_{1},\ldots r_{L}\right\} \in\mathfrak{P}}$,
inequality (\ref{eq:keresztbe}) must hold for the minimum value of
the right hand side, which is equal to $1$, due to (\ref{eq:L-be_tartozas-1-2}).
And this is the case for all $_{i_{1}\ldots i_{L}}^{r_{1}\ldots r_{L}}\in S$
and $\left\{ \gamma_{1},\ldots\gamma_{L-1}\right\} \subset\left\{ 1,\ldots L\right\} $.
\end{proof}
It is of course an empirical question what states a system has after
different physical preparations. In what follows, we will answer the
question: what can we say about the ``space'' of theoretically possible
states of a system? Where by ``theoretically possible states'' we
mean all vectors constructed by means of definition (\ref{eq:state})--(\ref{eq:conditional11})
from arbitrary relative frequencies satisfying\textbf{ }(\ref{eq:nem-b})--(\ref{eq:1d})
and (\ref{eq:pi_is_such-1})--(\ref{eq:pi_is_such-2}). Here we should
note that the general probabilistic description includes the possibility---again,
as an eventual empirical fact observed from the frequencies (\ref{eq:0a})---that
the system is deterministic, meaning that $\vec{Z}\in\left\{ 0,1\right\} ^{M+\left|S\right|}$,
or at least it behaves deterministically in some states.

We will show that the possible state vectors constitute a closed convex
polytope in $\mathbb{R}^{M+\left|S\right|}$, which we will denote
by $\varphi\left(M,S\right)$. First we will prove an important lemma.
\begin{lem}
\label{Lemma:convex}If $\vec{Z}_{1}$ and $\vec{Z}_{2}$ are possible
states then their convex linear combination $\vec{Z}_{3}=\lambda_{1}\vec{Z}_{1}+\lambda_{2}\vec{Z}_{2}$
($\lambda_{1},\lambda_{2}\geq0\,\,\,\,\lambda_{1}+\lambda_{2}=1$)
also constitutes a possible state.
\end{lem}

\begin{proof}
According to the definition of state, the observed relative frequencies
of the measurement outcomes in the two states are
\begin{alignat}{1}
p_{1}\left(X_{i}^{r}\right) & =Z_{1}{}_{i}^{r}p_{1}\left(a_{r}\right)\label{eq:pi_is_such-1-1}\\
p_{1}\left(X_{i_{1}}^{r_{1}}\wedge\ldots\wedge X_{i_{L}}^{r_{L}}\right) & =Z_{1}{}_{i_{1}\ldots i_{L}}^{r_{1}\ldots r_{L}}p_{1}\left(a_{r_{1}}\wedge\ldots\wedge a_{r_{L}}\right)\label{eq:pi_is_such-2-1}
\end{alignat}
and 
\begin{alignat}{1}
p_{2}\left(X_{i}^{r}\right) & =Z_{2}{}_{i}^{r}p_{2}\left(a_{r}\right)\label{eq:pi_is_such-1-2}\\
p_{2}\left(X_{i_{1}}^{r_{1}}\wedge\ldots\wedge X_{i_{L}}^{r_{L}}\right) & =Z_{2}{}_{i_{1}\ldots i_{L}}^{r_{1}\ldots r_{L}}p_{2}\left(a_{r_{1}}\wedge\ldots\wedge a_{r_{L}}\right)\label{eq:pi_is_such-2-2}
\end{alignat}
for all $_{i}^{r}\in I^{M}$ and $_{i_{1}\ldots i_{L}}^{r_{1}\ldots r_{L}}\in S$.
Due to \textbf{(E3)}, $p_{1}\left(a_{r}\right)$ and $p_{1}\left(a_{r_{1}}\wedge\ldots\wedge a_{r_{L}}\right)$
as well as $p_{2}\left(a_{r}\right)$ and $p_{2}\left(a_{r_{1}}\wedge\ldots\wedge a_{r_{L}}\right)$
can be arbitrary relative frequencies satisfying (\ref{eq:nem-a})--(\ref{eq:nem-c}).
Therefore, without loss of generality, we can take the case of
\begin{alignat*}{2}
 &  &  & p_{1}\left(a_{r}\right)=p_{2}\left(a_{r}\right)=p_{0}\left(a_{r}\right)\\
 &  &  & p_{1}\left(a_{r_{1}}\wedge\ldots\wedge a_{r_{L}}\right)=p_{2}\left(a_{r_{1}}\wedge\ldots\wedge a_{r_{L}}\right)=p_{0}\left(a_{r_{1}}\wedge\ldots\wedge a_{r_{L}}\right)
\end{alignat*}

Now, consider the convex linear combination $p_{3}=\lambda_{1}p_{1}+\lambda_{2}p_{2}$.
Obviously, $p_{3}$ satisfies (\ref{eq:nem-b})--(\ref{eq:osszeg=00003D1}),
and
\begin{eqnarray*}
p_{3}\left(a_{r}\right) & = & p_{0}\left(a_{r}\right)\\
p_{3}\left(a_{r_{1}}\wedge\ldots\wedge a_{r_{L}}\right) & = & p_{0}\left(a_{r_{1}}\wedge\ldots\wedge a_{r_{L}}\right)
\end{eqnarray*}
Accordingly, we have
\begin{align*}
p_{3}\left(X_{i}^{r}\right) & =\lambda_{1}p_{1}\left(X_{i}^{r}\right)+\lambda_{2}p_{2}\left(X_{i}^{r}\right)=\left(\lambda_{1}Z_{1}{}_{i}^{r}+\lambda_{2}Z_{2}{}_{i}^{r}\right)p_{3}\left(a_{r}\right)\\
p_{3}\left(X_{i_{1}}^{r_{1}}\wedge\ldots\wedge X_{i_{L}}^{r_{L}}\right) & =\lambda_{1}p_{1}\left(X_{i_{1}}^{r_{1}}\wedge\ldots\wedge X_{i_{L}}^{r_{L}}\right)+\lambda_{2}p_{2}\left(X_{i_{1}}^{r_{1}}\wedge\ldots\wedge X_{i_{L}}^{r_{L}}\right)\\
 & =\left(\lambda_{1}Z_{1}{}_{i_{1}\ldots i_{L}}^{r_{1}\ldots r_{L}}+\lambda_{2}Z_{2}{}_{i_{1}\ldots i_{L}}^{r_{1}\ldots r_{L}}\right)p_{3}\left(a_{r_{1}}\wedge\ldots\wedge a_{r_{L}}\right)
\end{align*}
 This means that $\vec{Z}_{3}=\lambda_{1}\vec{Z}_{1}+\lambda_{2}\vec{Z}_{2}$
satisfies condition (\ref{eq:pi_is_such-1})--(\ref{eq:pi_is_such-2}),
as $p_{3}\left(a_{r}\right)$ and $p_{3}\left(a_{r_{1}}\wedge\ldots\wedge a_{r_{L}}\right)$
can be arbitrary frequencies satisfying (\ref{eq:nem-a})--(\ref{eq:nem-c}).
That is, $\vec{Z}_{3}$ complies with the definition of state, meaning
that $\vec{Z}_{3}$ is a possible state of the system.
\end{proof}

\subsection{The State Space -- Polytope View\protect\label{sec:The-State-Space-polytope}}

Now we turn to the question of the ``space'' of possible states.
Let $\vec{e}\,_{i}^{r},\vec{e}\,{}_{i_{1}\ldots i_{L}}^{r_{1}\ldots r_{L}}\in\mathbb{R}^{M+\left|S\right|}$
denote the $_{i}^{r}$-th and $_{i_{1}\ldots i_{L}}^{r_{1}\ldots r_{L}}$-th
coordinate base vector of $\mathbb{R}^{M+\left|S\right|}$, where
$_{i}^{r}\in I^{M}$ and $_{i_{1}\ldots i_{L}}^{r_{1}\ldots r_{L}}\in S$,
and let $\vec{f}=\left(f_{i}^{r},f_{i_{1}\ldots i_{L}}^{r_{1}\ldots r_{L}}\right)\in\mathbb{R}^{M+\left|S\right|}$
denote an arbitrary vector.

The empirical facts \textbf{(E1)}--\textbf{(E3)}, partly through
Lemmas~\ref{Lemma:ineq}--\ref{Lemma:convex}, imply that the possible
state vectors constitute a closed convex polytope $\varphi\left(M,S\right)\subset\mathbb{R}^{M+\left|S\right|}$
defined by the following system of linear inequalities:
\begin{eqnarray}
f_{i}^{r} & \geq & 0\label{eq:elso=000020ineq-1}\\
f_{i}^{r} & \leq & 1\label{eq:f=00003D1}\\
f_{i_{1}\ldots i_{L}}^{r_{1}\ldots r_{L}} & \geq & 0\label{eq:konjunkcio-pozitiv-1}\\
f_{i_{1}\ldots i_{L}}^{r_{1}\ldots r_{L}}-f_{i_{\gamma_{1}}\ldots i_{\gamma_{L-1}}}^{r_{\gamma_{1}}\ldots r_{\gamma_{L-1}}} & \leq & 0\,\,\,\,\,\,\,\,\,\,\left\{ \gamma_{1},\ldots\gamma_{L-1}\right\} \subset\left\{ 1,\ldots L\right\} \label{eq:utolso-elotti-1}\\
\sum_{\begin{array}{c}
k\\
\left(_{k}^{r}\in I^{M}\right)
\end{array}}f_{k}^{r} & = & 1\label{eq:elso+-1}\\
\sum_{\begin{array}{c}
k_{1},k_{2}\ldots k_{L}\\
\left(_{k_{1}\ldots k_{L}}^{r_{1}\ldots r_{L}}\in S\right)
\end{array}}f_{k_{1}\ldots k_{L}}^{r_{1}\ldots r_{L}} & = & 1\label{eq:masodik+}\\
f_{i'_{1}\ldots i'_{L}}^{r'_{1}\ldots r'_{L}} & = & 0\,\,\,\,\,\,\,\,\,\,{}_{i'_{1}\ldots i'_{L}}^{r'_{1}\ldots r'_{L}}\in S_{0}\label{eq:utolso+-1}
\end{eqnarray}
for all $_{i}^{r}\in I^{M}$, $_{i_{1}\ldots i_{L}}^{r_{1}\ldots r_{L}}\in S$,
and 
\begin{eqnarray*}
S_{0} & = & \left\{ _{i_{1}\ldots i_{L}}^{r_{1}\ldots r_{L}}\in S\Bigl|r_{\gamma_{1}}=r_{\gamma_{2}},i_{\gamma_{1}}\neq i_{\gamma_{2}},\left\{ \gamma_{1},\gamma_{2}\right\} \subset\left\{ 1,\ldots L\right\} \right\} 
\end{eqnarray*}
Denote by $l\left(M,S\right)\subset\mathbb{R}^{M+\left|S\right|}$
the closed convex polytope defined by the first group of inequalities
(\ref{eq:elso=000020ineq-1})--(\ref{eq:utolso-elotti-1}). As is
well known (Pitowsky~1989, pp.~51 and 65), the vertices of $l\left(M,S\right)$
are all the vectors $\vec{v}\in\mathbb{R}^{M+\left|S\right|}$ such
that
\begin{itemize}
\item [(a)]$v_{i}^{r},v_{i_{1}\ldots i_{L}}^{r_{1}\ldots r_{L}}\in\left\{ 0,1\right\} $
for all $_{i}^{r}\in I^{M}$ and $_{i_{1}\ldots i_{L}}^{r_{1}\ldots r_{L}}\in S$.
\item [(b)]$v_{i_{1}\ldots i_{L}}^{r_{1}\ldots r_{L}}\leq\underset{^{_{\left\{ \gamma_{1},\gamma_{2},\ldots\gamma_{L-1}\right\} \subset\,\left\{ 1,2,\ldots L\right\} }}}{\prod}v_{i_{\gamma_{1}}\ldots i_{\gamma_{L-1}}}^{r_{\gamma_{1}}\ldots r_{\gamma_{L-1}}}$
for all $_{i_{1}\ldots i_{L}}^{r_{1}\ldots r_{L}}\in S$.
\end{itemize}
A vertex is called classical if the equality holds everywhere in (b),
and non-classical otherwise.

Obviously, $\varphi\left(M,S\right)\subseteq l\left(M,S\right)$.
What can be said about the vertices of $\varphi\left(M,S\right)$?
\begin{lem}
\label{The-vertices-of-fi}The vertices of $\varphi\left(M,S\right)$
are all the vectors $\vec{f}\in\varphi\left(M,S\right)$ such that
$f_{i}^{r},f_{i_{1}\ldots i_{L}}^{r_{1}\ldots r_{L}}\in\left\{ 0,1\right\} $
for all $_{i}^{r}\in I^{M}$ and $_{i_{1}\ldots i_{L}}^{r_{1}\ldots r_{L}}\in S$.
\end{lem}

\begin{proof}
One direction is trivial: if $\vec{f}\in\varphi\left(M,S\right)$
and $f_{i}^{r},f_{i_{1}\ldots i_{L}}^{r_{1}\ldots r_{L}}\in\left\{ 0,1\right\} $
for all $_{i}^{r}\in I^{M}$ and $_{i_{1}\ldots i_{L}}^{r_{1}\ldots r_{L}}\in S$,
then $\vec{f}$ is a vertex. For, if there exist $\vec{f}',\vec{f}''\in\varphi\left(M,S\right)$
such that $\vec{f}=\lambda\vec{f}'+(1-\lambda)\vec{f}''$ with some
$0<\lambda<1$, then obviously $\vec{f}'=\vec{f}''=\vec{f}$.

The proof of the other direction is quite involved. For a more concise
notation, introduce the following sets of indices:
\begin{eqnarray*}
I & = & \Bigl\{1|_{i}^{r},\,2|{}_{i}^{r},\,3|{}_{i_{1}\ldots i_{L}}^{r_{1}\ldots r_{L}},\,4|_{i_{1}\ldots i_{L}}^{r_{1}\ldots r_{L}}|{}_{i_{\gamma_{1}}\ldots i_{\gamma_{L-1}}}^{r_{\gamma_{1}}\ldots r_{\gamma_{L-1}}},\,5|r,\,6|r_{1}\ldots r_{L},\,7|{}_{i'_{1}\ldots i'_{L}}^{r'_{1}\ldots r'_{L}}\,\Bigl|\,\text{for all }\\
 &  & _{i}^{r}\in I^{M},\,\text{\ensuremath{_{i_{1}\ldots i_{L}}^{r_{1}\ldots r_{L}}}\ensuremath{\ensuremath{\in}S},}{}_{i'_{1}\ldots i'_{L}}^{r'_{1}\ldots r'_{L}}\in S_{0},\text{and }\left\{ \gamma_{1},\ldots\gamma_{L-1}\right\} \subset\left\{ 1,\ldots L\right\} \Bigr\}\\
I^{0} & = & \Bigl\{1|_{i}^{r},\,2|{}_{i}^{r},\,3|{}_{i_{1}\ldots i_{L}}^{r_{1}\ldots r_{L}},\,4|_{i_{1}\ldots i_{L}}^{r_{1}\ldots r_{L}}|{}_{i_{\gamma_{1}}\ldots i_{\gamma_{L-1}}}^{r_{\gamma_{1}}\ldots r_{\gamma_{L-1}}}\,\Bigl|\,\text{for all }_{i}^{r}\in I^{M},{}_{i_{1}\ldots i_{L}}^{r_{1}\ldots r_{L}}\in S,\\
 &  & \text{and }\left\{ \gamma_{1},\ldots\gamma_{L-1}\right\} \subset\left\{ 1,\ldots L\right\} \Bigr\}\\
I^{+} & = & \Bigl\{5|r,\,6|r_{1}\ldots r_{L},\,7|_{i'_{1}\ldots i'_{L}}^{r'_{1}\ldots r'_{L}}\,\left|\,\text{for all }1\leq r\leq m,\,\left\{ r_{1}\ldots r_{L}\right\} \in\mathfrak{P},\right.\\
 &  & \text{and }{}_{i'_{1}\ldots i'_{L}}^{r'_{1}\ldots r'_{L}}\in S_{0}\Bigr\}
\end{eqnarray*}
Obviously, $I=I^{0}\cup I^{+}$ and $I^{0}\cap I^{+}=\emptyset$.
Rewrite (\ref{eq:elso=000020ineq-1})--(\ref{eq:utolso+-1}) in the
following standard form:
\begin{eqnarray}
\left\langle \vec{\omega}_{\mu},\vec{f}\right\rangle -b_{\mu} & \leq & 0\,\,\,\,\,\text{for all }\mu\in I^{0}\label{eq:l(mS)-system-1}\\
\left\langle \vec{\omega}_{\mu},\vec{f}\right\rangle -b_{\mu} & = & 0\,\,\,\,\,\text{for all }\mu\in I^{+}
\end{eqnarray}
with the following $\vec{\omega}_{\mu}\in\mathbb{R}^{M+\left|S\right|}$
and $b_{\mu}\in\mathbb{R}$:
\begin{eqnarray}
\vec{\omega}_{1|_{i}^{r}} & = & (0\ldots0\,\overset{\underset{\smallsmile}{_{i}^{r}}}{-1}\,0\ldots0)\label{eq:omegak-1}\\
b_{1|_{i}^{r}} & = & 0\label{eq:omegak-2}\\
\vec{\omega}_{2|_{i}^{r}} & = & (0\ldots0\,\overset{\underset{\smallsmile}{_{i}^{r}}}{1}\,0\ldots0)\label{eq:omegak-3}\\
b_{2|_{i}^{r}} & = & 1\label{eq:omegak-4}\\
\vec{\omega}_{3|_{i_{1}\ldots i_{L}}^{r_{1}\ldots r_{L}}} & = & (0\ldots0\,\overset{\underset{\smallsmile}{_{i_{1}\ldots i_{L}}^{r_{1}\ldots r_{L}}}}{-1}\,0\ldots0)\label{eq:omegak-5}\\
b_{3|_{i_{1}\ldots i_{L}}^{r_{1}\ldots r_{L}}} & = & 0\label{eq:omegak-6}\\
\vec{\omega}_{4|_{i_{1}\ldots i_{L}}^{r_{1}\ldots r_{L}}|{}_{i_{\gamma_{1}}\ldots i_{\gamma_{L-1}}}^{r_{\gamma_{1}}\ldots r_{\gamma_{L-1}}}} & = & (0\ldots0\,\overset{\underset{\smallsmile}{_{i_{\gamma_{1}}\ldots i_{\gamma_{L-1}}}^{r_{\gamma_{1}}\ldots r_{\gamma_{L-1}}}}}{-1}\,0\ldots0\,\overset{\underset{\smallsmile}{_{i_{1}\ldots i_{L}}^{r_{1}\ldots r_{L}}}}{1}\,0\ldots0)\label{eq:omegak-7}\\
b_{4|_{i_{1}\ldots i_{L}}^{r_{1}\ldots r_{L}}|{}_{i_{\gamma_{1}}\ldots i_{\gamma_{L-1}}}^{r_{\gamma_{1}}\ldots r_{\gamma_{L-1}}}} & = & 0\label{eq:omegak-8}\\
\vec{\omega}_{5|r} & = & (0\ldots0\,\overset{\underset{\smallsmile}{_{1}^{r}}}{1}\,\overset{\underset{\smallsmile}{_{2}^{r}}}{1}\,\overset{\underset{\smallsmile}{_{3}^{r}}}{1}\ldots\overset{\underset{\smallsmile}{_{n_{r}}^{r}}}{1}\,0\ldots0)\label{eq:omegak-9}\\
b_{5|r} & = & 1\label{eq:omegak-10}\\
\vec{\omega}_{6|r_{1}\ldots r_{L}} & = & (0\ldots0\,\overset{\underset{\smallsmile}{_{\,\,1\,\,\,1\,\,\ldots\,1}^{r_{1}r_{2}\ldots r_{L}}}}{1}\,\,\overset{\underset{\smallsmile}{_{\,\,2\,\,\,1\,\,\ldots\,1}^{r_{1}r_{2}\ldots r_{L}}}}{1}\ldots\overset{\underset{\smallsmile}{_{n_{r_{1}}\,n_{r_{2}}\ldots n_{r_{L}}}^{\,r_{1}\,\,\,\,r_{2}\,\,\,\ldots r_{L}}}}{1}\,0\ldots0)\label{eq:omegak-13}\\
b_{6|r_{1}\ldots r_{L}} & = & 1\label{eq:omegak-14}\\
\vec{\omega}_{7|_{i'_{1}\ldots i'_{L}}^{r'_{1}\ldots r'_{L}}} & = & (0\ldots0\,\overset{\underset{\smallsmile}{_{i'_{1}\ldots i'_{L}}^{r'_{1}\ldots r'_{L}}}}{1}\,0\ldots0)\label{eq:omegak-11}\\
b_{7|_{i'_{1}\ldots i'_{L}}^{r'_{1}\ldots r'_{L}}} & = & 0\label{eq:omegak-12}
\end{eqnarray}
where $_{i}^{r}\in I^{M}$, $_{i_{1}\ldots i_{L}}^{r_{1}\ldots r_{L}}\in S$,
$\left\{ \gamma_{1},\ldots\gamma_{L-1}\right\} \subset\left\{ 1,\ldots L\right\} $,
and $_{i'_{1}\ldots i'_{L}}^{r'_{1}\ldots r'_{L}}\in S_{0}$. Notice
that $l\left(M,S\right)$ is defined by (\ref{eq:l(mS)-system-1}).

For an arbitrary $\vec{f}\in l\left(M,S\right)$ we define the following
sets:
\begin{eqnarray*}
I_{\vec{f}} & = & \left\{ \mu\in I\left|\,\left\langle \vec{\omega}_{\mu},\vec{f}\right\rangle -b_{\mu}=0\right.\right\} \\
I_{\vec{f}}^{0} & = & \left\{ \mu\in I^{0}\left|\,\left\langle \vec{\omega}_{\mu},\vec{f}\right\rangle -b_{\mu}=0\right.\right\} 
\end{eqnarray*}
Notice that if $\vec{f}\in\varphi\left(M,S\right)$, then $I_{\vec{f}}=I_{\vec{f}}^{0}\cup I^{+}$,
due to the fact that (\ref{eq:elso+-1})--(\ref{eq:utolso+-1}) can
be satisfied only with equality.

$I_{\vec{f}}$ constitutes the so called `active index set' for $\vec{f}\in\varphi\left(M,S\right)$;
and according to a known theorem (see Appendix~1){\small{}  }$\vec{f}$
is a vertex of $\varphi\left(M,S\right)$ if and only if 
\begin{equation}
\mathrm{span}\left\{ \vec{\omega}_{\mu}\right\} _{\mu\in I_{\vec{f}}}=\mathbb{R}^{M+\left|S\right|}\label{eq:span-3}
\end{equation}
Similarly, a vector $\vec{f}\in l\left(M,S\right)$ is a vertex of
$l\left(M,S\right)$ if and only if 
\begin{equation}
\mathrm{span}\left\{ \vec{\omega}_{\mu}\right\} _{\mu\in I_{\vec{f}}^{0}}=\mathbb{R}^{M+\left|S\right|}\label{eq:span-1-1}
\end{equation}
For all $\vec{f}\in l\left(M,S\right)$ define 
\begin{eqnarray*}
J_{\vec{f}} & = & \left\{ _{i}^{r}\left|_{i}^{r}\in I^{M}\text{ and }0<f_{i}^{r}<1\right.\right\} \\
J'_{\vec{f}} & = & \left\{ _{i_{1}\ldots i_{L}}^{r_{1}\ldots r_{L}}\left|_{i_{1}\ldots i_{L}}^{r_{1}\ldots r_{L}}\in S\text{ and }0<f{}_{i_{1}\ldots i_{L}}^{r_{1}\ldots r_{L}}<1\right.\right\} 
\end{eqnarray*}
Notice that for all $_{i}^{r}\in I^{M}$ and $_{i_{1}\ldots i_{L}}^{r_{1}\ldots r_{L}}\in S$,
\begin{eqnarray}
\vec{e}\,_{i}^{r}\in\mathrm{span}\left\{ \vec{\omega}_{\mu}\right\} _{\mu\in I_{\vec{f}}^{0}} & \text{if} & f_{i}^{r}\in\{0,1\}\label{eq:0-bazis-1}\\
\vec{e}\,{}_{i_{1}\ldots i_{L}}^{r_{1}\ldots r_{L}}\in\mathrm{span}\left\{ \vec{\omega}_{\mu}\right\} _{\mu\in I_{\vec{f}}^{0}} & \text{if} & f_{i_{1}\ldots i_{L}}^{r_{1}\ldots r_{L}}\in\{0,1\}\label{eq:0-bazis-2}
\end{eqnarray}
since the corresponding inequalities (\ref{eq:elso=000020ineq-1})--(\ref{eq:utolso-elotti-1})
must hold with equality. The only case that requires a bit of reflection
is when $f_{i_{1}\ldots i_{L}}^{r_{1}\ldots r_{L}}=1$. For example,
if $f_{i_{1}i_{2}}^{r_{1}r_{2}}=1$ then (\ref{eq:utolso-elotti-1})
is satisfied with equality, so that $f_{i_{1}}^{r_{1}}=1$, therefore
(\ref{eq:f=00003D1}) is also satisfied with equality. Consequently,
\begin{eqnarray*}
\vec{\omega}_{2|_{i_{1}}^{r_{1}}} & \in & \mathrm{span}\left\{ \vec{\omega}_{\mu}\right\} _{\mu\in I_{\vec{f}}^{0}}\\
\vec{\omega}_{4|_{i_{1}i_{2}}^{r_{1}r_{2}}|{}_{i_{1}}^{r_{1}}} & \in & \mathrm{span}\left\{ \vec{\omega}_{\mu}\right\} _{\mu\in I_{\vec{f}}^{0}}
\end{eqnarray*}
At the same time, as it can be seen from (\ref{eq:omegak-3}) and
(\ref{eq:omegak-7}),
\[
\vec{e}\,_{i_{1}i_{2}}^{r_{1}r_{2}}=\vec{\omega}_{2|_{i_{1}}^{r_{1}}}+\vec{\omega}_{4|_{i_{1}i_{2}}^{r_{1}r_{2}}|{}_{i_{1}}^{r_{1}}}\in\mathrm{span}\left\{ \vec{\omega}_{\mu}\right\} _{\mu\in I_{\vec{f}}^{0}}
\]
This can be recursively continued for the triple and higher conjunction
indices.

Assume now that $\vec{f}\in\varphi\left(M,S\right)$ is such that
$J_{\vec{f}}\cup J'_{\vec{f}}\neq\emptyset$, and at the same time
it is a vertex of $\varphi\left(M,S\right)$, that is, (\ref{eq:span-3})
is satisfied. We are going to show that this leads to contradiction.

Due to (\ref{eq:0-bazis-1})--(\ref{eq:0-bazis-2}), the assumption
that $\vec{f}$ is a vertex implies that all base vectors of $\mathbb{R}^{M+\left|S\right|}$
must belong to $\mathrm{span}\left\{ \vec{\omega}_{\mu}\right\} _{\mu\in I_{\vec{f}}^{0}}$
save for some $\vec{e}\,_{i}^{r}$'s with $_{i}^{r}\in J_{\vec{f}}$
and/or some $\vec{e}\,{}_{i_{1}\ldots i_{L}}^{r_{1}\ldots r_{L}}$'s
with $_{i_{1}\ldots i_{L}}^{r_{1}\ldots r_{L}}\in J'_{\vec{f}}$.
On the other hand, $\vec{f}$ being a vertex implies that (\ref{eq:span-3})
holds, therefore
\begin{eqnarray*}
\vec{e}\,_{i}^{r} & \in & \mathrm{span}\left\{ \vec{\omega}_{\mu}\right\} _{\mu\in I_{\vec{f}}}\,\,\,\,\,\,\,\,\,\text{for all }{}_{i}^{r}\in J_{\vec{f}}\\
\vec{e}\,{}_{i_{1}\ldots i_{L}}^{r_{1}\ldots r_{L}} & \in & \mathrm{span}\left\{ \vec{\omega}_{\mu}\right\} _{\mu\in I_{\vec{f}}}\,\,\,\,\,\,\,\,\,\text{for all }{}_{i_{1}\ldots i_{L}}^{r_{1}\ldots r_{L}}\in J'_{\vec{f}}
\end{eqnarray*}
Taking into account that $I_{\vec{f}}=I_{\vec{f}}^{0}\cup I^{+}$,
it means that for all $_{i}^{r}\in J_{\vec{f}}$ and arbitrary $\tau_{i}^{r}\neq0$
there exist vectors 
\begin{eqnarray}
_{i}^{r}\vec{v} & \in & \mathrm{span}\left\{ \vec{\omega}_{\mu}\right\} _{\mu\in I_{\vec{f}}^{0}}\label{eq:v-span}
\end{eqnarray}
such that,
\begin{eqnarray}
\tau_{i}^{r}\vec{e}\,_{i}^{r} & = & _{i}^{r}\vec{v}+\sum_{s=1}^{m}{}_{i}^{r}\kappa_{s}\vec{\omega}_{5|s}+\sum_{\left\{ r_{1},\ldots r_{L'}\right\} \in\mathfrak{P}}{}_{i}^{r}\kappa'_{r_{1}\ldots r_{L'}}\vec{\omega}_{6|r_{1}\ldots r_{L'}}\nonumber \\
 &  & +\sum_{_{i'_{1}\ldots i'_{L}}^{r'_{1}\ldots r'_{L}}\in S_{0}}{}_{i}^{r}\lambda_{_{i'_{1}\ldots i'_{L}}^{r'_{1}\ldots r'_{L}}}\vec{\omega}_{7|_{i'_{1}\ldots i'_{L}}^{r'_{1}\ldots r'_{L}}}\label{eq:dekompozicio}
\end{eqnarray}
with some real numbers $_{i}^{r}\kappa_{s}$, $_{i}^{r}\kappa'_{r_{1}\ldots r_{L'}}$,
and $_{i}^{r}\lambda_{_{i'_{1}\ldots i'_{L}}^{r'_{1}\ldots r'_{L}}}$.
From the definitions of $\vec{\omega}_{5|r}$, $\vec{\omega}_{6|r_{1}\ldots r_{L}}$,
and $\vec{\omega}_{7|_{i'_{1}\ldots i'_{L}}^{r'_{1}\ldots r'_{L}}}$
in (\ref{eq:omegak-1})--(\ref{eq:omegak-12}) we can write:
\begin{align}
\tau_{i}^{r}\vec{e}\,_{i}^{r} & =_{i}^{r}\vec{v}+\sum_{_{j}^{s}\in I^{M}}{}_{i}^{r}\kappa_{s}\vec{e}\,_{j}^{s}+\sum_{_{j_{1}\ldots j_{L'}}^{s_{1}\ldots s_{L'}}\in S}{}_{i}^{r}\kappa'_{s_{1}\ldots s_{L'}}\vec{e}\,{}_{j_{1}\ldots j_{L'}}^{s_{1}\ldots s_{L'}}\nonumber \\
 & +\sum_{_{i'_{1}\ldots i'_{L}}^{r'_{1}\ldots r'_{L}}\in S_{0}}{}_{i}^{r}\lambda_{_{i'_{1}\ldots i'_{L}}^{r'_{1}\ldots r'_{L}}}\vec{e}\,{}_{i'_{1}\ldots i'_{L}}^{r'_{1}\ldots r'_{L}}\label{eq:dekompozicio2}
\end{align}
\emph{Mutatis mutandis}, we have the same equation for $\vec{e}\,{}_{i_{1}\ldots i_{L}}^{r_{1}\ldots r_{L}}$
for all $_{i_{1}\ldots i_{L}}^{r_{1}\ldots r_{L}}\in J'_{\vec{f}}$
with arbitrary $\tau{}_{i_{1}\ldots i_{L}}^{r_{1}\ldots r_{L}}\neq0$
and with some numbers $_{i_{1}\ldots i_{L}}^{r_{1}\ldots r_{L}}\kappa_{s}$,
$_{i_{1}\ldots i_{L}}^{r_{1}\ldots r_{L}}\kappa'_{r_{1}\ldots r_{L'}}$,
and $_{i_{1}\ldots i_{L}}^{r_{1}\ldots r_{L}}\lambda_{_{i'_{1}\ldots i'_{L}}^{r'_{1}\ldots r'_{L}}}$:
~
\begin{align}
\tau{}_{i_{1}\ldots i_{L}}^{r_{1}\ldots r_{L}}\vec{e}\,{}_{i_{1}\ldots i_{L}}^{r_{1}\ldots r_{L}} & =\,{}_{i_{1}\ldots i_{L}}^{r_{1}\ldots r_{L}}\vec{v}+\sum_{_{j}^{s}\in I^{M}}{}_{i_{1}\ldots i_{L}}^{r_{1}\ldots r_{L}}\kappa_{s}\vec{e}\,_{j}^{s}+\sum_{_{j_{1}\ldots j_{L'}}^{s_{1}\ldots s_{L'}}\in S}{}_{i_{1}\ldots i_{L}}^{r_{1}\ldots r_{L}}\kappa'_{s_{1}\ldots s_{L'}}\vec{e}\,{}_{j_{1}\ldots j_{L'}}^{s_{1}\ldots s_{L'}}\nonumber \\
 & +\sum_{_{i'_{1}\ldots i'_{L}}^{r'_{1}\ldots r'_{L}}\in S_{0}}{}_{i_{1}\ldots i_{L}}^{r_{1}\ldots r_{L}}\lambda_{_{i'_{1}\ldots i'_{L}}^{r'_{1}\ldots r'_{L}}}\vec{e}\,{}_{i'_{1}\ldots i'_{L}}^{r'_{1}\ldots r'_{L}}\label{eq:dekompozicio2-1}
\end{align}
With some rearrangement, from (\ref{eq:dekompozicio2}) we have
\begin{eqnarray}
 &  & \sum_{\begin{array}{c}
_{j}^{s}\in J_{\vec{f}}\\
_{j}^{s}\neq_{i}^{r}
\end{array}}{}_{i}^{r}\kappa_{s}\vec{e}\,_{j}^{s}+\left(_{i}^{r}\kappa_{r}-\tau_{i}^{r}\right)\vec{e}\,_{i}^{r}+\sum_{_{j_{1}\ldots j_{L'}}^{s_{1}\ldots s_{L'}}\in J'_{\vec{f}}}{}_{i}^{r}\kappa'_{s_{1}\ldots s_{L'}}\vec{e}\,{}_{j_{1}\ldots j_{L'}}^{s_{1}\ldots s_{L'}}\nonumber \\
 &  & \,\,\,\,\,\,\,\,=-\sum_{_{i'_{1}\ldots i'_{L}}^{r'_{1}\ldots r'_{L}}\in S_{0}}{}_{i}^{r}\lambda_{_{i'_{1}\ldots i'_{L}}^{r'_{1}\ldots r'_{L}}}\vec{e}\,{}_{i'_{1}\ldots i'_{L}}^{r'_{1}\ldots r'_{L}}-\sum_{\begin{array}{c}
_{j_{1}\ldots j_{L'}}^{s_{1}\ldots s_{L'}}\in S\\
_{j_{1}\ldots j_{L'}}^{s_{1}\ldots s_{L'}}\notin J'_{\vec{f}}
\end{array}}{}_{i}^{r}\kappa'_{s_{1}\ldots s_{L'}}\vec{e}\,{}_{j_{1}\ldots j_{L'}}^{s_{1}\ldots s_{L'}}\nonumber \\
 &  & \,\,\,\,\,\,\,\,\,\,\,\,\,\,\,\,\,\,\,\,\,\,\,\,\,\,\,\,-\sum_{\begin{array}{c}
_{j}^{s}\in I^{M}\\
_{j}^{s}\notin J_{\vec{f}}
\end{array}}{}_{i}^{r}\kappa_{s}\vec{e}\,_{j}^{s}-{}_{i}^{r}\vec{v}\label{eq:az=000020egyenlet}
\end{eqnarray}
for all $_{i}^{r}\in J_{\vec{f}}$.

Similarly, from (\ref{eq:dekompozicio2-1}) we have 
\begin{eqnarray}
 &  & \sum_{_{j}^{s}\in J_{\vec{f}}}{}_{i_{1}\ldots i_{L}}^{r_{1}\ldots r_{L}}\kappa_{s}\vec{e}\,_{j}^{s}\,\,+\sum_{\begin{array}{c}
_{j_{1}\ldots j_{L'}}^{s_{1}\ldots s_{L'}}\in J'_{\vec{f}}\\
_{j_{1}\ldots j_{L'}}^{s_{1}\ldots s_{L'}}\neq{}_{i_{1}\ldots i_{L}}^{r_{1}\ldots r_{L}}
\end{array}}{}_{i_{1}\ldots i_{L}}^{r_{1}\ldots r_{L}}\kappa'_{s_{1}\ldots s_{L'}}\vec{e}\,{}_{j_{1}\ldots j_{L'}}^{s_{1}\ldots s_{L'}}\nonumber \\
 &  & \,\,\,\,\,\,\,\,\,\,\,\,\,\,\,\,\,\,\,\,\,\,\,\,\,\,\,\,\,\,\,\,\,\,\,\,\,\,\,\,\,\,\,\,\,\,\,\,\,\,\,\,\,\,\,\,\,\,+\left(_{i_{1}\ldots i_{L}}^{r_{1}\ldots r_{L}}\kappa'_{r_{1}\ldots r_{L'}}-\tau{}_{i_{1}\ldots i_{L}}^{r_{1}\ldots r_{L}}\right)\vec{e}\,{}_{i_{1}\ldots i_{L}}^{r_{1}\ldots r_{L}}\nonumber \\
 &  & \,\,\,\,\,\,\,\,=-\sum_{_{i'_{1}\ldots i'_{L}}^{r'_{1}\ldots r'_{L}}\in S_{0}}{}_{i_{1}\ldots i_{L}}^{r_{1}\ldots r_{L}}\lambda_{_{i'_{1}\ldots i'_{L}}^{r'_{1}\ldots r'_{L}}}\vec{e}\,{}_{i'_{1}\ldots i'_{L}}^{r'_{1}\ldots r'_{L}}\nonumber \\
 &  & \,\,\,\,\,\,\,\,\,\,\,\,\,\,\,\,\,\,\,\,\,\,\,\,-\sum_{\begin{array}{c}
_{j_{1}\ldots j_{L'}}^{s_{1}\ldots s_{L'}}\in S\\
_{j_{1}\ldots j_{L'}}^{s_{1}\ldots s_{L'}}\notin J'_{\vec{f}}
\end{array}}{}_{i_{1}\ldots i_{L}}^{r_{1}\ldots r_{L}}\kappa'_{s_{1}\ldots s_{L'}}\vec{e}\,{}_{j_{1}\ldots j_{L'}}^{s_{1}\ldots s_{L'}}\nonumber \\
 &  & \,\,\,\,\,\,\,\,\,\,\,\,\,\;\,\,\,\,\,\,\,\,\,\,\,\,\,\,\,\,\,\,\,\,\,-\sum_{\begin{array}{c}
_{j}^{s}\in I^{M}\\
_{j}^{s}\notin J_{\vec{f}}
\end{array}}{}_{i_{1}\ldots i_{L}}^{r_{1}\ldots r_{L}}\kappa_{s}\vec{e}\,_{j}^{s}-{}_{i_{1}\ldots i_{L}}^{r_{1}\ldots r_{L}}\vec{v}\label{eq:az=000020egyenlet-2}
\end{eqnarray}
for all $_{i_{1}\ldots i_{L}}^{r_{1}\ldots r_{L}}\in J'_{\vec{f}}$.

Denote the right hand side of (\ref{eq:az=000020egyenlet}) by $\vec{B}_{_{i}^{r}}$
and the right hand side of (\ref{eq:az=000020egyenlet-2}) by $\vec{B}_{_{i_{1}\ldots i_{L}}^{r_{1}\ldots r_{L}}}$.
Notice that the vectors $\vec{B}_{_{i}^{r}}$ and $\vec{B}_{_{i_{1}\ldots i_{L}}^{r_{1}\ldots r_{L}}}$
are contained in $\mathrm{span}\left\{ \vec{\omega}_{\mu}\right\} _{\mu\in I_{\vec{f}}^{0}}$,
due to (\ref{eq:v-span}), and (\ref{eq:0-bazis-1})--(\ref{eq:0-bazis-2}).
So, in (\ref{eq:az=000020egyenlet})--(\ref{eq:az=000020egyenlet-2}),
together, we have a system of linear equations with vector-variables
$\left\{ \vec{e}\,_{j}^{s}\right\} _{_{j}^{s}\in J_{\vec{f}}}$ and
$\left\{ \vec{e}\,{}_{j_{1}\ldots j_{L'}}^{s_{1}\ldots s_{L'}}\right\} _{_{j_{1}\ldots j_{L'}}^{s_{1}\ldots s_{L'}}\in J'_{\vec{f}}}$,
which can be written in the following form:
\begin{eqnarray}
\sum_{_{j}^{s}\in J_{\vec{f}},\,{}_{j_{1}\ldots j_{L'}}^{s_{1}\ldots s_{L'}}\in J'_{\vec{f}}}\beta_{\Bigl({}_{i}^{r},\,{}_{i_{1}\ldots i_{L}}^{r_{1}\ldots r_{L}}\Bigr)\Bigl({}_{j}^{s},\,{}_{j_{1}\ldots j_{L'}}^{s_{1}\ldots s_{L'}}\Bigr)}\left(\vec{e}\,_{j}^{s},\,\vec{e}\,{}_{j_{1}\ldots j_{L'}}^{s_{1}\ldots s_{L'}}\right)=\left(\vec{B}_{_{i}^{r}},\vec{B}_{_{i_{1}\ldots i_{L}}^{r_{1}\ldots r_{L}}}\right)\label{eq:az=000020egyenlet-1}
\end{eqnarray}
where $\beta_{\Bigl({}_{i}^{r},\,{}_{i_{1}\ldots i_{L}}^{r_{1}\ldots r_{L}}\Bigr)\Bigl({}_{j}^{s},\,{}_{j_{1}\ldots j_{L'}}^{s_{1}\ldots s_{L'}}\Bigr)}$
is a $\left(|J_{\vec{f}}|+|J'_{\vec{f}}|\right)\times\left(|J_{\vec{f}}|+|J'_{\vec{f}}|\right)$
matrix with diagonal elements 
\begin{align*}
\beta_{_{i}^{r}{}_{i}^{r}} & =_{i}^{r}\kappa_{r}-\tau_{i}^{r}\\
\beta_{_{i_{1}\ldots i_{L}}^{r_{1}\ldots r_{L}}\,{}_{i_{1}\ldots i_{L}}^{r_{1}\ldots r_{L}}} & =_{i_{1}\ldots i_{L}}^{r_{1}\ldots r_{L}}\kappa'_{r_{1}\ldots r_{L}}-\tau{}_{i_{1}\ldots i_{L}}^{r_{1}\ldots r_{L}}
\end{align*}
The off diagonal elements depend only on $_{i}^{r}\kappa_{s}$'s and
$_{i_{1}\ldots i_{L}}^{r_{1}\ldots r_{L}}\kappa'_{s_{1}\ldots s_{L'}}$'s.
Since the numbers $\tau_{i}^{r}\neq0$ and $\tau{}_{i_{1}\ldots i_{L}}^{r_{1}\ldots r_{L}}\neq0$
in the diagonal can be chosen arbitrarily, we may assume that $\det\beta_{\Bigl({}_{i}^{r},\,{}_{i_{1}\ldots i_{L}}^{r_{1}\ldots r_{L}}\Bigr)\Bigl({}_{j}^{s},\,{}_{j_{1}\ldots j_{L'}}^{s_{1}\ldots s_{L'}}\Bigr)}\neq0$.
Therefore, the system of linear equations (\ref{eq:az=000020egyenlet-1})
has a unique solution for all vector-variables $\vec{e}\,_{i}^{r}$
and $\vec{e}\,{}_{i_{1}\ldots i_{L}}^{r_{1}\ldots r_{L}}$, namely,
\[
\left(\vec{e}\,{}_{i}^{r},\,\vec{e}\,{}_{i_{1}\ldots i_{L}}^{r_{1}\ldots r_{L}}\right)=\sum_{_{j}^{s}\in J_{\vec{f}},\,{}_{j_{1}\ldots j_{L'}}^{s_{1}\ldots s_{L'}}\in J'_{\vec{f}}}\beta_{\Bigl({}_{i}^{r},\,{}_{i_{1}\ldots i_{L}}^{r_{1}\ldots r_{L}}\Bigr)\Bigl({}_{j}^{s},\,{}_{j_{1}\ldots j_{L'}}^{s_{1}\ldots s_{L'}}\Bigr)}^{-1}\left(\vec{B}_{_{j}^{s}},\,\vec{B}_{_{j_{1}\ldots j_{L'}}^{s_{1}\ldots s_{L'}}}\right)
\]
Taking into account that $\vec{B}_{_{i}^{r}},\vec{B}_{_{i_{1}\ldots i_{L}}^{r_{1}\ldots r_{L}}}\in\mathrm{span}\left\{ \vec{\omega}_{\mu}\right\} _{\mu\in I_{\vec{f}}^{0}},$
this all means that for all $_{i}^{r}\in J_{\vec{f}}$, the base vectors
$\vec{e}\,_{i}^{r}$, and for all $_{i_{1}\ldots i_{L}}^{r_{1}\ldots r_{L}}\in J'_{\vec{f}}$,
the base vectors $\vec{e}\,{}_{i_{1}\ldots i_{L}}^{r_{1}\ldots r_{L}}$
can be expressed as linear combinations of vectors contained in $\mathrm{span}\left\{ \vec{\omega}_{\mu}\right\} _{\mu\in I_{\vec{f}}^{0}}$.
As all the rest of base vectors belong to $\mathrm{span}\left\{ \vec{\omega}_{\mu}\right\} _{\mu\in I_{\vec{f}}^{0}}$
(as we have already mentioned above), we have 
\[
\mathrm{span}\left\{ \vec{\omega}_{\mu}\right\} _{\mu\in I_{\vec{f}}^{0}}=\mathbb{R}^{M+\left|S\right|}
\]
meaning that $\vec{f}$ must be a vertex of $l\left(M,S\right)$.
Due to the fact that all components of a vertex of $l\left(M,S\right)$
are necessarily $0$ or $1$, there cannot exists a vertex $\vec{f}\in\varphi\left(M,S\right)$
with $0<f_{i}^{r}<1$ and/or $0<f_{i_{1}\ldots i_{L}}^{r_{1}\ldots r_{L}}<1$.
\end{proof}
\noindent All this means that the vertices of $\varphi\left(M,S\right)$
are those vertices of $l\left(M,S\right)$ that satisfy the further
restrictions (\ref{eq:elso+-1})--(\ref{eq:utolso+-1}).

To sum up, the ``space'' of possible states is a closed convex polytope
$\varphi\left(M,S\right)\subset\mathbb{R}^{M+\left|S\right|}$ whose
vertices are the vectors $\vec{w}\in\mathbb{R}^{M+\left|S\right|}$
such that
\begin{itemize}
\item [(a)]$w_{i}^{r},w_{i_{1}\ldots i_{L}}^{r_{1}\ldots r_{L}}\in\left\{ 0,1\right\} $
for all $_{i}^{r}\in I^{M}$ and $_{i_{1}\ldots i_{L}}^{r_{1}\ldots r_{L}}\in S$.
\item [(b)]$w_{i_{1}\ldots i_{L}}^{r_{1}\ldots r_{L}}\leq\underset{^{_{\left\{ \gamma_{1},\ldots\gamma_{L-1}\right\} \subset\,\left\{ 1,\ldots L\right\} }}}{\prod}w_{i_{\gamma_{1}}\ldots i_{\gamma_{L-1}}}^{r_{\gamma_{1}}\ldots r_{\gamma_{L-1}}}$
for all $_{i_{1}\ldots i_{L}}^{r_{1}\ldots r_{L}}\in S$.
\item [(c)]$w_{i_{1}\ldots i_{L}}^{r_{1}\ldots r_{L}}=0$ for all $_{i_{1}\ldots i_{L}}^{r_{1}\ldots r_{L}}\in S_{0}$.
\item [(d)]For all $1\leq r\leq m$ there is exactly one $1\leq i_{*}^{r}\leq n_{r}$
such that $w_{i_{*}^{r}}^{r}=1$.
\item [(e)]For all $\left\{ r_{1},\ldots r_{L}\right\} \in\mathfrak{P}$
there is exactly one $_{i_{1}^{*}\ldots i_{L}^{*}}^{r_{1}\ldots r_{L}}\in S$
such that $w_{i_{1}^{*}\ldots i_{L}^{*}}^{r_{1}\ldots r_{L}}=1$.
\end{itemize}
(We note in advance that property (d) will be crucial in the proof
of Theorem~\ref{Theorem:Main}.) Let $\mathcal{W}=\left\{ \vec{w}_{\vartheta}\right\} _{\vartheta\in\Theta}$
denote the set of vertices of $\varphi\left(M,S\right)$.

Again, we emphasize that $\varphi\left(M,S\right)$ is the total space
of theoretically possible states, determined by the totality of possible
relative frequency functions over $\mathcal{A}$ that satisfy conditions
\textbf{(E1)}--\textbf{(E3)}. It may be that the empirically determined
possible states of the system for different physical preparations
constitute only a subset $\varPhi\subseteq\varphi\left(M,S\right)$.
Note that such a subset is not necessarily a convex subset in $\varphi\left(M,S\right)$---notice
that Lemma~\ref{Lemma:convex} is about the whole $\varphi\left(M,S\right)$.

As a simple illustration, consider a coin made of a very light material,
which is empty inside. Inside we install a heavy metal disk, such
that the disk can be fixed in different positions. (Fig.~\ref{coin})
In this way, the system can be prepared in different states. There
is only one measurement we may perform, tossing the coin, with two
possible outcomes, Heads and Tails, and one conjunction. 
\begin{figure}
\begin{centering}
\includegraphics[width=0.8\columnwidth]{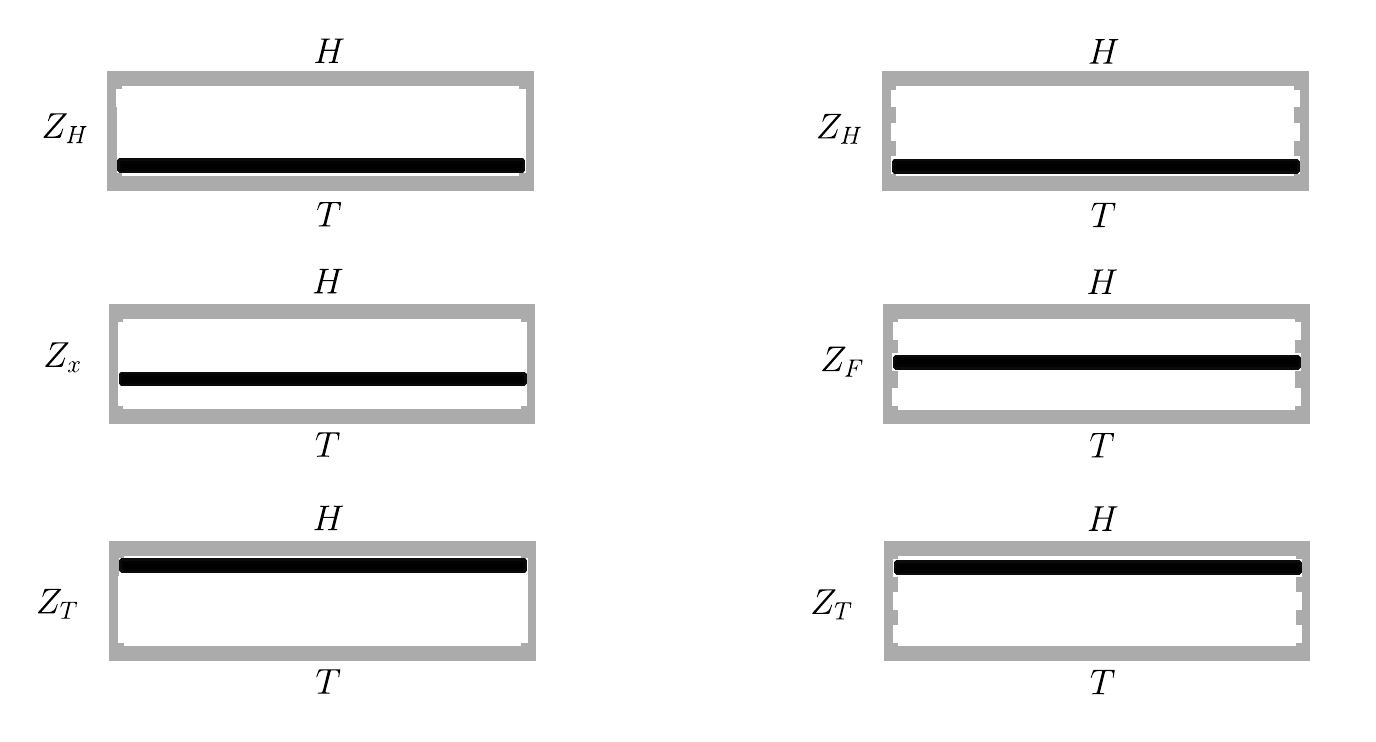}
\par\end{centering}
\caption{Baised coin}
\label{coin}
\end{figure}
 Keeping the preparation fixed, that is, keeping the position of the
disk fixed, we can take the statistics of Heads and Tails. The observed
data will perfectly satisfy \textbf{(E1)}--\textbf{(E3)}, and the
state space $\varphi\left(M,S\right)$ (where $M=2$ and $S=\left\{ _{1\,2}^{1\,1}\right\} $)
is a one-dimensional polytope in $\mathbb{R}^{3}$, determined by
two vertices $\left(1,0,0\right)$ and $\left(0,1,0\right)$. Denote
an arbitrary state vector by $\vec{Z}=(z_{H},z_{T},0)$, meaning that
whenever we perform the coin-toss and the coin is in state $\vec{Z}$,
we get Heads with relative frequency $z_{H}$, Tails with $z_{T}$,
but never the Heads and Tails in conjunction.

Consider now the particular physical situation on the left hand side
of Fig.~\ref{coin}. The position of the disk can be continuously
varied between two extreme positions; one in which the coin is maximally
biased for Heads,
\[
\vec{Z}_{H}=(0.8,0.2,0)
\]
the other in which the coin is maximally biased for Tails, 
\[
\vec{Z}_{T}=(0.2,0.8,0)
\]
Each intermediate position is possible and the corresponding state
vector $\vec{Z}_{x}$ falls between vectors $\vec{Z}_{H}$ and $\vec{Z}_{T}$.
So, the physically possible sates are restricted to a proper subset
$\varPhi\subset\varphi\left(M,S\right)$, the line segment between
$\vec{Z}_{H}$ and $\vec{Z}_{T}$; but, $\varPhi$ is closed under
convex combination.

Consider now another physical situation shown on the right hand side
of Fig.~\ref{coin}. There are three separate slots where the disk
can be clicked. Accordingly, there are only three possible positions,
the two extremes and a third in the middle. The corresponding state
vectors are $\vec{Z}_{H}$, $\vec{Z}_{T}$, and the ``fair'' state
\[
\vec{Z}_{F}=(0.5,0.5,0)
\]
Now, 
\begin{equation}
\varPhi=\left\{ \vec{Z}_{H},\vec{Z}_{F},\vec{Z}_{T}\right\} \subset\varphi\left(M,S\right)\label{eq:coin-diszkret}
\end{equation}
and it is obviously not closed under convex combination. While for
example $\frac{1}{2}\vec{Z}_{H}+\frac{1}{2}\vec{Z}_{T}=\vec{Z}_{F}$
is indeed contained in $\varPhi$, $\frac{2}{3}\vec{Z}_{H}+\frac{1}{3}\vec{Z}_{T}=(0.6,0.4,0)$
is not.

Thus, the set of physically possible states $\varPhi$ can be a strongly
restricted subset of the total space of theoretically possible states
$\varphi\left(M,S\right)$; and it is not necessary closed under convex
combination. (We do not see any reason to adopt for example the a
priori argumentation in the GPT literature that $\varPhi$ must be
convex; see Appendix~2.) The actual content of $\varPhi\subseteq\varphi\left(M,S\right)$
is determined by further physical facts beyond the\textbf{ }stipulated
\textbf{(E1)}--\textbf{(E3)}. We do not wish to impose such an additional
restriction, even if it could be empirically justified for certain
types of physical systems. The point of assumptions \textbf{(E1)}-\textbf{(E3)}
is precisely that they are general; we know of no physical system
whose description in operational terms does not satisfy them.

Keeping all this in mind, for the sake of generality, we will consider
$\varphi\left(M,S\right)$ as if it were the space of states without
any restrictions. For the purposes of our analysis below, this is
of no particular relevance, and all of our results below can be easily
modified for a particular subset $\varPhi\subset\varphi\left(M,S\right)$.

\subsection{The State Space -- Manifold View\protect\label{sec:The-State-Space-manifold}}

A closed convex polytope like $\varphi\left(M,S\right)\subset\mathbb{R}^{M+\left|S\right|}$
is a $\text{dim}\left(\varphi\left(M,S\right)\right)$-dimensional
manifold with boundary. Any coordinate system in the affine hull of
$\varphi\left(M,S\right)$ can be a natural coordination of $\varphi\left(M,S\right)$.

Thus, $\varphi\left(M,S\right)$ as a manifold with boundary is a
perfect mathematical representation of the states of the system; in
fact, it is the most straightforward one, expressible directly in
empirical terms. This is however not the only one. Any mathematical
object can represent the state of the system that determines the system's
probabilistic behavior against all possible measurement operations.
For example, for our later purposes the convex decomposition 
\begin{equation}
\vec{Z}=\sum_{\vartheta\in\Theta}\lambda_{\vartheta}\vec{w}_{\vartheta}\,\,\,\,\,\,\,\lambda_{\vartheta}\geq0,\sum_{\vartheta\in\Theta}\lambda_{\vartheta}=1\label{eq:Z-composition}
\end{equation}
will be a more suitable characterization of a point of the state space.
However, in general, this decomposition is not unique. In fact there
are continuum many ways of such decomposition for all $\vec{Z}\in\text{Int}\,\varphi\left(M,S\right)$;
and a unique one if $\vec{Z}$ is on the boundary. As we will show,
there are various good solutions for obtaining a unique representation
of states in terms of their vertex decomposition (\ref{eq:Z-composition}).

Introduce the following notation: $\vec{\lambda}=\left(\lambda_{\vartheta}\right)_{\vartheta\in\Theta}\in\mathbb{R}^{\left|\Theta\right|}$.
Let
\[
\Lambda=\left\{ \vec{\lambda}\in\mathbb{R}^{\left|\Theta\right|}\left|\lambda_{\vartheta}\geq0,\sum_{\vartheta\in\Theta}\lambda_{\vartheta}=1\right.\right\} 
\]
$\Lambda$ is the $\left(\left|\Theta\right|-1\right)$-dimensional
standard simplex in $\mathbb{R}^{\left|\Theta\right|}$. Obviously,
\begin{equation}
D:\Lambda\rightarrow\varphi\left(M,S\right);\,\,\,\,D\left(\vec{\lambda}\right)=\sum_{\vartheta\in\Theta}\lambda_{\vartheta}\vec{w}_{\vartheta}\label{eq:nyalab-projekcio}
\end{equation}
is a continuous projection, and it preserves convex combination. 

\begin{lem}
\label{D-1_politop}For all $\vec{Z}\in\varphi\left(M,S\right)$,
$D^{-1}\left(\vec{Z}\right)$ is a polytope contained in $\Lambda$.
\end{lem}

\begin{proof}
To satisfy (\ref{eq:Z-composition}), beyond being contained in $\Lambda$,
$\vec{\lambda}$ has to satisfy the following system of linear equations:
\begin{eqnarray}
\sum_{\vartheta\in\Theta}\lambda_{\vartheta}w_{\vartheta}{}_{i}^{r} & = & Z_{i}^{r}\,\,\,\,\,\,\,\,\,\,\,\,\,\,\,\,\,\,\,\,\,\,\,\,\,{}_{i}^{r}\in I^{M}\label{eq:az_egyenlet}\\
\sum_{\vartheta\in\Theta}\lambda_{\vartheta}w_{\vartheta}{}_{i_{1}\ldots i_{L}}^{r_{1}\ldots r_{L}} & = & Z_{i_{1}\ldots i_{L}}^{r_{1}\ldots r_{L}}\,\,\,\,\,\,\,\,\,\,\,\,\,{}_{i_{1}\ldots i_{L}}^{r_{1}\ldots r_{L}}\in S\label{eq:az_egyenlet2}
\end{eqnarray}
For a given $\vec{Z}$, the set of solutions constitute an affine
subspace $\mathfrak{a}_{\vec{Z}}\subset\mathbb{R}^{\left|\Theta\right|}$
with difference space $\mathcal{B}\subset\mathbb{R}^{\left|\Theta\right|}$
constituted by the solutions of the homogeneous equations 
\begin{eqnarray*}
\sum_{\vartheta\in\Theta}\lambda_{\vartheta}w_{\vartheta}{}_{i}^{r} & = & 0\,\,\,\,\,\,\,\,\,\,\,\,\,{}_{i}^{r}\in I^{M}\\
\sum_{\vartheta\in\Theta}\lambda_{\vartheta}w_{\vartheta}{}_{i_{1}\ldots i_{L}}^{r_{1}\ldots r_{L}} & = & 0\,\,\,\,\,\,\,\,\,\,\,\,\,{}_{i_{1}\ldots i_{L}}^{r_{1}\ldots r_{L}}\in S
\end{eqnarray*}
Notice that $D^{-1}\left(\vec{Z}\right)=\Lambda\cap\mathfrak{a}_{\vec{Z}}$.
Due to the fact that an intersection of a polytope with an affine
subspace is a polytope (Henk \emph{et al}.~2004), each $D^{-1}\left(\vec{Z}\right)$
is a polytope contained in $\Lambda$.
\end{proof}
\begin{lem}
\label{D-1_folytonos}$D^{-1}\left(\vec{Z}\right)$, as a subset of
$\mathbb{R}^{\left|\Theta\right|}$, continuously depends on $\vec{Z}$
in the following sense:
\begin{eqnarray}
\underset{\vec{Z}'\rightarrow\vec{Z}}{\text{lim}}\,\,\,\underset{\vec{\lambda}\in D^{-1}\left(\vec{Z}\right)}{\text{max}}\,d\left(\vec{\lambda},D^{-1}\left(\vec{Z}'\right)\right) & = & 0\label{eq:egy-a-masiktol-1}\\
\underset{\vec{Z}'\rightarrow\vec{Z}}{\text{lim}}\,\,\,\underset{\vec{\lambda}\in D^{-1}\left(\vec{Z}'\right)}{\text{max}}\,d\left(\vec{\lambda},D^{-1}\left(\vec{Z}\right)\right) & = & 0\label{eq:masik-az-egyiktol-1-2}
\end{eqnarray}
where $d\left(\,,\,\right)$ denotes the usual distance of a point
from a set.
\end{lem}

\begin{proof}
We have to show that (\ref{eq:egy-a-masiktol-1})--(\ref{eq:masik-az-egyiktol-1-2})
hold approaching from all possible directions to $\vec{Z}$. In other
words, if $t\in[0,1]$ and $\varDelta\vec{Z}\in\mathbb{R}^{M+\left|S\right|}$
is an arbitrary non-zero vector such that $\vec{Z}-\varDelta\vec{Z}\in\varphi\left(M,S\right)$,
then
\begin{eqnarray}
\underset{t\rightarrow0}{\text{lim}}\,\,\,\underset{\vec{\lambda}\in D^{-1}\left(\vec{Z}\right)}{\text{max}}\,d\left(\vec{\lambda},D^{-1}\left(\vec{Z}-t\varDelta\vec{Z}\right)\right) & = & 0\label{eq:egy-a-masiktol-2}\\
\underset{t\rightarrow0}{\text{lim}}\,\,\,\underset{\vec{\lambda}\in D^{-1}\left(\vec{Z}-t\varDelta\vec{Z}\right)}{\text{max}}\,d\left(\vec{\lambda},D^{-1}\left(\vec{Z}\right)\right) & = & 0\label{eq:masik-az-egyiktol-1}
\end{eqnarray}

Let $\varDelta\vec{\lambda}$ be a solution of equations (\ref{eq:az_egyenlet})--(\ref{eq:az_egyenlet2})
with $\varDelta\vec{Z}$:
\begin{eqnarray}
\sum_{\vartheta\in\Theta}\varDelta\lambda_{\vartheta}w_{\vartheta}{}_{i}^{r} & = & \varDelta Z_{i}^{r}\,\,\,\,\,\,\,\,\,\,\,\,\,\,\,\,\,\,\,\,\,\,\,\,\,{}_{i}^{r}\in I^{M}\label{eq:deltaz1}\\
\sum_{\vartheta\in\Theta}\varDelta\lambda_{\vartheta}w_{\vartheta}{}_{i_{1}\ldots i_{L}}^{r_{1}\ldots r_{L}} & = & \varDelta Z_{i_{1}\ldots i_{L}}^{r_{1}\ldots r_{L}}\,\,\,\,\,\,\,\,\,\,\,\,\,{}_{i_{1}\ldots i_{L}}^{r_{1}\ldots r_{L}}\in S\label{eq:deltaz2}
\end{eqnarray}
$\varDelta\vec{\lambda}$ can be orthogonally decomposed as follows:
\[
\varDelta\vec{\lambda}=\varDelta\vec{\lambda}^{\bot}+\varDelta\vec{\lambda}^{\Vert}\,\,\,\,\,\,\,\,\,\varDelta\vec{\lambda}^{\bot}\in\mathcal{B}^{\bot}\text{ and }\varDelta\vec{\lambda}^{\Vert}\in\mathcal{B}
\]
Obviously, $\varDelta\vec{\lambda}^{\bot}$ is uniquely determined
by $\varDelta\vec{Z}$; accordingly, replacing $\varDelta\vec{Z}$
with $t\varDelta\vec{Z}$ on the right hand side of (\ref{eq:deltaz1})--(\ref{eq:deltaz2})
we get $t\varDelta\vec{\lambda}^{\bot}$ in place of $\varDelta\vec{\lambda}^{\bot}$.
Notice that $\left|t\varDelta\vec{\lambda}^{\bot}\right|$ is the
distance between the affine subspaces of solutions $\mathfrak{a}_{\vec{Z}}$
and $\mathfrak{a}_{\vec{Z}-t\varDelta\vec{Z}}$; tending to zero if
$t\rightarrow0$.

Let $\vec{\lambda}$ be an arbitrary point in $D^{-1}\left(\vec{Z}\right)$
and let $\vec{\lambda}'$ be the point in $D^{-1}\left(\vec{Z}-\varDelta\vec{Z}\right)$
closest to $\vec{\lambda}$, that is,

\begin{eqnarray*}
d\left(\vec{\lambda},D^{-1}\left(\vec{Z}-\varDelta\vec{Z}\right)\right) & = & \left|\vec{\lambda}'-\vec{\lambda}\right|
\end{eqnarray*}
Consider the point
\[
\vec{\lambda}_{t}=\vec{\lambda}+t\left(\vec{\lambda}'-\vec{\lambda}\right)
\]
Obviously, $\vec{\lambda}_{t}\in\Lambda$ and $\vec{\lambda}_{t}\in\mathfrak{a}_{\vec{Z}-t\varDelta\vec{Z}}$
for all $t\in[0,1]$, that is, 
\[
\vec{\lambda}_{t}\in D^{-1}\left(\vec{Z}-t\varDelta\vec{Z}\right)
\]
Therefore,
\[
d\left(\vec{\lambda},D^{-1}\left(\vec{Z}-t\varDelta\vec{Z}\right)\right)\leq t\left|\vec{\lambda}'-\vec{\lambda}\right|
\]
which implies (\ref{eq:egy-a-masiktol-2}).

Also, notice that 
\[
\underset{t\rightarrow0}{\text{lim}}\,\,\,\underset{\vec{\lambda}\in D^{-1}\left(\vec{Z}-t\varDelta\vec{Z}\right)}{\text{max}}\,d\left(\vec{\lambda},\mathfrak{a}_{\vec{Z}}\right)=0
\]
which implies (\ref{eq:masik-az-egyiktol-1}), otherwise there would
exist a convergent sequence of points from different $D^{-1}\left(\vec{Z}-t\varDelta\vec{Z}\right)$
sets such that the limiting point is not contained in $D^{-1}\left(\vec{Z}\right)$,
contradicting to the facts that $\Lambda$ is closed and $D^{-1}\left(\vec{Z}\right)=\Lambda\cap\mathfrak{a}_{\vec{Z}}$.
\end{proof}
Lemma~\ref{D-1_politop} and \ref{D-1_folytonos} mean that the states
of the system can be represented in a continuous way by a disjoint
family of polytopes contained in $\Lambda$.  This is of course a
very unusual and inconvenient way of representation. However, we can
easily make it more convenient by assigning a point in each $D^{-1}\left(\vec{Z}\right)$
representing the entire polytope. There are several possibilities:
for example, the center of mass, or any other notion of the center
of a polytope. Here we will use the notion of the point of maximal
entropy, which is perhaps physically also meaningful (Pitowsky~1989,
p.~47).

The point of maximal entropy of an arbitrary polytope $\mathcal{S}\subset\Lambda$:
\[
\vec{c}(\mathcal{S})=\,\begin{cases}
\text{maximize} & H\left(\vec{\lambda}\right)=-\sum_{\vartheta\in\Theta}\lambda_{\vartheta}\text{log}\lambda_{\vartheta}\\
\text{subject to} & \vec{\lambda}\in\mathcal{S}
\end{cases}
\]
Since $\mathcal{S}$ is contained in $\Lambda$, this maximization
problem always has a solution. Meaning that $\vec{c}(\mathcal{S})$
is uniquely determined and always contained in $\mathcal{S}$.
\begin{lem}
Let us define the following section of the bundle projection (\ref{eq:nyalab-projekcio}):
\begin{eqnarray}
 &  & \sigma:\varphi\left(M,S\right)\rightarrow\Lambda\nonumber \\
 &  & \,\,\,\,\,\,\,\,\,\sigma\left(\vec{Z}\right)=\vec{c}\left(D^{-1}\left(\vec{Z}\right)\right)\,\,\,\in D^{-1}\left(\vec{Z}\right)\label{eq:section}
\end{eqnarray}

Then, $\sigma\left(\vec{Z}\right)$ is continuous in $\vec{Z}$, that
is, for all $\vec{Z},\vec{Z}'\in\varphi\left(M,S\right)$,
\[
\underset{\vec{Z}'\rightarrow\vec{Z}}{\text{lim}}\,\sigma\left(\vec{Z'}\right)=\sigma\left(\vec{Z}\right)
\]

\end{lem}

\begin{proof}
Consider a sufficiently fine division of the unit cube $C^{\left|\Theta\right|}\subset\mathbb{R}^{\left|\Theta\right|}$
into equally sized small cubes of volume $\varDelta V$. Denote the
$i$-th such elementary cube by $C_{i}$. The point of maximal entropy
of a polytope $\mathcal{S}\subset\Lambda\subset C^{\left|\Theta\right|}$
can be approximated with arbitrary precision in the following way:
\begin{equation}
\vec{c}\left(\mathcal{S}\right)\simeq\,\begin{cases}
\text{maximize} & H\left(^{i}\vec{\lambda}\right)=-\sum_{\vartheta\in\Theta}{}^{i}\lambda_{\vartheta}\text{log}{}^{i}\lambda_{\vartheta}\\
\text{subject to} & i\in\left\{ j\,|\,\mathcal{S}\cap C_{j}\neq\textrm{Ø}\right\} 
\end{cases}\label{eq:approximation}
\end{equation}
where $^{i}\vec{\lambda}$ is, say, the center of $C_{i}$. Due to
Lemma~\ref{D-1_folytonos}, for all $\varDelta V>0$ there is an
$\varepsilon>0$ such that, for all elementary cube $C_{i}$,
\[
D^{-1}\left(\vec{Z}'\right)\cap C_{i}\neq\textrm{Ø}\,\,\Leftrightarrow\,\,D^{-1}\left(\vec{Z}\right)\cap C_{i}\neq\textrm{Ø}\text{\,\,\,\,\,\,\, if \ensuremath{\left|\vec{Z}'-\vec{Z}\right|<\varepsilon}}
\]
Meaning that, for a sufficiently small $\varepsilon$, approximation
(\ref{eq:approximation}) leads to the same result for $D^{-1}\left(\vec{Z}'\right)$
and $D^{-1}\left(\vec{Z}\right)$. Therefore,
\[
\underset{\vec{Z}'\rightarrow\vec{Z}}{\text{lim}}\,\vec{c}\left(D^{-1}\left(\vec{Z}'\right)\right)=\vec{c}\left(D^{-1}\left(\vec{Z}\right)\right)
\]
\end{proof}
By means of $\sigma$ (or any similar continuous section) the whole
state space $\varphi\left(M,S\right)$ can be lifted into a $\text{dim}\left(\varphi\left(M,S\right)\right)$-dimensional
submanifold with boundary: 
\begin{equation}
\Lambda_{\sigma}=\sigma\left(\varphi\left(M,S\right)\right)\subset\Lambda\label{eq:Lambda0}
\end{equation}

\section{Dynamics\protect\label{sec:Dynamics}}

So far, nothing has been said about the dynamics of the system, that
is, about the time evolution of the state $\vec{Z}$. First we have
to introduce the concept of time evolution in general operational
terms. Let us start with the most general case.

Imagine that the system is in state $\vec{Z}\left(t_{0}\right)$ after
a certain physical preparation at time $t_{0}$. According to the
definition of state, this means that the system responds to the various
measurement operations right after time $t_{0}$ in a way described
in (\ref{eq:pi_is_such-1})--(\ref{eq:pi_is_such-2}). Let then the
system evolve under a given set of circumstances until time $t$.
Let $\vec{Z}\left(t\right)$ be the system's state at moment $t$.
Again, this means that the system responds to the various measurement
operations right after time $t$ in a way described in (\ref{eq:pi_is_such-1})--(\ref{eq:pi_is_such-2})
with $\vec{Z}\left(t\right)$. Thus, we have a temporal path of the
system in the space of states $\varphi\left(M,S\right)$. 

By means of a continuous cross section like (\ref{eq:section}), of
course, $\vec{Z}\left(t\right)$ can be lifted into $\Lambda_{\sigma}$
and expressed as a curve $\sigma\left(\vec{Z}\left(t\right)\right)$
on $\Lambda_{\sigma}$.

At this level of generality, we say nothing about the temporal path
$\vec{Z}\left(t\right)$. Whether it has some specific feature or
the time evolution of the system shows any regularity whatsoever,
is a matter of empirical facts reflected in the observed relative
frequencies under various circumstances. As a possible empirically
observed such regularity, we formulate a typical situation when the
time evolution $\vec{Z}\left(t\right)$ can be generated by a one-parameter
group of transformations of $\varphi\left(M,S\right)$.
\begin{lyxlist}{0.0.0}
\item [{\textbf{(E4)}}] The time evolutions of states are such that there
exists a one-parameter group of transformations of $\varphi\left(M,S\right)$,
$F_{t}$, satisfying the following conditions:
\begin{itemize}
\item [] $F_{t}:\varphi\left(M,S\right)\rightarrow\varphi\left(M,S\right)$
is one-to-one
\item [] $F:\mathbb{R}\times\varphi\left(M,S\right)\rightarrow\varphi\left(M,S\right);\,\left(t,\vec{Z}\right)\mapsto F_{t}\left(\vec{Z}\right)$
is continuous
\item [] $F_{t+s}=F_{s}\circ F_{t}$
\item [] $F_{-t}=F_{t}^{-1}$; consequently, $F_{0}=id_{\varphi\left(M,S\right)}$
\end{itemize}
and the time evolution of an arbitrary initial state $\vec{Z}(t_{0})\in\varphi\left(M,S\right)$
is $\vec{Z}(t)=F_{t-t_{0}}\left(\vec{Z}(t_{0})\right)$.
\end{lyxlist}
It is worth mentioning that although the state space $\varphi\left(M,S\right)$
is closed under convex combination, and, in some cases, the subset
$\varPhi\subseteq\varphi\left(M,S\right)$ of the actually observable
states may be closed under convex combination, the stipulated empirical
facts \textbf{(E1)}--\textbf{(E3)} do not imply that the time evolution
should preserve convex combinations. That is, 
\[
\vec{Z}_{3}\left(t_{0}\right)=\lambda_{1}\vec{Z}_{1}\left(t_{0}\right)+\lambda_{2}\vec{Z}_{2}\left(t_{0}\right)\,\,\,\,\,\,\,\,\,\,\,\lambda_{1},\lambda_{2}\geq0;\lambda_{1}+\lambda_{2}=1
\]
at time $t_{0}$ generally does not imply that 
\[
\vec{Z}_{3}\left(t\right)=\lambda_{1}\vec{Z}_{1}\left(t\right)+\lambda_{2}\vec{Z}_{2}\left(t\right)
\]
at any other moment of time $t$. This does not follow even if the
time evolution additionally satisfies condition \textbf{(E4)}.

As an illustration, consider our previously discussed example with
the coin that can be prepared into different biased states, shown
on the left hand side of Fig.~\ref{coin}, with the following modification.
Imagine that the disk is mounted on a threaded rod that rotates in
a given direction at a constant angular velocity. (Fig.~\ref{fig:csavar})
The threading of the rod is not uniform, but becomes denser and denser
towards the two ends, and the density of threads tends to infinity
as the two extreme disk positions are approached. Due to the rotation
of the rod, the disk, when placed in an arbitrary initial position,
moves upwards with velocity determined by the threading. 
\begin{figure}
\begin{centering}
\includegraphics[width=0.5\columnwidth]{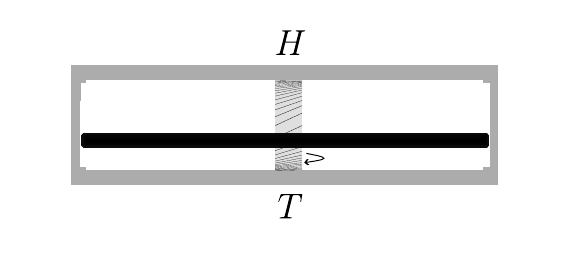}
\par\end{centering}
\caption{Moving disk}
\label{fig:csavar}
\end{figure}
 Meaning, that the coin continuously evolves from being maximally
biased for Heads towards being maximally biased for Tails. One can
imagine continuously many ways of such threading. For example, let
us assume that the resulted change of probabilities of Heads and Tails,
$F_{t}:\left(z_{H},z_{T},0\right)\mapsto\left(F_{t}\left(z_{H}\right),F_{t}\left(z_{T}\right),0\right)$,
is something like this:

\begin{eqnarray}
F_{t}(z_{H}) & = & 0.8-0.6\frac{\frac{\left(0.8-z_{H}\right)}{0.6}}{\frac{\left(0.8-z_{H}\right)}{0.6}+\left(1-\frac{\left(0.8-z_{H}\right)}{0.6}\right)\exp\left(-t\right)}\label{eq:F_t-formula1}\\
F_{t}(z_{T}) & = & 1-F_{t}(z_{H})\label{eq:F_t-formula2}
\end{eqnarray}
This is a physically entirely possible dynamics, and it satisfies
\textbf{(E4)}. But, as it must be obvious from the formula (\ref{eq:F_t-formula1})
itself, it does not preserve convex combination. (For a concrete numeric
example, see the end of Appendix~3.)

Thus, in our general operational-probabilistic framework based on
the assumptions \textbf{(E1)}--\textbf{(E3)}, even if \textbf{(E4)}
is satisfied, the time-evolution of states does not necessarily preserve
convex combination; and, as the above simple example suggests, there
is no reason to assume that experience shows such a restriction on
the possible dynamics. (This is in stark contrast to stipulations
in the GPT framework; see Appendix~3.)

By means of the continuous cross section (\ref{eq:section}), $F_{t}$
generates a one-parameter group of transformations on $\Lambda_{\sigma}$,
$K_{t}=\sigma\circ F_{t}\circ D$, with exactly the same properties:
\begin{itemize}
\item [] $K_{t}:\Lambda_{\sigma}\rightarrow\Lambda_{\sigma}$ is one-to-one
\item [] $K:\mathbb{R}\times\Lambda_{\sigma}\rightarrow\Lambda_{\sigma};\,\left(t,\vec{Z}\right)\mapsto K_{t}\left(\vec{Z}\right)$
is continuous
\item [] $K_{t+s}=K_{s}\circ K_{t}$
\item [] $K_{-t}=K_{t}^{-1}$; consequently, $K_{0}=id_{\Lambda_{\sigma}}$
\end{itemize}
Despite all the mathematical attractiveness of \textbf{(E4)}, and
despite the fact that it is satisfied for many physical systems, \textbf{(E4)}
is not taken to be satisfied by the time-evolution of a system in
general, and so will not be assumed in the remainder of this paper---with
some rare exceptions where we will indicate this. The main reason
is that \textbf{(E4)} is too strong a requirement, the violation of
which doesn't mean that the system cannot have a meaningful time-evolution.
For instance, consider the example depicted in Fig.~\ref{fig:csavar}
with the modification that the density of threads does not go to infinity
at the ends, therefore the disk reaches its extreme positions in finite
time, but there is a mechanism which changes the direction of the
rod's rotation whenever the disk reaches an endpoint. In this way,
the state $\left(z_{H},z_{T},0\right)$ will continuously oscillate
between the two extremes $(0.8,0.2,0)$ and $(0.2,0.8,0)$. Consequently,
the state---in the operational--probabilistic sense---does not
in itself determine its subsequent value, since that will depend on
whether it is in the downward or upward period. So, \textbf{(E4)}
is not satisfied because there does not exist such an $F_{t}:\varphi\left(M,S\right)\rightarrow\varphi\left(M,S\right)$
function. Yet, it is worth emphasizing that there is a definite regularity
according to which temporal development takes place in the underlying
ontology. This underlying temporal development is, in the above example,
deterministic and Markovian (cf. Barandes~2023). The only lesson
is that the operational--probabilistic notion of state is conceptually
different from the notion of Cauchy data for the underlying dynamics.

\section{Ontology \protect\label{sec:Hidden}}

So, at a given time instant, the operational--probabilistic state
fully characterizes the probabilistic behavior of the system with
respect to all possible measurements at that time instant---according
to Theorem~\ref{State--together}. In general, however, such a probabilistic
state admits different underlying ontological pictures even at the
given time instant. Though, as we will see, some of those underlying
ontologies imply further conditions on the observed relative frequencies.
We will mention three important cases, but various combinations are
conceivable.

\paragraph*{Case~1}

In the most general case, without any further restriction on the observed
relative frequencies, the outcomes of the measurements are random
events produced in the measurement process itself. The state $\vec{Z}$
characterizes the system in a dispositional sense: the system has
a propensity to behave in a certain way, that is, to produce a certain
statistics of outcomes, if a given combination of measurements is
performed. In general, the produced statistics is such that, for example,
\begin{equation}
p\left(X_{i}^{r}|a_{r}\wedge a_{r'}\right)\neq p\left(X_{i}^{r}|a_{r}\right)\,\,\,\,\,\,\,\,\,\,\,\,\,\left\{ r,r'\right\} \in\mathfrak{P}\label{eq:kontex}
\end{equation}
meaning that the underlying process is ``contextual'' in the sense
that the system's statistical behavior against measurement $a_{r}$
can be influenced by the performance of another measurement $a_{r'}$.

\paragraph*{Case~2}

In the second case we assume that there is no such cross-influence
in the underlying ontology. That is, the observed relative frequencies
satisfy the following general condition: 
\begin{eqnarray}
 &  & p\left(X_{i_{1}}^{r_{1}}\wedge\ldots\wedge X_{i_{L}}^{r_{L}}|a_{r_{1}}\wedge\ldots\wedge a_{r_{L}}\wedge a_{r'_{1}}\wedge\ldots\wedge a_{r'_{L'}}\right)\nonumber \\
 &  & \,\,\,\,\,\,\,\,\,\,\,\,\,\,\,\,\,\,\,\,\,\,\,\,=p\left(X_{i_{1}}^{r_{1}}\wedge\ldots\wedge X_{i_{L}}^{r_{L}}|a_{r_{1}}\wedge\ldots\wedge a_{r_{L}}\right)\label{eq:2'-1}
\end{eqnarray}
for all $L,L'$, $2\leq L+L'\leq m$, $_{i_{1}\ldots i_{L}}^{r_{1}\ldots r_{L}}\in S$,
and $\left\{ r_{1},\ldots r_{L},r'_{1},\ldots r'_{L'}\right\} \in\mathfrak{P}$.
This does not mean that there cannot be correlation between the outcomes
$X_{i_{1}}^{r_{1}}\wedge\ldots\wedge X_{i_{L}}^{r_{L}}$ and the performance
of measurement $a_{r'_{1}}\wedge\ldots\wedge a_{r'_{L'}}$ It only
means that the correlation must be the consequence of the fact that
the measurement operations $a_{r'_{1}}\wedge\ldots\wedge a_{r'_{L'}}$
and $a_{r_{1}}\wedge\ldots\wedge a_{r_{L}}$ are correlated; $a_{r_{1}}\wedge\ldots\wedge a_{r_{L}}$
must be the common cause responsible for the correlation. Indeed,
(\ref{eq:2'-1}) is equivalent with the following ``screening off''
condition: 
\begin{eqnarray}
 &  & p\left(a_{r'_{1}}\wedge\ldots\wedge a_{r'_{L'}}\wedge X_{i_{1}}^{r_{1}}\wedge\ldots\wedge X_{i_{L}}^{r_{L}}|a_{r_{1}}\wedge\ldots\wedge a_{r_{L}}\right)\nonumber \\
 &  & \,\,\,\,\,\,\,\,\,\,\,\,\,\,\,\,\,\,\,\,\,\,\,\,\,\,\,\,\,\,\,\,=p\left(a_{r'_{1}}\wedge\ldots\wedge a_{r'_{L'}}|a_{r_{1}}\wedge\ldots\wedge a_{r_{L}}\right)\nonumber \\
 &  & \,\,\,\,\,\,\,\,\,\,\,\,\,\,\,\,\,\,\,\,\,\,\,\,\,\,\,\,\,\,\,\,\,\,\,\,\,\times\,p\left(X_{i_{1}}^{r_{1}}\wedge\ldots\wedge X_{i_{L}}^{r_{L}}|a_{r_{1}}\wedge\ldots\wedge a_{r_{L}}\right)\label{eq:2-1}
\end{eqnarray}
for all $L,L'$, $2\leq L+L'\leq m$, $_{i_{1}\ldots i_{L}}^{r_{1}\ldots r_{L}}\in S$,
and $\left\{ r_{1},\ldots r_{L},r'_{1},\ldots r'_{L'}\right\} \in\mathfrak{P}$.

All this means that the state of the system $\vec{Z}$ reflects the
propensities of the system to produce a certain statistics of outcomes
against each possible measurement/measurement combination, separately.
The observed statistics reveals the propensity in question, but, in
general, we are not entitled to say that a single outcome (of a measurement/measurement
combination) reveals an element of reality existing independently
of the measurement(s). As we will see below, that would require a
stronger restriction on the observed frequencies.

\paragraph*{Case~3}

Assume that the underlying ontology contains such elements of reality.
Let us denote them by $\#X_{i}^{r}$ ($_{i}^{r}\in I^{M}$). More
precisely, let $\#X_{i}^{r}$ denote the event that the element of
reality revealed in the outcome $X_{i}^{r}$ is present in the given
run of the experiment. Certainly, every such event $\#X_{i}^{r}$,
even if hidden to us, must have some relative frequency. That is to
say, there must exists a relative frequency function $p'$ on the
extended free Boolean algebra $\mathcal{A}'$ generated by the set
\begin{equation}
G'=\left\{ a_{r}\right\} _{r=1,2,\ldots m}\cup\left\{ X_{i}^{r}\right\} _{_{i}^{r}\in I^{M}}\cup\left\{ \#X_{j}^{s}\right\} _{_{j}^{s}\in I^{M}}\label{eq:generator-1}
\end{equation}
such that
\begin{eqnarray}
p'\,\bigl|_{\mathcal{A}\subset\mathcal{A}'} & = & p\label{eq:egybeesik}
\end{eqnarray}

The ontological assumption that $\#X_{i}^{r}$ is revealed by the
measurement outcome $X_{i}^{r}$ means that 
\begin{eqnarray}
p'\left(X_{i}^{r}|a_{r}\wedge\#X_{i}^{r}\right) & = & 1\label{eq:element1}\\
p'\left(X_{i}^{r}|a_{r}\wedge\neg\#X_{i}^{r}\right) & = & 0\\
p'\left(a_{r}\wedge\#X_{i}^{r}\right) & = & p'\left(a_{r}\right)p'\left(\#X_{i}^{r}\right)
\end{eqnarray}
 Similarly, 
\begin{eqnarray}
p'\left(X_{i_{1}}^{r_{1}}\wedge\ldots\wedge X_{i_{L}}^{r_{L}}|a_{r_{1}}\wedge\ldots\wedge a_{r_{L}}\wedge\#X_{i_{1}}^{r_{1}}\wedge\ldots\wedge\#X_{i_{L}}^{r_{L}}\right)=1\\
p'\left(X_{i_{1}}^{r_{1}}\wedge\ldots\wedge X_{i_{L}}^{r_{L}}|a_{r_{1}}\wedge\ldots\wedge a_{r_{L}}\wedge\neg\left(\#X_{i_{1}}^{r_{1}}\wedge\ldots\wedge\#X_{i_{L}}^{r_{L}}\right)\right)=0\\
p'\left(a_{r_{1}}\wedge\ldots\wedge a_{r_{L}}\wedge\#X_{i_{1}}^{r_{1}}\wedge\ldots\wedge\#X_{i_{L}}^{r_{L}}\right)=p'\left(a_{r_{1}}\wedge\ldots\wedge a_{r_{L}}\right)\nonumber \\
\times\,p'\left(\#X_{i_{1}}^{r_{1}}\wedge\ldots\wedge\#X_{i_{L}}^{r_{L}}\right)\label{eq:element6}
\end{eqnarray}
for all $_{i_{1}\ldots i_{L}}^{r_{1}\ldots r_{L}}\in S$.

Now, (\ref{eq:element1})--(\ref{eq:element6}) and (\ref{eq:egybeesik})
imply that 
\begin{eqnarray}
p'\left(\#X_{i}^{r}\right) & = & p\left(X_{i}^{r}|a_{r}\right)=Z_{i}^{r}\label{eq:tul1}\\
p'\left(\#X_{i_{1}}^{r_{1}}\wedge\ldots\wedge\#X_{i_{L}}^{r_{L}}\right) & = & p\left(X_{i_{1}}^{r_{1}}\wedge\ldots\wedge X_{i_{L}}^{r_{L}}|a_{r_{1}}\wedge\ldots\wedge a_{r_{L}}\right)\nonumber \\
 &  & =Z_{i_{1}\ldots i_{L}}^{r_{1}\ldots r_{L}}\label{eq:tul2}
\end{eqnarray}
for all $_{i}^{r}\in I^{M}$ and $_{i_{1}\ldots i_{L}}^{r_{1}\ldots r_{L}}\in S$.

Notice that on the right hand side of (\ref{eq:tul1})--(\ref{eq:tul2})
we have the components of $\vec{Z}$. At the same time, on the left
hand side of (\ref{eq:tul1})--(\ref{eq:tul2}) we have numbers that
are values of relative frequencies. Therefore the components of $\vec{Z}$
must constitute values of relative frequencies (of the occurrences
of elements of reality $\#X_{i}^{r}$ and $\#X_{i_{1}}^{r_{1}}\ldots\wedge\#X_{i_{L}}^{r_{L}}$).
Since values of relative frequencies satisfy the Kolmogorovian laws
of classical probability, $\vec{Z}$ must be in the so-called classical
correlation polytope (Pitowsky~1989, Ch.~2):
\begin{equation}
\vec{Z}\in c\left(M,S\right)\label{eq:Z*in-1}
\end{equation}
(Equivalently, the components of $\vec{Z}$ must satisfy the corresponding
Bell-type inequalities.) In this case the physical state of the system
admits a more fine-grained characterization than the probabilistic
description provided by $\vec{Z}$: in each run of the experiment
the system can be thought of as being in an underlying physical state
(fixing whether the elements of reality $\#X_{i}^{r}$ and $\#X_{i_{1}}^{r_{1}}\ldots\wedge\#X_{i_{L}}^{r_{L}}$
are present or not) that predetermines the outcome of every possible
measurement, given that the measurement in question is performed.

\medskip{}
Thus, as we have seen from the above examples, the probabilistic--operational
notion of state admits different underlying ontologies, depending
on whether some further conditions are met or not. Note that condition
(\ref{eq:2'-1}) in Case~2 is sometimes called ``no-signaling condition'';
and Case~3 is usually interpreted as ``admitting deterministic non-contextual
hidden variables''. In what follows, we do not assume anything more
about the observed relative frequencies than we stipulated in \textbf{(E1)}--\textbf{(E3)}.
Meaning that we remain within the most general framework of Case~1.

\section{Quantum Representation\protect\label{sec:Quantum-Representation}}

So far in the previous sections, we have stayed within the framework
of classical Kolmogorovian probability theory; including the notion
of state, which is a simple vector constructed from classical conditional
probabilities. Meaning that any physical system---traditionally categorized
as classical or quantum, or ``more general than quantum''---that
can be described in operational terms can be described within classical
Kolmogorovian probability theory. It is worth pointing out that this
is also the case when the system is traditionally described in terms
of the Hilbert space quantum mechanical formalism. That is, all the
empirically expressible content of the quantum mechanical description
can be described in the language of classical Kolmogorovian probabilities;
including what we refer to as ``quantum probability'', given by
the usual trace formula, which can be expressed simply as classical
conditional probability. All this is in perfect alignment with the
Kolmogorovian Censorship Hypothesis.

In the remainder of the paper we will show that the opposite is also
true: anything that can be described in operational terms can be represented
in the Hilbert space quantum mechanical formalism, if we wish. From
assumptions \textbf{(E1)}-\textbf{(E3)} alone, we will show that there
always exists:
\begin{lyxlist}{00.00.0000}
\item [{(Q1)}] a suitable Hilbert space, such that
\item [{(Q2)}] the outcomes of each measurement can be represented by a
system of pairwise orthogonal closed subspaces, spanning the whole
Hilbert space,
\item [{(Q3)}] the states of the system can be represented by pure state
operators with suitable state vectors, and
\item [{(Q4)}] the probabilities of the measurement outcomes, with arbitrarily
high precision, can be reproduced by the usual trace formula of quantum
mechanics.
\end{lyxlist}
Moreover, in the case of real-valued quantities,
\begin{lyxlist}{00.00.0000}
\item [{(Q5)}] each quantity, if we wish, can be associated with a suitable
self-adjoint operator, such that
\item [{(Q6)}] in all states of the system, the expectation value of the
quantity can be reproduced, with arbitrarily high precision, by the
usual trace formula applied to the associated self-adjoint operator,
\item [{(Q7)}] the possible measurement results are the eigenvalues of
the operator,
\item [{(Q8)}] and the corresponding outcome events are represented by
the eigenspaces pertaining to the eigenvalues respectively, according
to the spectral decomposition of the operator in question.
\end{lyxlist}
In preparation for our quantum representation theorem, first we prove
a lemma, which is a straightforward consequence of previous results
in Pitowsky's \emph{Quantum Probability -- Quantum Logic.}
\begin{lem}
\label{Lemma:for=000020all=000020f} For each vector $\vec{f}\in l\left(M,S\right)$
there exists a Hilbert space $^{(\vec{f})}H$ and closed subspaces
$^{(\vec{f})}E_{i}^{r}$ in the subspace lattice $L\left(^{(\vec{f})}H\right)$
and a pure state $P_{\Psi_{\vec{f}}}$ with a suitable unit vector
$\Psi_{\vec{f}}\in{}^{(\vec{f})}H$, such that 
\begin{eqnarray}
f_{i}^{r} & \simeq & tr\left(P_{\Psi_{\vec{f}}}{}^{(\vec{f})}E_{i}^{r}\right)\label{eq:approximate1}\\
f_{i_{1}\ldots i_{L}}^{r_{1}\ldots r_{L}} & \simeq & tr\left(P_{\Psi_{\vec{f}}}\left(^{(\vec{f})}E_{i_{1}}^{r_{1}}\wedge\ldots\wedge^{(\vec{f})}E_{i_{L}}^{r_{L}}\right)\right)\label{eq:approximate2}
\end{eqnarray}
for all $_{i}^{r}\in I^{M}$ and $_{i_{1}\ldots i_{L}}^{r_{1}\ldots r_{L}}\in S$.
\end{lem}

\begin{proof}
It follows from a straightforward generalization of a theorem in (Pitowsky~1989,
p.~65) that the so called quantum polytope $q\left(M,S\right)$,
constituted by the vectors satisfying (\ref{eq:approximate1})--(\ref{eq:approximate2})
with exact equality, is a dense convex subset of $l\left(M,S\right)$;
it is essentially $l\left(M,S\right)$ save for some points on the
boundary of $l\left(M,S\right)$, namely the finite number of non-classical
vertices. $q\left(M,S\right)$ contains the interior of $l\left(M,S\right)$.
This means that arbitrary vector $\vec{f}\in l\left(M,S\right)$ can
be regarded as ``an element of'' $q\left(M,S\right)$ \emph{with
arbitrary precision}. That is, there exists a Hilbert space $^{(\vec{f})}H$
and for each $_{i}^{r}\in I^{M}$ a closed subspace/projector $^{(\vec{f})}E_{i}^{r}$
in the subspace/projector lattice $L\left(^{(\vec{f})}H\right)$ and
a suitable unit vector $\Psi_{\vec{f}}\in{}^{(\vec{f})}H$, such that
the approximate equalities (\ref{eq:approximate1})--(\ref{eq:approximate2})
hold.
\end{proof}
\begin{thm}
\label{Theorem:Main}There exists a Hilbert space $H$ and for each
outcome event $X_{i}^{r}$ a closed subspace/projector $E_{i}^{r}$
in the subspace/projector lattice $L\left(H\right)$, such that for
each state $\vec{Z}$ of the system there exists a pure state $P_{\Psi_{\vec{Z}}}$
with a suitable unit vector $\Psi_{\vec{Z}}\in H$, such that 
\begin{eqnarray}
Z_{i}^{r} & \simeq & tr\left(P_{\Psi_{\vec{Z}}}E_{i}^{r}\right)\label{eq:approximate-tetel-1}\\
Z_{i_{1}\ldots i_{L}}^{r_{1}\ldots r_{L}} & \simeq & tr\left(P_{\Psi_{\vec{Z}}}\left(E_{i_{1}}^{r_{1}}\wedge\ldots\wedge E_{i_{L}}^{r_{L}}\right)\right)\label{eq:approximate-tetel-2}
\end{eqnarray}
and 
\begin{eqnarray}
 & E_{i}^{r}\bot\,E_{j}^{r} & \,\,\,\,\,\,i\neq j\label{eq:orthogonal}\\
 & \stackrel[k=1]{n_{r}}{\vee}E_{k}^{r}=H\label{eq:kiadja}
\end{eqnarray}
for all $_{i}^{r},{}_{j}^{r}\in I^{M}$ and $_{i_{1}\ldots i_{L}}^{r_{1}\ldots r_{L}}\in S$.
\end{thm}

\begin{proof}
The proof is essentially based on Lemma~\ref{The-vertices-of-fi}
and proceeds in two major steps.

\subsubsection*{Step~I}

Consider the vertices of $\varphi\left(M,S\right)$, $\left\{ \vec{w}_{\vartheta}\right\} _{\vartheta\in\Theta}$.
Each $\vec{w}_{\vartheta}$ is a vector in $l\left(M,S\right).$ Therefore,
due to Lemma~\ref{Lemma:for=000020all=000020f}, for each $\vec{w}_{\vartheta}$
there exists a Hilbert space $^{\vartheta}\tilde{H}$ and closed subspaces
$^{\vartheta}\tilde{E}_{i}^{r}$ in the subspace lattice $L\left(^{\vartheta}\tilde{H}\right)$
and a pure state $P_{\tilde{\Psi}_{\vartheta}}$ with a suitable unit
vector $\tilde{\Psi}_{\vartheta}\in{}^{\vartheta}\tilde{H}$, such
that 
\begin{eqnarray}
w_{\vartheta}{}_{i}^{r} & \simeq & tr\left(P_{\tilde{\Psi}_{\vartheta}}{}^{\vartheta}\tilde{E}_{i}^{r}\right)\label{eq:approximate1-1}\\
w_{\vartheta}{}_{i_{1}\ldots i_{L}}^{r_{1}\ldots r_{L}} & \simeq & tr\left(P_{\tilde{\Psi}_{\vartheta}}\left(^{\vartheta}\tilde{E}_{i_{1}}^{r_{1}}\wedge\ldots\wedge{}^{\vartheta}\tilde{E}_{i_{L}}^{r_{L}}\right)\right)\label{eq:approximate2-1}
\end{eqnarray}
for all $_{i}^{r}\in I^{M}$ and $_{i_{1}\ldots i_{L}}^{r_{1}\ldots r_{L}}\in S$.

Now, let 
\begin{eqnarray}
^{\vartheta}H & = & H^{n_{1}}\otimes H^{n_{2}}\otimes\ldots\otimes H^{n_{m}}\otimes{}^{\vartheta}\tilde{H}\label{eq:large_H}
\end{eqnarray}
where $H^{n_{1}},H^{n_{2}},\ldots H^{n_{m}}$ are Hilbert spaces of
dimension $n_{1},n_{2},\ldots n_{m}$. Let $e_{1}^{r},e_{2}^{r},\ldots e_{n_{r}}^{r}$
be an orthonormal basis in $H^{n_{r}}$. Define the corresponding
subspace for each event $X_{i}^{r}$ as follows:
\begin{equation}
^{\vartheta}E_{i}^{r}=H^{n_{1}}\otimes\ldots H^{n_{r-1}}\otimes\left[e_{i}^{r}\right]\otimes H^{n_{r+1}}\ldots\otimes H^{n_{m}}\otimes{}^{\vartheta}\tilde{E}_{i}^{r}\label{eq:E_ri}
\end{equation}
where $\left[e_{i}^{r}\right]$ stands for the one-dimensional subspace
spanned by $e_{i}^{r}$ in $H^{n_{r}}$. Notice that, for all $r$,
\begin{equation}
^{\vartheta}E_{i}^{r}\bot{}^{\vartheta}E_{j}^{r}\,\,\,\,\,\,\,\,\,\,\text{ if \,\,\,\,}i\neq j\label{eq:orthogonality-teta}
\end{equation}
due to the fact that $e_{1}^{r},e_{2}^{r},\ldots e_{n_{r}}^{r}$ is
an orthonormal basis in $H^{n_{r}}$.

Due to Lemma~\ref{The-vertices-of-fi}, for all $1\leq r\leq m$
there is exactly one $1\leq{}^{\vartheta}i_{*}^{r}\leq n_{r}$ such
that $w_{\vartheta}{}_{^{\vartheta}i_{*}^{r}}^{r}=1$. This makes
it possible to define the state vector in $^{\vartheta}H$ as the
following unit vector:
\[
\Psi_{\vartheta}=e_{^{\vartheta}i_{*}^{1}}^{1}\otimes e_{^{\vartheta}i_{*}^{2}}^{2}\otimes\ldots\otimes e_{^{\vartheta}i_{*}^{r}}^{r}\otimes\ldots\otimes e_{^{\vartheta}i_{*}^{m}}^{m}\otimes\tilde{\Psi}_{\vartheta}
\]

Now, it is easily verifiable that 
\begin{eqnarray}
w_{\vartheta}{}_{i}^{r} & \simeq & tr\left(P_{\Psi_{\vartheta}}{}^{\vartheta}E_{i}^{r}\right)\label{eq:approximate-teta-1}\\
w_{\vartheta}{}_{i_{1}\ldots i_{L}}^{r_{1}\ldots r_{L}} & \simeq & tr\left(P_{\Psi_{\vartheta}}\left(^{\vartheta}E_{i_{1}}^{r_{1}}\wedge\ldots\wedge{}^{\vartheta}E_{i_{L}}^{r_{L}}\right)\right)\label{eq:approximate-teta-2}
\end{eqnarray}
for all $_{i}^{r}\in I^{M}$, $_{i_{1}\ldots i_{L}}^{r_{1}\ldots r_{L}}\in S$,
and for all $\vartheta\in\Theta$. For example:

\noindent If $w_{\vartheta}{}_{i}^{r}=1$, and so $i=^{\vartheta}i_{*}^{r}$,
then
\begin{eqnarray*}
tr\left(P_{\Psi_{\vartheta}}{}^{\vartheta}E_{i}^{r}\right)=\underbrace{tr\left(P_{e_{^{\vartheta}i_{*}^{1}}^{1}}H^{n_{1}}\right)}_{1}tr\left(P_{e_{^{\vartheta}i_{*}^{2}}^{2}}H^{n_{2}}\right)\ldots\underbrace{tr\left(P_{e_{^{\vartheta}i_{*}^{r}}^{r}}\left[e_{^{\vartheta}i_{*}^{r}}^{r}\right]\right)}_{1}\ldots\\
tr\left(P_{e_{^{\vartheta}i_{*}^{m}}^{m}}H^{n_{m}}\right)\underbrace{tr\left(P_{\tilde{\Psi}_{\vartheta}}{}^{\vartheta}\tilde{E}_{i}^{r}\right)}_{\,\,\,\,\,\,\,\,\,\,\simeq\,w_{\vartheta}{}_{i}^{r}=1}\simeq1
\end{eqnarray*}

\noindent If $w_{\vartheta}{}_{i}^{r}=0$, and so $i\neq{}^{\vartheta}i_{*}^{r}$,
then
\begin{eqnarray*}
tr\left(P_{\Psi_{\vartheta}}{}^{\vartheta}E_{i}^{r}\right)=\underbrace{tr\left(P_{e_{^{\vartheta}i_{*}^{1}}^{1}}H^{n_{1}}\right)}_{1}tr\left(P_{e_{^{\vartheta}i_{*}^{2}}^{2}}H^{n_{2}}\right)\ldots\underbrace{tr\left(P_{e_{^{\vartheta}i_{*}^{r}}^{r}}\left[e_{i\neq{}^{\vartheta}i_{*}^{r}}^{r}\right]\right)}_{0}\ldots\\
tr\left(P_{e_{^{\vartheta}i_{*}^{m}}^{m}}H^{n_{m}}\right)\underbrace{tr\left(P_{\tilde{\Psi}_{\vartheta}}{}^{\vartheta}\tilde{E}_{i}^{r}\right)}_{\,\,\,\,\,\,\,\,\,\,\,\simeq\,w_{\vartheta}{}_{i}^{r}=0}=0
\end{eqnarray*}

\noindent Similarly, if $w_{\vartheta}{}_{i_{1}}^{r_{1}}=0$, $w_{\vartheta}{}_{i_{2}}^{r_{2}}=1$
then
\begin{eqnarray*}
tr\left(P_{\Psi_{\vartheta}}\left(^{\vartheta}E_{i_{1}}^{r_{1}}\wedge{}^{\vartheta}E_{i_{2}}^{r_{2}}\right)\right)=\underbrace{tr\left(P_{e_{^{\vartheta}i_{*}^{1}}^{1}}(H^{n_{1}}\wedge H^{n_{1}})\right)}_{1}\ldots\\
\underbrace{tr\left(P_{e_{^{\vartheta}i_{*}^{r_{1}}}^{r_{1}}}\left(\left[e_{i_{1}\neq^{\vartheta}i_{*}^{r_{1}}}^{r_{1}}\right]\wedge H^{n_{r_{1}}}\right)\right)}_{0}\ldots\underbrace{tr\left(P_{e_{i_{*}^{r_{2}}}^{r_{2}}}\left(H^{n_{r_{2}}}\wedge\left[e_{i_{2}=^{\vartheta}i_{*}^{r_{2}}}^{r_{2}}\right]\right)\right)}_{1}\ldots\\
tr\left(P_{e_{^{\vartheta}i_{*}^{m}}^{m}}(H^{n_{m}}\wedge H^{n_{m}})\right)\underbrace{tr\left(P_{\tilde{\Psi}_{\vartheta}}\left(^{\vartheta}\tilde{E}_{i_{1}}^{r_{1}}\wedge{}^{\vartheta}\tilde{E}_{i_{2}}^{r_{2}}\right)\right)}_{\,\,\,\,\,\,\,\,\,\,\,\,\,\,\,\,\,\,\,\,\,\simeq\,w_{\vartheta}{}_{i_{1}i_{2}}^{r_{1}r_{2}}=0}=0
\end{eqnarray*}
in accordance with that $^{\vartheta}w{}_{i_{1}i_{2}}^{r_{1}r_{2}}$
must be equal to $0$, due to (\ref{eq:utolso-elotti-1}).

\noindent If $w_{\vartheta}{}_{i_{1}}^{r_{1}}=1$, $w_{\vartheta}{}_{i_{2}}^{r_{2}}=1$
then
\begin{eqnarray*}
tr\left(P_{\Psi_{\vartheta}}\left(^{\vartheta}E_{i_{1}}^{r_{1}}\wedge{}^{\vartheta}E_{i_{2}}^{r_{2}}\right)\right)=\underbrace{tr\left(P_{e_{^{\vartheta}i_{*}^{1}}^{1}}(H^{n_{1}}\wedge H^{n_{1}})\right)}_{1}\ldots\\
\underbrace{tr\left(P_{e_{^{\vartheta}i_{*}^{r_{1}}}^{r_{1}}}\left(\left[e_{i_{1}={}^{\vartheta}i_{*}^{r_{1}}}^{r_{1}}\right]\wedge H^{n_{r_{1}}}\right)\right)}_{1}\ldots\underbrace{tr\left(P_{e_{i_{*}^{r_{2}}}^{r_{2}}}\left(H^{n_{r_{2}}}\wedge\left[e_{i_{2}=^{\vartheta}i_{*}^{r_{2}}}^{r_{2}}\right]\right)\right)}_{1}\ldots\\
tr\left(P_{e_{^{\vartheta}i_{*}^{m}}^{m}}(H^{n_{m}}\wedge H^{n_{m}})\right)\underbrace{tr\left(P_{\tilde{\Psi}_{\vartheta}}\left(^{\vartheta}\tilde{E}_{i_{1}}^{r_{1}}\wedge{}^{\vartheta}\tilde{E}_{i_{2}}^{r_{2}}\right)\right)}_{\,\,\,\,\,\,\,\,\,\,\,\,\,\,\,\,\,\,\,\,\,\simeq\,w_{\vartheta}{}_{i_{1}i_{2}}^{r_{1}r_{2}}}\simeq w_{\vartheta}{}_{i_{1}i_{2}}^{r_{1}r_{2}}
\end{eqnarray*}

\subsubsection*{Step~II}

Consider an arbitrary state $\vec{Z}$. Since $\vec{Z}\in\varphi\left(M,S\right)$,
it can be decomposed in terms of the vertices $\left\{ \vec{w}_{\vartheta}\right\} _{\vartheta\in\Theta}$
in the fashion of (\ref{eq:Z-composition}) with some coefficients
$\left\{ \lambda_{\vartheta}\right\} _{\vartheta\in\Theta}$.

Now we construct the Hilbert space $H$ and the state vector $\Psi_{\vec{Z}}$:

\begin{eqnarray}
H & = & \underset{\vartheta\in\Theta}{\mathop{\oplus}}^{\vartheta}H\label{eq:linear=000020sum=000020space}\\
\Psi_{\vec{Z}} & = & \underset{\vartheta\in\Theta}{\mathop{\oplus}}\sqrt{\lambda_{\vartheta}}\Psi_{\vartheta}\label{eq:state-vector-construction}
\end{eqnarray}
Obviously, 
\[
\left\langle \Psi_{\vec{Z}},\Psi_{\vec{Z}}\right\rangle =\sum_{\vartheta\in\Theta}\lambda_{\vartheta}\left\langle \Psi_{\vartheta},\Psi_{\vartheta}\right\rangle =1
\]
The subspaces $E_{i}^{r}$ representing the outcome events will be
defined further below. First we consider the following subspaces of
$H$:
\begin{eqnarray*}
^{*}E_{i}^{r} & = & \underset{\vartheta\in\Theta}{\mathop{\oplus}}^{\vartheta}E_{i}^{r}
\end{eqnarray*}
Since
\begin{eqnarray*}
tr\left(P_{\Psi_{\vec{Z}}}{}^{*}E_{i}^{r}\right) & = & \left\langle \Psi_{\vec{Z}},{}^{*}E_{i}^{r}\Psi_{\vec{Z}}\right\rangle =\sum_{\vartheta\in\Theta}\left\langle \sqrt{\lambda_{\vartheta}}\Psi_{\vartheta},^{\vartheta}E_{i}^{r}\sqrt{\lambda_{\vartheta}}\Psi_{\vartheta}\right\rangle \\
 &  & =\sum_{\vartheta\in\Theta}\lambda_{\vartheta}tr\left(P_{\Psi_{\vartheta}}{}^{\vartheta}E_{i}^{r}\right)
\end{eqnarray*}

\begin{eqnarray*}
tr\left(P_{\Psi_{\vec{Z}}}\left(^{*}E_{i_{1}}^{r_{1}}\wedge\ldots\wedge{}^{*}E_{i_{L}}^{r_{L}}\right)\right) & = & \left\langle \Psi_{\vec{Z}},\left(^{*}E_{i_{1}}^{r_{1}}\wedge\ldots\wedge{}^{*}E_{i_{L}}^{r_{L}}\right)\Psi_{\vec{Z}}\right\rangle \\
 & = & \sum_{\vartheta\in\Theta}\left\langle \sqrt{\lambda_{\vartheta}}\Psi_{\vartheta},\left(^{\vartheta}E_{i_{1}}^{r_{1}}\wedge\ldots\wedge{}^{\vartheta}E_{i_{L}}^{r_{L}}\right)\sqrt{\lambda_{\vartheta}}\Psi_{\vartheta}\right\rangle \\
 & = & \sum_{\vartheta\in\Theta}\lambda_{\vartheta}tr\left(P_{\Psi_{\vartheta}}\left(^{\vartheta}E_{i_{1}}^{r_{1}}\wedge\ldots\wedge{}^{\vartheta}E_{i_{L}}^{r_{L}}\right)\right)
\end{eqnarray*}
from (\ref{eq:approximate-teta-1})--(\ref{eq:approximate-teta-2})
and (\ref{eq:Z-composition}) we have
\begin{eqnarray}
Z_{i}^{r} & \simeq & tr\left(P_{\Psi_{\vec{Z}}}{}^{*}E_{i}^{r}\right)\label{eq:approximate-tetel-1-1}\\
Z_{i_{1}\ldots i_{L}}^{r_{1}\ldots r_{L}} & \simeq & tr\left(P_{\Psi_{\vec{Z}}}\left(^{*}E_{i_{1}}^{r_{1}}\wedge\ldots\wedge{}^{*}E_{i_{L}}^{r_{L}}\right)\right)\label{eq:approximate-tetel-2-1}
\end{eqnarray}
Also, as direct sum preserves orthogonality, from (\ref{eq:orthogonality-teta})
we have
\begin{equation}
^{*}E_{i}^{r}\bot{}^{*}E_{j}^{r}\text{ \,\,\,\,\,if \,\,\,}i\neq j\label{eq:orthogonal-1}
\end{equation}

For all $1\leq r\leq m$, let $^{*}E_{i_{0}}^{r}\in\left\{ ^{*}E_{1}^{r},{}^{*}E_{2}^{r},\ldots{}^{*}E_{n_{r}}^{r}\right\} $
be arbitrarily chosen, and let $^{*}E_{\bot}^{r}=\left(\stackrel[i=1]{n_{r}}{\vee}{}^{*}E_{i}^{r}\right)^{\bot}=\stackrel[i=1]{n_{r}}{\wedge}\left(^{*}E_{i}^{r}\right)^{\bot}$.
We define the subspaces representing the outcome events as follows:
\begin{eqnarray}
E_{i}^{r} & = & \begin{cases}
^{*}E_{i}^{r} & i\neq i_{0}\\
^{*}E_{i_{0}}^{r}\vee{}^{*}E_{\bot}^{r} & i=i_{0}
\end{cases}\label{eq:E-definicio}
\end{eqnarray}
Obviously, (\ref{eq:orthogonal-1}) implies $^{*}E_{i_{0}}^{r}\leq\stackrel[i\neq i_{0}]{}{\wedge}\left(^{*}E_{i}^{r}\right)^{\bot}$.
Due to the orthomodularity of the subspace lattice $L(H)$, we have
\[
^{*}E_{i_{0}}^{r}\vee\left(\underbrace{\left(^{*}E_{i_{0}}^{r}\right)^{\bot}\wedge\left(\stackrel[i\neq i_{0}]{}{\wedge}\left(^{*}E_{i}^{r}\right)^{\bot}\right)}_{^{*}E_{\bot}^{r}}\right)=\stackrel[i\neq i_{0}]{}{\wedge}\left(^{*}E_{i}^{r}\right)^{\bot}
\]
meaning that
\[
E_{i_{0}}^{r}=\stackrel[i\neq i_{0}]{}{\wedge}\left(^{*}E_{i}^{r}\right)^{\bot}
\]
Therefore, taking into account (\ref{eq:orthogonal-1}) and (\ref{eq:E-definicio}),
\begin{equation}
E_{i}^{r}\bot E_{j}^{r}\text{ \,\,\,\,\,if \,\,\,}i\neq j\label{eq:orthogonal-1-1}
\end{equation}
Also, it is obviously true that 
\begin{equation}
\stackrel[i=1]{n_{r}}{\vee}E_{i}^{r}=H\label{eq:kiteszi-1}
\end{equation}
Both (\ref{eq:orthogonal-1-1}) and (\ref{eq:kiteszi-1}) hold for
all $1\leq r\leq m$. There remains to show (\ref{eq:approximate-tetel-1})--(\ref{eq:approximate-tetel-2}).

It follows from (\ref{eq:E-definicio}) that
\[
E_{i}^{r}\geq{}^{*}E_{i}^{r}
\]
for all $_{i}^{r}\in I^{M}$. Similarly, 
\[
E_{i_{1}}^{r_{1}}\wedge\ldots\wedge E_{i_{L}}^{r_{L}}\geq{}^{*}E_{i_{1}}^{r_{1}}\wedge\ldots\wedge{}^{*}E_{i_{L}}^{r_{L}}
\]
for all $_{i_{1}\ldots i_{L}}^{r_{1}\ldots r_{L}}\in S$. Therefore,
for all $\vec{Z}\in\varphi\left(M,S\right)$,
\begin{equation}
\left\langle \Psi_{\vec{Z}},E_{i}^{r}\Psi_{\vec{Z}}\right\rangle \geq\left\langle \Psi_{\vec{Z}},{}^{*}E_{i}^{r}\Psi_{\vec{Z}}\right\rangle \label{eq:minden_nagyobb-1}
\end{equation}
and 
\begin{equation}
\left\langle \Psi_{\vec{Z}},E_{i_{1}}^{r_{1}}\wedge\ldots\wedge E_{i_{L}}^{r_{L}}\Psi_{\vec{Z}}\right\rangle \geq\left\langle \Psi_{\vec{Z}},{}^{*}E_{i_{1}}^{r_{1}}\wedge\ldots\wedge{}^{*}E_{i_{L}}^{r_{L}}\Psi_{\vec{Z}}\right\rangle \label{eq:minden_nagyobb-2}
\end{equation}

Now, (\ref{eq:osszeg=00003D1}) and (\ref{eq:approximate-tetel-1-1})
imply that
\[
\sum_{\begin{array}{c}
i\\
\left(_{i}^{r}\in I^{M}\right)
\end{array}}\left\langle \Psi_{\vec{Z}},{}^{*}E_{i}^{r}\Psi_{\vec{Z}}\right\rangle \simeq1
\]
At the same time, taking into account (\ref{eq:orthogonal-1-1})--(\ref{eq:kiteszi-1}),
we have
\begin{equation}
1=\sum_{\begin{array}{c}
i\\
\left(_{i}^{r}\in I^{M}\right)
\end{array}}\left\langle \Psi_{\vec{Z}},E_{i}^{r}\Psi_{\vec{Z}}\right\rangle \geq\sum_{\begin{array}{c}
i\\
\left(_{i}^{r}\in I^{M}\right)
\end{array}}\left\langle \Psi_{\vec{Z}},{}^{*}E_{i}^{r}\Psi_{\vec{Z}}\right\rangle \simeq1\label{eq:osszeg_nagyobb-1}
\end{equation}
From (\ref{eq:minden_nagyobb-1}) and (\ref{eq:osszeg_nagyobb-1}),
therefore,
\begin{equation}
tr\left(P_{\Psi_{\vec{Z}}}E_{i}^{r}\right)\simeq tr\left(P_{\Psi_{\vec{Z}}}{}^{*}E_{i}^{r}\right)\label{eq:bovebb-trace-1}
\end{equation}
Similarly, on the one hand, (\ref{eq:1d}) and (\ref{eq:approximate-tetel-2-1})
imply that
\begin{equation}
\sum_{\begin{array}{c}
i_{1},i_{2}\ldots i_{L}\\
\left(_{i_{1}\ldots i_{L}}^{r_{1}\ldots r_{L}}\in S\right)
\end{array}}\left\langle \Psi_{\vec{Z}},{}^{*}E_{i_{1}}^{r_{1}}\wedge\ldots\wedge{}^{*}E_{i_{L}}^{r_{L}}\Psi_{\vec{Z}}\right\rangle \simeq1\label{eq:konjuncio_is_1}
\end{equation}
On the other hand, $\left\{ E_{i_{1}}^{r_{1}}\wedge\ldots\wedge E_{i_{L}}^{r_{L}}\right\} _{\begin{array}{c}
i_{1},i_{2}\ldots i_{L}\\
\left(_{i_{1}\ldots i_{L}}^{r_{1}\ldots r_{L}}\in S\right)
\end{array}}$is an orthogonal system of subspaces. Therefore,
\begin{align}
 & 1\geq\sum_{\begin{array}{c}
i_{1},i_{2}\ldots i_{L}\\
\left(_{i_{1}\ldots i_{L}}^{r_{1}\ldots r_{L}}\in S\right)
\end{array}}\left\langle \Psi_{\vec{Z}},E_{i_{1}}^{r_{1}}\wedge\ldots\wedge E_{i_{L}}^{r_{L}}\Psi_{\vec{Z}}\right\rangle \nonumber \\
 & \,\,\,\,\,\,\,\,\,\,\,\,\,\,\,\,\geq\sum_{\begin{array}{c}
i_{1},i_{2}\ldots i_{L}\\
\left(_{i_{1}\ldots i_{L}}^{r_{1}\ldots r_{L}}\in S\right)
\end{array}}\left\langle \Psi_{\vec{Z}},{}^{*}E_{i_{1}}^{r_{1}}\wedge\ldots\wedge{}^{*}E_{i_{L}}^{r_{L}}\Psi_{\vec{Z}}\right\rangle \simeq1\label{eq:osszeg_nagyobb-2}
\end{align}
From (\ref{eq:minden_nagyobb-2}) and (\ref{eq:osszeg_nagyobb-2}),
for all $_{i_{1}\ldots i_{L}}^{r_{1}\ldots r_{L}}\in S$, we have
\begin{equation}
tr\left(P_{\Psi_{\vec{Z}}}E_{i_{1}}^{r_{1}}\wedge\ldots\wedge E_{i_{L}}^{r_{L}}\right)\simeq tr\left(P_{\Psi_{\vec{Z}}}{}^{*}E_{i_{1}}^{r_{1}}\wedge\ldots\wedge{}^{*}E_{i_{L}}^{r_{L}}\right)\label{eq:bovebb-trace-2}
\end{equation}

Thus, (\ref{eq:approximate-tetel-1-1})--(\ref{eq:approximate-tetel-2-1})
together with (\ref{eq:bovebb-trace-1}) and (\ref{eq:bovebb-trace-2})
imply (\ref{eq:approximate-tetel-1})--(\ref{eq:approximate-tetel-2}).
\end{proof}
With Theorem~\ref{Theorem:Main} we have accomplished (Q1)--(Q4).
The next two theorems cover statements (Q5)--(Q8).
\begin{thm}
\label{trace}Let $a_{r}$ be the measurement of a real valued quantity
with labeling (\ref{eq:labeling-1}). On the Hilbert space $H$, there
exists a self-adjoint operator $A_{r}$ such that for every state
of the system $\vec{Z}$,
\begin{equation}
\left\langle \alpha_{r}\right\rangle _{\vec{Z}}\simeq tr\left(P_{\Psi_{\vec{Z}}}A_{r}\right)\label{eq:mean-trA}
\end{equation}
\end{thm}

\begin{proof}
Let 
\begin{equation}
A_{r}=\sum_{i=1}^{n_{r}}\alpha_{i}^{r}E_{i}^{r}\label{eq:spektral}
\end{equation}
$A_{r}$ is obviously a self-adjoint operator, and
\begin{eqnarray*}
\left\langle \alpha_{r}\right\rangle _{\vec{Z}} & = & \sum_{i=1}^{n_{r}}\alpha_{i}^{r}p\left(X_{i}^{r}|a_{r}\right)=\sum_{i=1}^{n_{r}}\alpha_{i}^{r}Z_{i}^{r}\simeq\sum_{i=1}^{n_{r}}\alpha_{i}^{r}tr\left(P_{\Psi_{\vec{Z}}}E_{i}^{r}\right)\\
 &  & =tr\left(P_{\Psi_{\vec{Z}}}\sum_{i=1}^{n_{r}}\alpha_{i}^{r}E_{i}^{r}\right)=tr\left(P_{\Psi_{\vec{Z}}}A_{r}\right)
\end{eqnarray*}
\end{proof}
\begin{thm}
\label{spectral}The possible measurement results of the $\alpha_{r}$-measurement
are exactly the eigenvalues of the associated operator $A_{r}.$ The
subspace $E_{i}^{r}$ representing the outcome event labeled by $\alpha_{i}^{r}$
is the eigenspace pertaining to eigenvalue $\alpha_{i}^{r}$. Accordingly,
(\ref{eq:spektral}) constitutes the spectral decomposition of $A_{r}$.
\end{thm}

\begin{proof}
First, let $\psi\in E_{i}^{r}$. Then, due to (\ref{eq:orthogonal}),
$A_{r}\psi=\left(\sum_{i=1}^{n_{r}}\alpha_{i}^{r}E_{i}^{r}\right)\psi=\alpha_{i}^{r}\psi$.
Meaning that every $\alpha_{i}^{r}$ is an eigenvalue of $A_{r}.$
Now consider an arbitrary eigenvector of $A_{r}$, that is, a vector
$\psi\in H$ such that 
\begin{equation}
A_{r}\psi=x\psi\label{eq:eigen}
\end{equation}
with some $x\in\mathbb{R}$. Due to (\ref{eq:orthogonal})--(\ref{eq:kiadja}),
$\left\{ E_{1}^{r},E_{2}^{r},\ldots E_{n_{r}}^{r}\right\} $ constitutes
an orthogonal decomposition of $H$, meaning that arbitrary $\psi\in H$
can be decomposed as 
\[
\psi=\sum_{i=1}^{n_{r}}\psi_{i}\,\,\,\,\,\,\,\psi_{i}\in E_{i}^{r}
\]

\noindent From (\ref{eq:spektral}) we have
\begin{equation}
\sum_{i=1}^{n_{r}}\alpha_{i}^{r}\psi_{i}=\sum_{i=1}^{n_{r}}x\psi_{i}\label{eq:valamelyik-nem-nulla}
\end{equation}
About the labeling we have assumed that $\alpha_{i}^{r}\neq\alpha_{j}^{r}$
for $i\neq j$, therefore, (\ref{eq:valamelyik-nem-nulla}) implies
that
\begin{eqnarray*}
x & = & \alpha_{i}^{r}\,\,\,\,\,\,\,\text{ for one }\alpha_{i}^{r}\\
\psi_{j} & = & 0\,\,\,\,\,\,\,\,\,\text{ for all }j\neq i
\end{eqnarray*}
that is, $\psi\in E_{i}^{r}$. Meaning that (\ref{eq:spektral}) is
the spectral decomposition of $A_{r}$.
\end{proof}
A consequence of Theorems~\ref{trace} and \ref{spectral} is that
if $f:\mathbb{R}\rightarrow\mathbb{R}$ is an arbitrary injective
function ``re-labeling'' the outcomes, then 
\begin{align*}
\left\langle f(\alpha_{r})\right\rangle _{\vec{Z}} & =\sum_{i=1}^{n_{r}}f(\alpha_{i}^{r})p\left(X_{i}^{r}|a_{r}\right)=\sum_{i=1}^{n_{r}}f(\alpha_{i}^{r})Z_{i}^{r}\simeq\sum_{i=1}^{n_{r}}f(\alpha_{i}^{r})\,tr\left(P_{\Psi_{\vec{Z}}}E_{i}^{r}\right)\\
 & =tr\left(P_{\Psi_{\vec{Z}}}\sum_{i=1}^{n_{r}}f(\alpha_{i}^{r})E_{i}^{r}\right)=tr\left(P_{\Psi_{\vec{Z}}}f(A_{r})\right)
\end{align*}

\section{Representation of Dynamics}

Notice that not all unit vectors of $H$ are involved in the representation
of states. In order to specify the ones being involved, consider the
following subspace $\mathcal{H}\subset H$:
\begin{eqnarray*}
\mathcal{H} & = & \mathsf{span}\left\{ \Psi_{\vartheta}\right\} _{\vartheta\in\Theta}
\end{eqnarray*}
where $\left\{ \Psi_{\vartheta}\right\} _{\vartheta\in\Theta}$ is
the set of vectors in the direct sum (\ref{eq:state-vector-construction}),
understood as being pairwise orthogonal, unit-length elements of $H$.
Denote by $\mathcal{O}$ the closed first hyperoctant (orthant) of
the $\left(\left|\Theta\right|-1\right)$-dimensional sphere of unit
vectors in $\mathcal{H}$:
\[
\mathcal{O}=\left\{ \sum_{\vartheta\in\Theta}o_{\vartheta}\Psi_{\vartheta}\left|o_{\vartheta}\geq0\,\,\,\,\sum_{\vartheta\in\Theta}o_{\vartheta}^{2}=1\right.\right\} 
\]
Obviously, there is a continuous one-to-one map between the $\Lambda$
and $\mathcal{O}$: 
\[
O:\Lambda\rightarrow\mathcal{O};\,\,O\left(\vec{\lambda}\right)=\sum_{\vartheta\in\Theta}\sqrt{\lambda_{\vartheta}}\Psi_{\vartheta}
\]
As we have shown, however, the states of the system actually are represented
on the $\text{dim}\left(\varphi\left(M,S\right)\right)$-dimensional
slice $\Lambda_{\sigma}\subset\Lambda$ (see (\ref{eq:Lambda0})).
Accordingly, the quantum mechanical representation of states constitutes
a $\text{dim}\left(\varphi\left(M,S\right)\right)$-dimensional submanifold
with boundary: $\mathcal{O}_{\sigma}=O\left(\Lambda_{\sigma}\right)\subset\mathcal{O}$.

Consequently, the time evolution of state $\vec{Z}\left(t\right)$
will be represented by a path in $\mathcal{O}_{\sigma}$:
\[
\Psi(t)=O\circ\sigma\left(\vec{Z}\left(t\right)\right)
\]
The representation is of course not unique, as it depends on the choice
of cross section $\sigma$. This is however inessential; just like
a choice of a coordinate system.

As emphasized at the beginning of Section~\ref{sec:Quantum-Representation},
the quantum representation was derived exclusively from the assumptions
\textbf{(E1)}-\textbf{(E3)}. If in addition \textbf{(E4)} holds, that
is the time evolution $\vec{Z}\left(t\right)$ can be generated by
a one-parameter group of transformations on $\varphi\left(M,S\right)$,
$\vec{Z}(t)=F_{t-t_{0}}\left(\vec{Z}(t_{0})\right)$, then the same
is true for $\mathcal{O}_{\sigma}$. Let $G_{t}=O\circ\sigma\circ F_{t}\circ D\circ O^{-1}$.
Obviously, $G_{t}$ is a map $\mathcal{O}_{\sigma}\rightarrow\mathcal{O}_{\sigma}$,
such that
\begin{itemize}
\item [] $G_{t}:\mathcal{O}_{\sigma}\rightarrow\mathcal{O}_{\sigma}$ is
one-to-one
\item [] $G:\mathbb{R}\times\mathcal{O}_{\sigma}\rightarrow\mathcal{O}_{\sigma};\,\left(t,\Psi\right)\mapsto G_{t}\left(\Psi\right)$
is continuous
\item [] $G_{t+s}=G_{s}\circ G_{t}$
\item [] $G_{-t}=G_{t}^{-1}$; consequently, $G_{0}=id_{\mathcal{O}_{\sigma}}$
\end{itemize}
and the time evolution of an arbitrary initial state $\Psi(t_{0})\in\mathcal{O}_{\sigma}$
is $\Psi(t)=G_{t-t_{0}}\left(\Psi(t_{0})\right)$.

\section{Questionable and Unquestionable in Quantum Mechanics\protect\label{sec:Questionable-and-Unquestionable}}

What we have \emph{proved} in the above theorems, that is, statements
(Q1)--(Q8), are nothing but the basic postulates of quantum theory.
This means that the basic postulates of quantum theory are in fact
analytic statements: they do not tell us anything about a physical
system beyond the fact that the system can be described in empirical/operational
terms---even if this logical relationship is not so evident. In this
sense, of course, these postulates of quantum theory are unquestionable.
Though, as we have seen, the Hilbert space quantum mechanical formalism
is only an optional mathematical representation of the probabilistic
behavior of a system---empirical facts do not necessitate this choice.

Nevertheless, it must be mentioned that the quantum-mechanics-like
representation, characterized by (Q1)--(Q8), is not completely identical
with standard quantum mechanics. There are several subtle deviations:
\begin{itemize}
\item [(D1)] There is no one-to-one correspondence between operationally
meaningful physical quantities and self-adjoint operators. First of
all, it is not necessarily true that every self-adjoint operator represents
some operationally meaningful quantity.
\item [(D2)]\label{-no_commutation} There is no necessary connection between
commutation of the associated self-adjoint operators and joint measurability
of the corresponding physical quantities. In general, there is no
obvious role of the mathematically definable algebraic structures
over the self-adjoint operators in the operational context. First
of all because those mathematically ``natural'' structures are mostly
meaningless in an operational sense. As we have already mentioned,
the outcome \emph{events} are ontologically prior to the labeling
of the outcomes by means of numbers; and the events themselves are
well represented in the subspace/projector lattice, prior to any self-adjoint
operator associated with a numerical coordination.

For example, consider three real-valued physical quantities with labelings
$\alpha_{r_{1}},\alpha_{r_{2}},\alpha_{r_{3}}$. The three physical
quantities reflect three different features of the system defined
by three different measurement operations. A functional relationship
$\alpha_{r_{1}}=f\left(\alpha_{r_{2}},\alpha_{r_{3}}\right)$ means
that whenever we perform the measurements $a_{r_{1}},a_{r_{2}},a_{r_{3}}$
in conjunction (meaning that $\left\{ r_{1},r_{2},r_{3}\right\} \in\mathfrak{P}$)
the outcomes $X_{i_{1}}^{r_{1}},X_{i_{2}}^{r_{2}},X_{i_{3}}^{r_{3}}$
are strongly correlated: if $X_{i_{2}}^{r_{2}}$ and $X_{i_{3}}^{r_{3}}$
are the outcomes of $a_{r_{2}}$ and $a_{r_{3}}$, labeled by $\alpha_{i_{2}}^{r_{2}}$
and $\alpha_{i_{3}}^{r_{3}}$, then the outcome of measurement $a_{r_{1}}$,
$X_{i_{1}}^{r_{1}}$, is the one labeled by $\alpha_{i_{1}}^{r_{1}}=f\left(\alpha_{i_{2}}^{r_{2}},\alpha_{i_{3}}^{r_{3}}\right)$.
That is, in probabilistic terms:
\begin{align}
p\left(\alpha_{r_{1}}^{-1}\left(f\left(\alpha_{i_{2}}^{r_{2}},\alpha_{i_{3}}^{r_{3}}\right)\right)\wedge\alpha_{r_{2}}^{-1}\left(\alpha_{i_{2}}^{r_{2}}\right)\wedge\alpha_{r_{3}}^{-1}\left(\alpha_{i_{3}}^{r_{3}}\right)\left|a_{r_{1}}\wedge a_{r_{2}}\wedge a_{r_{3}}\right.\right)\nonumber \\
=p\left(\alpha_{r_{2}}^{-1}\left(\alpha_{i_{2}}^{r_{2}}\right)\wedge\alpha_{r_{3}}^{-1}\left(\alpha_{i_{3}}^{r_{3}}\right)\left|a_{r_{1}}\wedge a_{r_{2}}\wedge a_{r_{3}}\right.\right)\label{eq:functional}
\end{align}
This contingent fact of regularity in the observed relative frequencies
of physical events is what is a part of the ontology. And it is well
reflected in our quantum mechanical representation, in spite of the
fact that the relationship (\ref{eq:functional}) is generally not
reflected in some algebraic or other functional relation of the associated
self-adjoint operators $A_{r_{1}}$, $A_{r_{2}}$ and $A_{r_{3}}$.

The fact that in our quantum-mechanics-like representation there is
no correspondence between commutation and co-measurability explains
why there is no need to require satisfaction of the Cirel'son inequalities
(cf. Cirenl'son~1980; Popescu and Rohrlich 1994; Müller~2021), which
would mean further restriction on the observed relative frequencies
beyond \textbf{(E1)}--\textbf{(E3)}.
\item [(D3)] It is worthwhile emphasizing that the Hilbert space of representation
is finite dimensional and real. It is of course no problem to embed
the whole representation into a complex Hilbert space of the same
dimension. As it follows from (\ref{eq:orthogonal}) and (\ref{eq:large_H}),
the required minimal dimension increases with increasing the number
of possible measurements $m$, and/or increasing the number of possible
outcomes $n_{r}$. In any event, it is finite until we have a finite
operational setup. Employing complex Hilbert spaces is only necessary
if, in addition to the stipulated operational setup, we have some
further algebraic requirements, for example, in the form of commutation
relations, and the likes. How those further requirements are justified
in operational terms, of course, can be a question.
\item [(D4)] There is no problem with the empirical meaning of the lattice-theoretic
meet of subspaces/projectors representing outcome events: the meet
represents the empirically meaningful conjunction of the outcome events,
regardless whether the corresponding projectors commute or not. Of
course, by definition (\ref{eq:pi_is_such-2}), the conjunctions that
do not belong to $S$ have zero probability in all states of the system.

In contrast, the lattice-theoretic joins and orthocomplements, in
general, have nothing to do with the disjunctions and negations of
the outcome events. Nevertheless, as we have seen, the quantum state
uniquely determines the probabilities on the whole event algebra,
including the conjunctions, disjunctions and negations of all events---in
the sense of Theorem~\ref{State--together}.
\item [(D5)] All possible states of the system, $\vec{Z}\in\varphi\left(M,S\right)$,
are represented by \emph{pure} states. That is to say, the quantum
mechanical notion of mixed state is not needed. The reason is very
simple. $\varphi\left(M,S\right)$ is a convex polytope being closed
under convex linear combinations. The state of the system intended
to be represented by a mixed state, say,
\[
W=\mu_{1}P_{\Psi_{\vec{Z}_{1}}}+\mu_{2}P_{\Psi_{\vec{Z}_{2}}}\,\,\,\,\,\,\,\,\,\,\,\,\,\,\mu_{1},\mu_{2}\geq0;\,\,\,\mu_{1}+\mu_{2}=1
\]
is nothing but another element of $\varphi\left(M,S\right)$,
\[
\vec{Z}_{3}=\mu_{1}\vec{Z}_{1}+\mu_{2}\vec{Z}_{2}\in\varphi\left(M,S\right)
\]
However, in our representation theorem (Theorem~\ref{Theorem:Main})
the Hilbert space and the representations of the outcome events are
constructed in a way that all states $\vec{Z}\in\varphi\left(M,S\right)$
are represented by a suitable state vector in one and the same Hilbert
space. Therefore, $\vec{Z}_{3}$ is also represented by a pure state
$P_{\Psi_{\vec{Z}_{3}}}$ with a suitably constructed state vector
$\Psi_{\vec{Z}_{3}}$. Namely, given that 
\begin{eqnarray*}
\vec{Z}_{1} & = & \sum_{\vartheta\in\Theta}\lambda_{\vartheta}^{1}\vec{w}_{\vartheta}\,\,\,\,\,\,\,\lambda_{\vartheta}^{1}\geq0,\sum_{\vartheta\in\Theta}\lambda_{\vartheta}^{1}=1\\
\vec{Z}_{2} & = & \sum_{\vartheta\in\Theta}\lambda_{\vartheta}^{2}\vec{w}_{\vartheta}\,\,\,\,\,\,\,\lambda_{\vartheta}^{2}\geq0,\sum_{\vartheta\in\Theta}\lambda_{\vartheta}^{2}=1
\end{eqnarray*}
we have 
\begin{eqnarray*}
\vec{Z}_{3} & = & \sum_{\vartheta\in\Theta}\left(\mu_{1}\lambda_{\vartheta}^{1}+\mu_{2}\lambda_{\vartheta}^{2}\right)\vec{w}_{\vartheta}
\end{eqnarray*}
therefore, from (\ref{eq:state-vector-construction}), 
\begin{eqnarray*}
\Psi_{\vec{Z}_{3}} & = & \underset{\vartheta\in\Theta}{\mathop{\oplus}}\sqrt{\mu_{1}\lambda_{\vartheta}^{1}+\mu_{2}\lambda_{\vartheta}^{2}}\Psi_{\vartheta}
\end{eqnarray*}

To avoid a possible misunderstanding, it is worthwhile mentioning
that all we said above is not in contradiction with the mathematical
fact that the density operators $W$ and $P_{\Psi_{\vec{Z}_{3}}}$
generate different ``quantum probability'' measures over the \emph{whole}
subspace lattice $L\left(H\right)$. The two measures will coincide
on those elements of $L\left(H\right)$ that represent operationally
meaningful events---$E_{i}^{r},E_{i_{1}}^{r_{1}}\wedge\ldots\wedge E_{i_{L}}^{r_{L}}$
for $_{i}^{r}\in I^{M},{}_{i_{1}\ldots i_{L}}^{r_{1}\ldots r_{L}}\in S$.
This reinforces the point in (D4) that there is no one-to-one correspondence
between the operationally meaningful events and the elements of $L\left(H\right)$.
\item [(D6)] We don't need to invoke the entire Hilbert space for representing
the totality of operationally meaningful possible states of the system;
subspace $\mathcal{H}$ is sufficient. Even in this restricted sense,
there is no one-to-one correspondence between the rays of the subspace
$\mathcal{H}\subset H$ and the states of the system. The unit vectors
involved in the representation are the ones pointing to $\mathcal{O}_{\sigma}$,
a $\text{dim}\left(\varphi\left(M,S\right)\right)$-dimensional submanifold
with boundary on the unit sphere of $\mathcal{H}$.
\item [(D7)] The so called ``superposition principle'' does not hold.
The ray determined by the linear combination of two different vectors
pointing to $\mathcal{O}_{\sigma}$ does not necessarily intersect
$\mathcal{O}_{\sigma}$; meaning that such a linear combination, in
general, has nothing to do with a third state of the system. Neither
has it anything to do with the logical/probability theoretic notion
of ``disjunction'' of events, of course. Nevertheless, as we have
already emphasized in (D4) and (D5), all possible states of the system
are well represented in $\mathcal{O}_{\sigma}$; and these states
uniquely determine the probabilities on the whole event algebra of
operationally meaningful events, including their disjunctions too.
\item [(D8)] The dynamics of the system can be well represented in the
usual way, by means of temporal evolution on the state manifold $\mathcal{O}_{\sigma}$.
In the case where (E4) is also satisfied, this temporal evolution
can be generated by a one-parameter group of transformations of $\mathcal{O}_{\sigma}$.
However, these transformations are in no way related to the unitary
transformations of $H$ (or $\mathcal{H}$); because they do not respect
the linear structure of the Hilbert space or orthogonality; but they
do respect that the state space $\mathcal{O}_{\sigma}$ is a manifold
with boundary.
\end{itemize}
It is remarkable that most of the above mentioned deviations from
the quantum mechanical folklore are related with exactly those issues
in the foundations of quantum mechanics that have been hotly debated
for long decades (e.g. Strauss~1936; Reichenbach~1944; Popper~1967;
Park and Margenau~1968; 1971; Ross~1974; Bell~1987; Gudder~1988;
Malament~1992; Leggett~1998; Griffiths~2013; Cassinelli and Lahti~2017;
Fröhlich and Pizzo~2022). The fact that so much of the core of quantum
theory can be unquestionably deduced from three elementary empirical
conditions, equally true about all physical systems whether classical
or quantum, or beyond, may shed new light on these old problems in
the foundations.

\section*{Appendices}

\subsection*{Appendix~1}

\noindent The following theorem is mentioned as an exercise in most
texts. We formulate it using the same notation we used in the proof
of Lemma~\ref{The-vertices-of-fi}.
\begin{thm}
\label{Theorem:=000020vertex}Let $P$ be a polytope in $\mathbb{R}^{d}$,
defined by the following set of linear inequalities:
\begin{eqnarray}
\left\langle \vec{\omega}_{\mu},\vec{f}\right\rangle -b_{\mu} & \leq & 0\,\,\,\,\,\text{for all }\mu\in I\label{eq:egyenlotlensegek}
\end{eqnarray}
For each $\vec{f}\in P$, define the active index set:{\small
\[
I_{\vec{f}}:=\left\{ \mu\in I\left|\,\left\langle \vec{\omega}_{\mu},\vec{f}\right\rangle -b_{\mu}=0\right.\right\} 
\]
}$\vec{f}\in P$ is a vertex of $P$ if and only if 
\begin{equation}
\mathrm{span}\left\{ \vec{\omega}_{\mu}\right\} _{\mu\in I_{\vec{f}}}=\mathbb{R}^{d}\label{eq:span-2}
\end{equation}
\begin{proof}
First, suppose $\vec{f}$ is vertex of $P$, but $\mathrm{span}\left\{ \vec{\omega}_{\mu}\right\} _{\mu\in I_{\vec{f}}}\neq\mathbb{R}^{d}$.
Then choose a non-zero $\vec{g}\in\left(\mathrm{span}\left\{ \vec{\omega}_{\mu}\right\} _{\mu\in I_{\vec{f}}}\right)^{\bot}$.
Obviously, if $\mu\not\in I_{\vec{f}}$ then there exists a neighborhood
$U$ of $\vec{f}$ such that $\mu\not\in I_{\vec{f}_{*}}$ for all
$\vec{f}_{*}\in U$. Consider the points $\vec{f}+\lambda\vec{g}$.
If $\lambda$ is small enough, both $\vec{f}+\lambda\vec{g}$ and
$\vec{f}-\lambda\vec{g}$ are in $P$, since (\ref{eq:egyenlotlensegek})
are satisfied. Now, we can write 
\[
\vec{f}=\frac{1}{2}\left(\left(\vec{f}+\lambda\vec{g}\right)+\left(\vec{f}-\lambda\vec{g}\right)\right)
\]
which contradicts the fact that $\vec{f}$ is vertex of $P$.

Second, now suppose that $\vec{f}\in P$ and $\mathrm{span}\left\{ \vec{\omega}_{\mu}\right\} _{\mu\in I_{\vec{f}}}=\mathbb{R}^{d}$.
Suppose $\vec{f}=\lambda\vec{f}_{*}+(1-\lambda)\vec{f}_{**}$ with
some $\vec{f}_{*},\vec{f}_{**}\in P$ and $0<\lambda<1$. We know
that $\mu\in I_{\vec{f}}$ implies
\[
\left\langle \vec{\omega}_{\mu},\vec{f}\right\rangle =\lambda\left\langle \vec{\omega}_{\mu},\vec{f}_{*}\right\rangle +(1-\lambda)\left\langle \vec{\omega}_{\mu},\vec{f}_{**}\right\rangle =b_{\mu}
\]
 On the other hand, from (\ref{eq:egyenlotlensegek}) we have 
\begin{eqnarray*}
\left\langle \vec{\omega}_{\mu},\vec{f}_{*}\right\rangle  & \leq & b_{\mu}\\
\left\langle \vec{\omega}_{\mu},\vec{f}_{**}\right\rangle  & \leq & b_{\mu}
\end{eqnarray*}
which implies that $\left\langle \vec{\omega}_{\mu},\vec{f}\right\rangle =\left\langle \vec{\omega}_{\mu},\vec{f}_{*}\right\rangle =\left\langle \vec{\omega}_{\mu},\vec{f}_{**}\right\rangle $
(for all $\mu\in I_{\vec{f}}$). Therefore, 
\[
\left(\vec{f}-\vec{f}_{*}\right),\left(\vec{f}-\vec{f}_{**}\right)\in\left(\mathrm{span}\left\{ \vec{\omega}_{\mu}\right\} _{\mu\in I_{\vec{f}}}\right)^{\bot}=\textrm{Ø}
\]
meaning that $\vec{f}=\vec{f}_{*}=\vec{f}_{**}$. Therefore, $\vec{f}$
is a vertex.
\end{proof}
\end{thm}

\subsection*{Appendix~2}

\noindent The standard argumentation in the GPT literature (Hardy~2008;
Holevo~2011, pp. 4--5; Müller~2021, p.~14) that the space of physically
possible states $\varPhi$ \emph{must} be convex is based on a problematic
notion of ``statistical mixture of preparation procedures''. A ``mixture''
of two preparation procedures, resulting in states $\vec{Z}_{1}$
and $\vec{Z}_{2}$, would be a procedure in which we alternate between
two preparations in some ratio of $\lambda_{1}$ and $\lambda_{2}$
($\lambda_{1}+\lambda_{2}=1$), say randomly, with probabilities $\lambda_{1}$
and $\lambda_{2}$. Under some circumstances---excluding any tricky
correlations, for example, between the choice of preparations and
the choice of measurements, or between the system's behavior in one
run and its preparation in the previous run, etc.---the surface observed
statistics would indeed be as if the system were in state $\lambda_{1}\vec{Z}_{1}+\lambda_{2}\vec{Z}_{2}$.
``The situation described above can be considered as a special way
of state preparation'' (Holevo~2011, p.~5), the argument says,
therefore $\lambda_{1}\vec{Z}_{1}+\lambda_{2}\vec{Z}_{2}\in\varPhi$.

We believe that this argument is conceptually flawed. After all, the
notion of relative frequency itself is based on a series of measurements
that are performed in the same probabilistic setup, and not in a setup
changing from one run to the next. Taking into account what the probabilistic
setup is that is fixed in the ``mixing'' procedure, the ``mixture''
$\lambda_{1}\vec{Z}_{1}+\lambda_{2}\vec{Z}_{2}$ is not a state of
the original system, but a state of the composed system consisting
of the preparation device (denote it by $\mathscr{D}$) and the original
physical system (denote it by $\mathscr{S}$).  We should not be
mislead by the fact that different physical things can be described
by the same mathematical object. In order to avoid the confusion it
is better to use different notation for the physically different things.
When we are talking about the mixture, we are talking about states
of the composed ($\mathscr{D}$+$\mathscr{S}$) system:
\begin{lyxlist}{00.00.0000}
\item [{$^{(\mathscr{D}+\mathscr{S})}\vec{Z}_{1}\in{}^{(\mathscr{D}+\mathscr{S})}\varPhi$:}] the
state of ($\mathscr{D}$+$\mathscr{S}$) in which the preparation
device $\mathscr{D}$ has propensity $1$ to set system $\mathscr{S}$
into state $\vec{Z}_{1}\in\varPhi$
\item [{$^{(\mathscr{D}+\mathscr{S})}\vec{Z}_{2}\in{}^{(\mathscr{D}+\mathscr{S})}\varPhi$:}] the
state of ($\mathscr{D}$+$\mathscr{S}$) in which the preparation
device $\mathscr{D}$ has propensity $1$ to set system $\mathscr{S}$
into state $\vec{Z}_{2}\in\varPhi$
\item [{$^{(\mathscr{D}+\mathscr{S})}\vec{Z}_{3}\in{}^{(\mathscr{D}+\mathscr{S})}\varPhi$:}] the
state of ($\mathscr{D}$+$\mathscr{S}$) in which the preparation
device $\mathscr{D}$ has propensity $\lambda_{1}$ to set system
$\mathscr{S}$ into state $\vec{Z}_{1}\in\varPhi$ and propensity
$\lambda_{2}$ to set system $\mathscr{S}$ into state $\vec{Z}_{2}\in\varPhi$
\end{lyxlist}
Indeed, 
\[
^{(\mathscr{D}+\mathscr{S})}\vec{Z}_{3}=\lambda_{1}{}^{(\mathscr{D}+\mathscr{S})}\vec{Z}_{1}+\lambda_{2}{}^{(\mathscr{D}+\mathscr{S})}\vec{Z}_{2}
\]
and, if the preparation device is such that propensities $\lambda_{1}$
and $\lambda_{2}$ ($\lambda_{1}+\lambda_{2}=1$) can be arbitrary,
then $^{(\mathscr{D}+\mathscr{S})}\varPhi$ is closed under convex
combination.

Just take our Biased Coin example we discussed at the end of section~\ref{sec:The-State-Space-polytope}.
Consider a composed ($\mathscr{D}$+$\mathscr{S}$) system, where
$\mathscr{S}$ is the coin on the right hand side of Fig.~\ref{coin},
and $\mathscr{D}$ is a device setting the disk inside the coin into
one of the possible positions $Z_{H},Z_{F}$, or $Z_{T}$ with some
probabilities $\lambda_{H},\lambda_{F},$ and $\lambda_{T}$. Consider
only one possible measurement, tossing the coin, with two possible
outcomes, Heads and Tails, and one conjunction. Now, keeping the preparation
of the ($\mathscr{D}$+$\mathscr{S}$) system fixed, that is, the
coin always has the same three slots where the disk can be clicked
and the preparation device always has the same propensities $\lambda_{H},\lambda_{F},$
and $\lambda_{T}$, we can take the statistics of Heads and Tails.
The observed data will satisfy \textbf{(E1)}--\textbf{(E3)}, and
the state space $\varphi\left(M,S\right)$ (where $M=2$ and $S=\left\{ _{1\,2}^{1\,1}\right\} $)
is the one-dimensional polytope in $\mathbb{R}^{3}$, determined by
vertices $\left(1,0,0\right)$ and $\left(0,1,0\right)$---just as
it was the case in our coin example. Denote an arbitrary state vector
of the ($\mathscr{D}$+$\mathscr{S}$) system by $^{(\mathscr{D}+\mathscr{S})}\vec{Z}=\left(^{(\mathscr{D}+\mathscr{S})}z_{H},^{(\mathscr{D}+\mathscr{S})}z_{T},0\right)\in\varphi(M,S)$,
meaning that whenever we perform the coin-toss and the ($\mathscr{D}$+$\mathscr{S}$)
system is in state $^{(\mathscr{D}+\mathscr{S})}\vec{Z}$, we get
Heads with relative frequency $^{(\mathscr{D}+\mathscr{S})}z_{H}$,
Tails with $^{(\mathscr{D}+\mathscr{S})}z_{T}$, but never the Heads
and Tails in conjunction. Assuming that the device $\mathscr{D}$
is such that its propensities $\lambda_{H},\lambda_{F},$ and $\lambda_{T}$
can be arbitrary ($\lambda_{H}+\lambda_{F}+\lambda_{T}=1$), the set
of the physically possible states of the ($\mathscr{D}$+$\mathscr{S}$)
system, $^{(\mathscr{D}+\mathscr{S})}\varPhi\subset\varphi\left(M,S\right)$,
is the line segment between $(0.8,0.2,0)$ and $(0.2,0.8,0)$. Hence
$^{(\mathscr{D}+\mathscr{S})}\varPhi$ is closed under convex combination.
While, recall (\ref{eq:coin-diszkret}), the set of the physically
possible states of the coin, the original system $\mathscr{S}$ in
itself, is 
\[
\varPhi=\left\{ (0.8,0.2,0),(0.5,0.5,0),(0.2,0.8,0)\right\} 
\]
which is not closed under convex combination.

In contrast, consider an ($\mathscr{D}$+$\mathscr{S}$) system such
that $\mathscr{S}$ is a coin like the one on the left hand side of
Fig.~\ref{coin}, but the preparation device $\mathscr{D}$ has only
two possible propensity-states: placing the disk in position $Z_{H}$
with probability $1$, or placing the disk in position $Z_{T}$ with
probability $1$. In this case, the set of the physically possible
states of the ($\mathscr{D}$+$\mathscr{S}$) system consists of only
two points, 
\[
^{(\mathscr{D}+\mathscr{S})}\varPhi=\left\{ (0.8,0.2,0),(0.2,0.8,0)\right\} 
\]
hence it is not closed under convex combination. While, the set of
the physically possible states of the coin in itself, $\varPhi$,
is the whole line segment between $(0.8,0.2,0)$ and $(0.2,0.8,0)$,
which is closed under convex combination.

To sum up, the \textquotedblleft statistical mixture of preparation
procedures\textquotedblright{} is a misleading conception. It actually
changes the notion of the system in question from the original $\mathscr{S}$
to a composed system ($\mathscr{D}$+$\mathscr{S}$). And even in
this sense it involves unjustified a priori assumptions, without regard
to the actual physical properties of the preparation device.

\subsection*{Appendix~3}

\noindent In the GPT literature, in order to achieve the dynamics
be linear, the following three assumptions are made (Hardy~2008,
Appendix~1). In our own notations:
\begin{lyxlist}{00.00.0000}
\item [{(a)}] The time-evolution map $F_{t}$ preserves convex combination.
\item [{(b)}] The set of states contains the null vector.
\item [{(c)}] The time evolution maps the null vector into itself.
\end{lyxlist}
Notice that assumptions (b) and (c) directly contradict our assumed
empirical facts (\ref{eq:osszeg=00003D1})--(\ref{eq:1d}) in \textbf{(E2)};\textbf{
}$\varphi\left(M,S\right)$ does not contain the null vector. It is
essential to clarify this contradiction. First of all, it should be
noted that these assumptions are based on conceptually untenable claims,
for example that the null vector is that state of the system ``when
the system is not present'' (Hardy~2008, p.~5). But how do we imagine
a measurement on a physical system that is not present? Fortunately,
however, assumptions (b) and (c) are not necessary for the time-evolution
$F_{t}$ be linear. Indeed, assumption (a) implies in itself that
$F_{t}$ is an $\mathbb{R}^{M+\left|S\right|}\rightarrow\mathbb{R}^{M+\left|S\right|}$
affine transformation restricted to the convex set of states (see
Meyer and Kay~1973, Theorem~4). $F_{t}$ is linear if and only if
it satisfies (a) and
\begin{lyxlist}{00.00.0000}
\item [{(d)}] its affine extension on $\mathbb{R}^{M+\left|S\right|}$
preserves the null vector.
\end{lyxlist}
And condition (d) does not require the null vector to be in the set
of states.

In our view, of course, the fulfillment of condition (d) is a matter
of further empirical information, beyond the stipulated empirical
facts \textbf{(E1)}--\textbf{(E4)}. In any case, let us come to hypothesis
(a) itself. How is it that assumption (a) does not hold even in an
example as simple and physically completely plausible as the one with
the coin discussed in section~\ref{sec:Dynamics}? In our view, the
reason is that the usual argumentation (see Müller 2021, pp. 20-21)
in favor of assumption (a) is conceptually flawed, as it operates
with the same problematic notion of ``statistical mixture of preparation
procedures'' that we have already criticized in Appendix~2. The
argument goes as follows. Consider the following two procedures:
\begin{lyxlist}{00.00.0000}
\item [{(i)}] The preparation device prepares the state $\vec{Z}_{i}$
with probability $\lambda_{i}$. Take the ``statistical mixture of
these states $\sum_{i}\lambda_{i}\vec{Z}_{i}$''; then we let this
``mixed state'' evolve in time, resulting in some final state $\vec{Z}'$.
\item [{(ii)}] The preparation device prepares the state $\vec{Z}_{i}$
with probability $\lambda_{i}$. We let each $\vec{Z}_{i}$ evolve
in time into $F_{t}\vec{Z}_{i}$. Finally, we take the ``statistical
mixture of states $F_{t}\vec{Z}_{i}$ with the same $\lambda_{i}$-s,
$\vec{Z}''=\sum_{i}\lambda_{i}F_{t}\vec{Z}_{i}$
\end{lyxlist}
``Clearly, (i) and (ii) are different descriptions of one and the
same laboratory procedure; they must hence result in the exact same
statistics of any measurements that we may decide to perform in the
end, and therefore lead to the same final state{[}.{]}'' ( Müller
2021, p.~21) That is,
\begin{equation}
\vec{Z}'=\vec{Z}''=\sum_{i}\lambda_{i}F_{t}\vec{Z}_{i}\label{eq:ketfele-preparacio}
\end{equation}
So far, with some clarification below, we agree with the argument.
As, indeed, if the ``time-evolution'' of the ``mixed state $\sum_{i}\lambda_{i}\vec{Z}_{i}$''
makes sense at all, it probably means the same process as described
in (ii). Where we disagree is the next claim: ``But this implies
that'' (\emph{ibid.)}
\begin{equation}
F_{t}\left(\sum_{i}\lambda_{i}\vec{Z}_{i}\right)=\sum_{i}\lambda_{i}F_{t}\vec{Z}_{i}\label{eq:preserves}
\end{equation}
Because $F_{t}$ has nothing to do with a ``mixed state.'' As we
pointed out in Appendix~2, a ``mixed state $\sum_{i}\lambda_{i}\vec{Z}_{i}$''
is not the state of the system, conceptually, but the state of the
composed system consisting of the preparation device ($\mathscr{D})$
and the original system ($\mathscr{S}$)---independently of whether
or not the vector $\sum_{i}\lambda_{i}\vec{Z}_{i}$ is among the possible
states of the original system. Applying the more precise notations
we introduced in Appendix~2, (\ref{eq:ketfele-preparacio}) reads
as follows:
\begin{equation}
^{(\mathscr{D}+\mathscr{S})}\vec{Z}'={}^{(\mathscr{D}+\mathscr{S})}\vec{Z}''=\sum_{i}\lambda_{i}\,{}^{(\mathscr{D}+\mathscr{S})}\left(F_{t}\vec{Z}_{i}\right)\label{eq:ketfele-preparacio-pontosan}
\end{equation}
$F_{t}:\varPhi\rightarrow\varPhi$ is not the time-evolution map of
``mixed states,'' states of the joint system ($\mathscr{D}$+$\mathscr{S}$).
The time-evolution of the ($\mathscr{D}$+$\mathscr{S}$) system is
something else; a map $T_{t}:{}^{(\mathscr{D}+\mathscr{S})}\varPhi\rightarrow{}^{(\mathscr{D}+\mathscr{S})}\varPhi$
, which is perhaps correctly claimed to be equivalent to the procedure
described in point (ii). So, (\ref{eq:ketfele-preparacio}), more
precisely (\ref{eq:ketfele-preparacio-pontosan}), means, so to say,
by the definition of $T_{t}$, that 
\begin{equation}
T_{t}\left(\sum_{i}\lambda_{i}\,{}^{(\mathscr{D}+\mathscr{S})}\vec{Z}_{i}\right)=\sum_{i}\lambda_{i}\,{}^{(\mathscr{D}+\mathscr{S})}\left(F_{t}\vec{Z}_{i}\right)\label{eq:T-definition}
\end{equation}
instead of (\ref{eq:preserves}).

For instance, in our Biased Coin example, consider the following states:
\begin{lyxlist}{00.00.0000}
\item [{$^{(\mathscr{D}+\mathscr{S})}\vec{Z}_{1}=(0.7,0.3,0)\in{}^{(\mathscr{D}+\mathscr{S})}\varPhi$:}] the
state of ($\mathscr{D}$+$\mathscr{S}$) in which the preparation
device $\mathscr{D}$ has propensity $1$ to set the coin into state
$\vec{Z}_{1}=(0.7,0.3,0)\in\varPhi$
\item [{$^{(\mathscr{D}+\mathscr{S})}\vec{Z}_{2}=(0.3,0.7,0)\in{}^{(\mathscr{D}+\mathscr{S})}\varPhi$:}] the
state of ($\mathscr{D}$+$\mathscr{S}$) in which the preparation
device $\mathscr{D}$ has propensity $1$ to set the coin into state
$\vec{Z}_{2}=(0.3,0.7,0)\in\varPhi$
\item [{$^{(\mathscr{D}+\mathscr{S})}\vec{Z}_{mix}=(0.5,0.5,0)\in{}^{(\mathscr{D}+\mathscr{S})}\varPhi$:}] the
state of ($\mathscr{D}$+$\mathscr{S}$) in which the preparation
device $\mathscr{D}$ has propensity $\frac{1}{2}$ to to set the
coin into state $\vec{Z}_{1}=(0.7,0.3,0)\in\varPhi$ and propensity
$\frac{1}{2}$ to set the coin into state $\vec{Z}_{2}=(0.3,0.7,0)\in\varPhi$
\end{lyxlist}
Obviously, $^{(\mathscr{D}+\mathscr{S})}\vec{Z}_{mix}=\frac{1}{2}{}^{(\mathscr{D}+\mathscr{S})}\vec{Z}_{1}+\frac{1}{2}{}^{(\mathscr{D}+\mathscr{S})}\vec{Z}_{2}$.
Now, adopting the assumption that (\ref{eq:T-definition}) holds,
and applying (\ref{eq:F_t-formula1})--(\ref{eq:F_t-formula2}) with,
say, $t=2$, 
\begin{eqnarray*}
T_{2}\left(^{(\mathscr{D}+\mathscr{S})}\vec{Z}_{mix}\right)=T_{2}\left(\frac{1}{2}{}^{(\mathscr{D}+\mathscr{S})}\vec{Z}_{1}+\frac{1}{2}{}^{(\mathscr{D}+\mathscr{S})}\vec{Z}_{2}\right)\,\,\,\,\,\,\,\,\,\,\,\,\,\,\,\,\,\,\,\,\,\,\,\,\,\,\,\,\,\,\,\,\\
=\frac{1}{2}F_{2}\vec{Z}_{1}+\frac{1}{2}F_{2}\vec{Z}_{2}=\frac{1}{2}F_{2}\left(0.7,0.3,0\right)+\frac{1}{2}F_{2}\left(0.3,0.7,0\right)\,\,\,\,\,\,\,\,\,\,\,\,\,\,\,\,\,\,\,\,\,\,\,\,\,\\
\approx\frac{1}{2}(0.44,0.56,0)+\frac{1}{2}(0.22,0.78,0)=\left(0.33,0.67,0\right)\,\,\,\,\,\,\,\,\,\,\,\,\,\,\,\,\,\,\,\,\,\,\,\,\,\,\,\,\,\,\,\,\,\,\,\,\,\,\,\,\,\,
\end{eqnarray*}
At the same time, 
\[
F_{2}\left(\frac{1}{2}\vec{Z}_{1}+\frac{1}{2}\vec{Z}_{2}\right)=F_{2}\left(\frac{1}{2}\left(0.7,0.3,0\right)+\frac{1}{2}\left(0.3,0.7,0\right)\right)=F_{2}\left(0.5,0.5,0\right)\approx\left(0.27,0.73,0\right)
\]
Thus, as expected from the formulas (\ref{eq:F_t-formula1})--(\ref{eq:F_t-formula2}),
$F_{t}$ does not preserve convex combination, whether or not (\ref{eq:T-definition})
is satisfied.

\section*{Funding}

This work was supported by Hungarian National Research, Development
and Innovation Office (Grant No. K134275 ).

\section*{References}
\begin{lyxlist}{00.00.0000}
\item [{Abramsky,~S.~and~Heunen,~C.~(2016):}] Operational Theories
and Categorical Quantum Mechanics. In J. Chubb, A. Eskandarian, and
V. Harizanov (eds.), \emph{Logic and Algebraic Structures in Quantum
Computing, }Cambridge University Press, Cambridge.
\item [{Aerts,~D.~(2009):}] Operational Quantum Mechanics, Quantum Axiomatics
and Quantum Structures. In D. Greenberger, K. Hentschel and F. Wienert
(Eds.), \emph{Compendium of Quantum Physics Concepts. Experiments,
History and Philosophy,} Springer-Verlag, Berlin, Heidelberg.
\item [{Bana,~G.~and~Durt,~T.~(1997):}] Proof of Kolmogorovian Censorship,
\emph{Found. Phys.} \textbf{27}, 1355.
\item [{Barandes,~J.~(2023):}] The Stochastic-Quantum Correspondence,
https://philsci-archive.pitt.edu/22501.
\item [{Barnum,~H.,~Barrett,~J.,~Leifer,~M.,~and~Wilce,~A.~(2007):}] Generalized
No-Broadcasting Theorem, \emph{Phys. Rev. Let.} \textbf{99}, 240501.
\item [{Barnum,~H.,~Barrett,~J.,~Leifer,~M.,~and~Wilce,~A.~(2008):}] Teleportation
in General Probabilistic Theories, arXiv:0805.3553 {[}quant-ph{]}
\item [{Bell,~J.~S.~(1987):}] \emph{Speakable and unspeakable in quantum
mechanics}, Cambridge University Press, Cambridge.
\item [{Busch,~P.,~Grabowski,~M.,~and~Lahti,~P.~J.~(1995):}] \emph{Operational
Quantum Physics}, Springer-Verlag, Berlin, Heidelberg.
\item [{Cassinelli~G.~and~Lahti,~P.~(2017):}] Quantum mechanics: why
complex Hilbert space? \emph{Phil. Trans. R. Soc. A} \textbf{375},
20160393.
\item [{Cirel'son,~B.~S.~(1980):}] Quantum Generalizations Of Bell'S
Inequality, \emph{Letters in Mathematical Physics} \textbf{4}, 93--100.
\item [{Davies,\,E.~B.~(1976):}] \emph{Quantum Theory of Open Systems},
Academic Press, London.
\item [{Foulis,~D.~J.~and~Randall,~C.~H.~(1974):}] Empirical Logic
and Quantum Mechanics, \emph{Synthese}\textbf{~29}, 81--111.
\item [{Fröhlich~J.~and~Pizzo~A.~(2022):}] The Time-Evolution of States
in Quantum Mechanics according to the ETH-Approach, \emph{Commun.
Math. Phys.} \textbf{389}, 1673--1715
\item [{Gudder,~S.~(1988):}] Two-Hole Experiment (Ch.~2. Sec.~5), \emph{Quantum
probability}, Academic Press, Boston.
\item [{Griffiths,~R.~B.~(2013):}] Hilbert space quantum mechanics is
noncontextual, \emph{Studies in History and Philosophy of Modern Physics}
\textbf{44}, 174--181.
\item [{Hardy,~L.~(2008):}] Quantum Theory From Five Reasonable Axioms,
http://arxiv.org/abs/quant-ph/0101012v4.
\item [{Henk~M.,~Richter-Gebert~J.,~and~Ziegler~G.~M.~(2004):}] Basic
Properties Of Convex Polytopes, in \emph{Handbook of Discrete and
Computational Geometry}, \emph{Second Edition}, J. E. Goodman and
J. O\textquoteright Rourke (eds.), Chapman \& Hall/CRC, Boca Raton,
Fla.
\item [{Hofer-Szabó, G., Rédei, M., and Szabó, L.E. (2013):}] \emph{The
Principle of the Common Cause}, Cambridge University Press, Cambridge.
\item [{Holevo,~A.~(2011):}] \emph{Probabilistic and statistical aspects
of quantum theory}, Edizioni della Normale, Springer, Pisa.
\item [{Jauch,~I.~M.~and~~Piron,~C.~(1963):}] Can Hidden Variables
be Excluded in Quantum Mechanics?, \emph{Helv. Phys. Acta.} \textbf{36},
827--837.
\item [{Leggett,~A.~J.~(1998):}] Macroscopic Realism: What Is It, and
What Do We Know about It from Experiment?, in: \emph{Quantum Measurement:
Beyond Paradox}, University of Minnesota Press, Minneapolis.
\item [{Ludwig,~G.~(1970):}] \emph{Deutung des Begriffs \textquotedblright physikalische
Theorie\textquotedblright{} und axiomatische Grundlagung der Hilbertraumstruktur
der Quantenmechanik durch Hauptsätze des Messens}, Springer-Verlag,
Berlin, Heidelberg.
\item [{Malament,~D.~B.~(1992):}] Critical notice: Itamar Pitowsky's
`Quantum Probability -- Quantum Logic\emph{', Philosophy of Science}
\textbf{59}, 300--320.
\item [{Meyer,~W.~and~Kay,~D.~C.~(1973):}] A convexity structure
admits but one real linearization of dimension greater than one, \emph{J.
London Math. Soc.} \textbf{2}, 124--130.
\item [{Müller,~M.~P.~(2021):}] Probabilistic theories and reconstructions
of quantum theory, \emph{SciPost Phys. Lect. Notes} 28 (doi: 10.21468/SciPostPhysLectNotes.28).
\item [{Pitowsky,~I.~(1989):}] \emph{Quantum Probability -- Quantum
Logic}, Springer-Verlag, Berlin, Heidelberg.
\item [{Park,~J.~L.~and~Margenau,~H.~(1968):}] Simultaneous Measurability
in Quantum Theory, \emph{Int. J. Theoretical Physics} \textbf{1},
211--283.
\item [{Park,~J.~L.~(1970):}] The Concept of Transition in Quantum Mechanics,
\emph{Foundations of Physics} \textbf{1}, 23--33.
\item [{Park,~J.~L.~and~Margenau,~H.~(1971):}] The Logic of Noncommutability
of Quantum-Mechanical Operators--and Its Empirical Consequences,
in: \emph{Perspectives in Quantum Theory -- Essays in Honor of Alfred
Landé}, W. Yourgrau and A. van der Merwe (eds.), The MIT Press, Cambridge,
Massachusetts.
\item [{Popescu~S.~and~Rohrlich,~D.~(1994):}] Quantum Nonlocality
as an Axiom, \emph{Foundations of Physics} \textbf{24}, 379--385.
\item [{Popper,~K.~R.~(1967):}] Quantum mechanics without \textquoteleft the
observer\textquoteright , in: \emph{Quantum theory and reality, }M.
Bunge (Ed.), Springer, Berlin, Heidelberg.
\item [{Rédei,~M.~(2010):}] Kolmogorovian Censorship Hypothesis For General
Quantum Probability Theories, \emph{Manuscrito} \textbf{33,} 365-380.
\item [{Reichenbach,~H.~(1944):}] Analysis of an interference experiment
(Ch.~1~$\mathsection$~7), in: \emph{Philosophical foundations
of quantum mechanics}, University of California Press, Los Angeles.
\item [{Ross,~D.~J.~(1974):}] Operator--observable correspondence,
\emph{Synthese} \textbf{29}, 373--403.
\item [{Schmid,~D.,~Spekkens,~R.W.,~and~Wolfe,~E.~(2018):}] All
the noncontextuality inequalities for arbitrary prepare-and-measure
experiments with respect to any fixed set of operational equivalences,\emph{
Phys. Rev. A} \textbf{97}, 062103.
\item [{Schroeck,~F.~E.,~Jr.~and~Foulis,~D.~J.~(1990):}] Stochastic
Quantum Mechanics Viewed from the Language of Manuals, \emph{Foundations
of Physics }\textbf{20}\emph{, }823--858.
\item [{Spekkens,~R.W.~(2005):}] Contextuality for preparations, transformations,
and unsharp measurements, \emph{Phys. Rev. A} \textbf{71}, 052108.
\item [{Strauss,~M.~(1936):}] The logic of complementarity and the foundation
of quantum theory, in: \emph{The Logico-Algebraic Approach To Quantum
Mechanics, Volume I, Historical Evolution}, C. A. Hooker (ed.), D.
Reidel Publishing Company, Dordrecht, 1975.
\item [{Szabó,~L.~E.~(1995):}] Is quantum mechanics compatible with
a deterministic universe? Two interpretations of quantum probabilities,
\emph{Foundations of Physics Letters} \textbf{8}, 421.
\item [{Szabó,~L.~E.~(1998):}] Quantum structures do not exist in reality,
\emph{Int. J. of Theor. Phys.} \textbf{37}, 449--456.
\item [{Szabó,~L.~E.~(2001):}] Critical reflections on quantum probability
theory, in: \emph{John von Neumann and the Foundations of Quantum
Physics}, M. Rédei, M. Stoeltzner (eds.), Kluwer Academic Publishers,
Dordrecht.
\item [{Szabó,~L.~E.~(2008):}] The Einstein-{}-Podolsky-{}-Rosen Argument
and the Bell Inequalities, \emph{Internet Encyclopedia of Philosophy.
}(https://iep.utm.edu/einstein-podolsky-rosen-argument-bell-inequalities).
\item [{Szabó,~L.~E.~(2020):}] Intrinsic, extrinsic, and the constitutive
a priori, \emph{Foundations of Physics} \textbf{50}, 555--567.
\end{lyxlist}

\end{document}